\documentclass[a4paper, UKenglish, cleveref, autoref, thm-restate]{lipics-v2021}

\hideLIPIcs  

\title{One-Parametric Presburger Arithmetic has Quantifier Elimination}

\titlerunning{One-Parametric Presburger Arithmetic has Quantifier Elimination}

\author{Alessio Mansutti}{IMDEA Software Institute, Madrid, Spain}{alessio.mansutti@imdea.org}{https://orcid.org/0000-0002-1104-7299}{Funded by the Madrid Regional Government (César Nombela grant 2023-T1/COM-29001), and by MCIN/AEI (grant PID2022-138072OB-I00).}

\author{Mikhail R.~Starchak}{St.~Petersburg State University, St.~Petersburg, Russia}{m.starchak@spbu.ru}{https://orcid.org/0000-0002-2288-9483}{Supported by the Russian Science Foundation, project 23-71-01041.}

\authorrunning{A.~Mansutti and M.~Starchak} 

\Copyright{Alessio Mansutti, Mikhail R.~Starchak} 



\category{} 

\relatedversion{\href{https://doi.org/10.4230/LIPIcs.MFCS.2025.58}{10.4230/LIPIcs.MFCS.2025.58}} 

\nolinenumbers 

\EventEditors{John Q. Open and Joan R. Access}
\EventNoEds{2}
\EventLongTitle{42nd Conference on Very Important Topics (CVIT 2016)}
\EventShortTitle{CVIT 2016}
\EventAcronym{CVIT}
\EventYear{2016}
\EventDate{December 24--27, 2016}
\EventLocation{Little Whinging, United Kingdom}
\EventLogo{}
\SeriesVolume{42}
\ArticleNo{23}

\usepackage[utf8]{inputenc}
\usepackage[T1]{fontenc}
\usepackage{xspace}

\usepackage{amsmath,amssymb,amsfonts,amsthm}
\usepackage{mathtools}
\usepackage{hyperref}
\usepackage{stmaryrd}
\usepackage{etoolbox}
\usepackage{float}
\usepackage{bm}
\usepackage{cite}
\usepackage{graphicx}
\usepackage{textcomp}
\usepackage{cleveref}
\usepackage{tabto}
\usepackage{thm-restate}
\usepackage{blkarray}
\usepackage{setspace}
\usepackage{multirow}

\usepackage{caption}

\usepackage[dvipsnames]{xcolor}

\definecolor{lip-yellow}{rgb}{0.99,0.78,0.07}

\def\BibTeX{{\rm B\kern-.05em{\sc i\kern-.025em b}\kern-.08em
    T\kern-.1667em\lower.7ex\hbox{E}\kern-.125emX}}

\usepackage{tikz}
\usetikzlibrary{
  decorations.pathreplacing,
  calc,
  arrows,
  positioning,
  backgrounds,
}
\usepackage{pgfplots}
\usepackage{svg}

\usepackage{colortbl}

\usepackage{algorithm}
\usepackage[noend]{algpseudocode}

\algnewcommand\algorithmiccontext{\textbf{Context:}}
\algnewcommand\Context{\item[\algorithmiccontext]}
\algnewcommand\algorithmicbranchoutput{\textbf{Branch Output:}}
\algnewcommand\BranchOutput{\item[\algorithmicbranchoutput]}
\algnewcommand\algorithmicndbranchoutput{\textbf{Output of each branch ($\beta$):}}
\algnewcommand\NDBranchOutput{\item[\algorithmicndbranchoutput]}
\algnewcommand\algorithmicglobaloutput{\textbf{Global Output:}}
\algnewcommand\GlobalOutput{\item[\algorithmicglobaloutput]}
\algnewcommand\algorithmicglobalspec{\textbf{Ensuring:}}
\algnewcommand\GlobalSpec{\item[\algorithmicglobalspec]}
\algnewcommand\algorithmicswitch{\textbf{switch}}
\algnewcommand\algorithmiccase{\textbf{case}}
\algnewcommand\algorithmicforeach{\textbf{foreach}}
\algnewcommand\algorithmicrep{\textbf{repeat}}
\algnewcommand\algorithmicnondet{\textbf{nondet}}
\algnewcommand\algorithmicor{\textbf{or}}
\algnewcommand\algorithmicassert{\textbf{assert}}
\algnewcommand\algorithmiclet{\textbf{let}}
\algnewcommand\Assert[1]{\State \algorithmicassert(#1)}
\algdef{SE}[SWITCH]{Switch}{EndSwitch}[1]{\algorithmicswitch\ #1\ \algorithmicdo}{\algorithmicend\ \algorithmicswitch}%
\algtext*{EndSwitch}%
\algdef{SE}[CASE]{Case}{EndCase}[1]{\algorithmiccase\ #1\ \textbf{:}}{\algorithmicend\ \algorithmiccase}%
\algtext*{EndCase}%
\algdef{S}[FOR]{ForEach}[1]{\algorithmicforeach\ #1\ \algorithmicdo}%
\algdef{S}[FOR]{Rep}[1]{\algorithmicrep\ #1\ \algorithmicdo}%
\algdef{SE}[NONDET]{Nondet}{EndNondet}{\algorithmicnondet\ }{\algorithmicend\ \algorithmicnondet}%
\algtext*{EndNondet}%
\algdef{SE}[Or]{Or}{EndOr}{\algorithmicor\ }{\algorithmicend\ \algorithmicor}%
\algtext*{EndOr}%
\algdef{SE}[LET]{Let}{EndLet}{\algorithmiclet\ }{\algorithmicend\ \algorithmiclet}%
\algtext*{EndLet}%

\AtBeginEnvironment{algorithm}{\linespread{1.05}}

\newcounter{stepnum}
\usepackage{todonotes}
\colorlet{dmitry}{green!50!black!80}
\colorlet{alessio}{blue!70!black!80}
\colorlet{hitarth}{red!70!blue!80}
\definecolor{mikhail}{HTML}{9F2D20}
\definecolor{guru}{HTML}{9F2D20}



\makeatletter
\@ifpackageloaded{stix}{%
}{%
  \DeclareFontEncoding{LS2}{}{\noaccents@}
  \DeclareFontSubstitution{LS2}{stix}{m}{n}
  \DeclareSymbolFont{stix@largesymbols}{LS2}{stixex}{m}{n}
  \SetSymbolFont{stix@largesymbols}{bold}{LS2}{stixex}{b}{n}
  \DeclareMathDelimiter{\lBrace}{\mathopen} {stix@largesymbols}{"E8}%
                                            {stix@largesymbols}{"0E}
  \DeclareMathDelimiter{\rBrace}{\mathclose}{stix@largesymbols}{"E9}%
                                            {stix@largesymbols}{"0F}
}
\makeatother

\newcommand{\newextmathcommand}[2]{%
    \newcommand{#1}{\ensuremath{#2}\xspace}
}

\makeatletter
\newcommand{\labeltext}[2]{%
  #1%
  \@bsphack%
  \csname phantomsection\endcsname 
  \def\@currentlabel{#1}{\label{#2}}%
  \@esphack%
}
\makeatother

\makeatletter
\newcommand{\customlabel}[2]{%
  \@bsphack%
  \csname phantomsection\endcsname 
  \def\@currentlabel{#1}{\label{#2}}%
  \@esphack%
}
\makeatother

\makeatletter
               {\list{}{\leftmargin=0pt
                        \labelwidth\z@ \itemindent-\leftmargin
                        }}%
               {\endlist}
\makeatother

\newextmathcommand{\N}{\mathbb{N}}
\newextmathcommand{\Np}{\Nat_+}
\newextmathcommand{\Z}{\mathbb{Z}}
\newextmathcommand{\Q}{\mathbb{Q}}
\newextmathcommand{\R}{\mathbb{R}}
\newextmathcommand{\PP}{\mathbb{P}}
\newextmathcommand{\X}{\mathbb{X}}

\newcommand{\sem}[1]{\ensuremath{\left\llbracket#1\right\rrbracket}\xspace}

\newextmathcommand{\ptime}{\textup{\textsc{P}}\xspace}
\newextmathcommand{\fptime}{\textup{\textsc{FP}}\xspace}
\newextmathcommand{\bpp}{\textup{\textsc{BPP}}\xspace}
\newextmathcommand{\np}{\textup{\textsc{NP}}\xspace}
\newextmathcommand{\conp}{\textup{\textsc{coNP}}\xspace}
\newextmathcommand{\fnp}{\textup{\textsc{FNP}}\xspace}
\newextmathcommand{\pspace}{\textup{\textsc{PSpace}}\xspace}
\newextmathcommand{\nexptime}{\textup{\textsc{NExpTime}}\xspace}
\newextmathcommand{\expspace}{\textup{\textsc{ExpSpace}}\xspace}
\newextmathcommand{\conexp}{\textup{\textsc{coNExp}}\xspace}
\newextmathcommand{\twonexptime}{\textup{\textsc{2NExpTime}}\xspace}
\newextmathcommand{\threeexptime}{\textup{\textsc{3ExpTime}}\xspace}
\newextmathcommand{\tower}{\textup{\textsc{Tower}}\xspace}
\newextmathcommand{\npo}{\textup{\textsc{NPO}}\xspace}
\newextmathcommand{\npocmp}{\textup{\textsc{NPO-cmp}}\xspace}
\newextmathcommand{\STA}{\textup{\textsc{STA}}\xspace}
\newextmathcommand{\factoring}{\textup{\textsc{factoring}}\xspace}

\let\temp\phi
\let\phi\varphi
\let\varphi\temp
\newextmathcommand{\totient}{\varphi}
\renewcommand{\vec}{\bm}

\newextmathcommand{\lcm}{{\rm lcm}}

\newcommand{\abs}[1]{\ensuremath{\left|#1\right|}\xspace}
\newcommand{\ceil}[1]{\ensuremath{\left\lceil#1\right\rceil}\xspace}
\newcommand{\floor}[1]{\ensuremath{\left\lfloor#1\right\rfloor}\xspace}
\newcommand{\sceil}[1]{\ensuremath{\bigl\lceil#1\bigr\rceil}\xspace}
\newcommand{\sfloor}[1]{\ensuremath{\bigl\lfloor#1\bigr\rfloor}\xspace}
\newcommand{\frpart}[1]{\ensuremath{\left\{#1\right\}}\xspace}

\newcommand{\intdiv}[2]{\ensuremath{\bigl\lfloor\frac{#1}{#2}\bigr\rfloor}\xspace}

\DeclareFontEncoding{LS2}{}{\noaccents@}
  \DeclareFontSubstitution{LS2}{stix}{m}{n}
  \DeclareSymbolFont{stix@largesymbols}{LS2}{stixex}{m}{n}
  \SetSymbolFont{stix@largesymbols}{bold}{LS2}{stixex}{b}{n}
  \DeclareMathDelimiter{\lBrace}{\mathopen} {stix@largesymbols}{"E8}%
                                            {stix@largesymbols}{"0E}
  \DeclareMathDelimiter{\rBrace}{\mathclose}{stix@largesymbols}{"E9}%
                                            {stix@largesymbols}{"0F}

\newcommand{\sub}[2]{\ensuremath{[#1\,/\,#2]}\xspace}
\newcommand{\abssub}[2]{\ensuremath{\lBrace #1\,/\,#2 \rBrace}\xspace}
\newcommand{\vigsub}[2]{\ensuremath{{[\mkern-1mu[#1\mathbin{/}#2]\mkern-1mu]}}\xspace}

\newcommand{\bitlength}[1]{\ensuremath{\langle{#1}\rangle}\xspace}

\newcommand{\onenorm}[1]{\ensuremath{\lVert{#1}\rVert_{1}}\xspace}
\newextmathcommand{\card}{\#}

\newextmathcommand{\V}{V_2}

\newextmathcommand{\fterms}{\textup{terms}}
\newextmathcommand{\fmod}{\textit{mod}}
\newextmathcommand{\fdivub}{\textit{div.ub}}
\newextmathcommand{\divides}{\mathrel{|}}
\newextmathcommand{\lst}{\textit{lst}}

\newextmathcommand{\lead}{\ell(t)}
\newextmathcommand{\prevlead}{p(t)}

\newextmathcommand{\ILP}{\textup{ILP}}
\newextmathcommand{\PA}{\textup{PrA}}
\newextmathcommand{\PPA}{\textup{PrA}[t]}

\definecolor{light-gray}{gray}{0.95}
\newcolumntype{g}{>{\columncolor{light-gray}}r}

\newextmathcommand{\varcoeff}{\textit{coeff}}

\newextmathcommand{\paratom}{\textit{atom}}
\newextmathcommand{\parvars}{\textit{vars}}
\newextmathcommand{\parfunc}{\textit{func}}
\newextmathcommand{\parconst}{\bitlength{\textit{const}}}

\newextmathcommand{\parterms}{\textit{terms}}
\newextmathcommand{\parnonshift}{\card\textit{noshifts}}
\newextmathcommand{\parmod}{\card\textit{mod}}
\newextmathcommand{\paroccmod}{\card\textit{occ.mod}}
\newextmathcommand{\parcoeff}{\bitlength{\varcoeff}}
\newextmathcommand{\parbound}{\bitlength{B}}
\newextmathcommand{\parcoeffnum}{\card\varcoeff}

\newcommand{\myif}{\textbf{if}\xspace}

\newcommand{\myelse}{\textbf{else}\xspace}
\newcommand{\mycontinue}{\textbf{continue}\xspace}
\newcommand{\myguess}{\textbf{guess}\xspace}

\newextmathcommand{\Master}{\textsc{$\PPA$-QE}}
\newextmathcommand{\BoundedQE}{\textsc{BoundedQE}}
\newextmathcommand{\BoundedElimDiv}{\textsc{ElimDiv}}
\newextmathcommand{\ElimBounded}{\textsc{ElimBounded}}

\newextmathcommand{\kl}{(k,\ell)}
\newextmathcommand{\leac}{LEAC\xspace}

\newextmathcommand{\tests}{\textit{TP}}
\newextmathcommand{\testsilep}{\textit{TP}_{\textup{ILEP}}}
\newextmathcommand{\disciplineilep}{\abssub{\cdot}{\cdot}_{\textup{ILEP}}}

\newextmathcommand{\constraints}{\mathcal{C}}
\newextmathcommand{\objectives}{\mathcal{F}}
\newextmathcommand{\objcons}{\mathcal{S}}

\newextmathcommand{\interiorpoints}{\textup{interior}}

\newextmathcommand{\goal}{\mathrm{goal}}
\newextmathcommand{\opt}{\mathrm{opt}}
\newextmathcommand{\short}{\mathrm{short}}
\newextmathcommand{\sol}{\mathrm{sol}}

\newextmathcommand{\lin}{\mathrm{lin}}
\newextmathcommand{\context}{\mathrm{ctx}}

\newextmathcommand{\odd}{\textit{odd}}

\newextmathcommand{\trunc}{T}
\newextmathcommand{\indic}{\mathbf{1}}

\begin{document}

\title{One-Parametric Presburger Arithmetic\\ has Quantifier Elimination}

\maketitle

\begin{abstract}
  We give a quantifier elimination procedure 
  for \emph{one-parametric Presburger arithmetic}, the extension of Presburger arithmetic 
  with the function $x \mapsto t \cdot x$, where 
  $t$ is a fixed free variable ranging over the integers. 
  This resolves an open problem proposed in \emph{[Bogart et al., Discrete Analysis, 2017]}.
  As conjectured in~\emph{[Goodrick, Arch.~Math.~Logic, 2018]},
  quantifier elimination is obtained for the extended structure featuring all integer division functions 
  $x \mapsto \sfloor{\frac{x}{f(t)}}$, one for each integer polynomial~$f$. 


  Our algorithm works by iteratively eliminating blocks of existential quantifiers. 
  The elimination of a block builds on two sub-procedures, both running in non-deterministic polynomial time. 
  The first one is an adaptation of a recently developed and efficient quantifier elimination procedure for Presburger arithmetic, 
  modified to handle formulae with coefficients over the ring~$\Z[t]$ of univariate polynomials.
  The second is reminiscent of the so-called ``base $t$ division method'' used by Bogart \emph{et al}. 
  As a result, we deduce that the satisfiability problem for the existential fragment of 
  one-parametric Presburger arithmetic (which encompasses a broad class of non-linear integer programs) is in~\np, 
  and that the smallest solution to a satisfiable formula 
  in this fragment is of polynomial bit size.
\end{abstract}

\section{Introduction}
\label{sec:intro}

The first-order theory of the integers~$\Z$ with addition and order, which is also known as \emph{Presburger arithmetic} (\PA)
or linear integer arithmetic, has been intensively studied during almost a century~\cite{Pre29}.
It is a textbook fact that Presburger arithmetic admits quantifier elimination 
when the structure $\langle\Z;\,0,\,1,\,+,\,\leq\rangle$ is extended with the predicates $(d \divides (\cdot))_{d \in \Z}$ 
for divisibilities by fixed integers $d$: 
in the theory of this extended structure,
for every quantifier-free formula $\phi(x,\vec y)$
there is a quantifier-free formula $\psi(\vec y)$ that is equivalent to $\exists x : \phi(x,\vec y)$. The construction of $\psi$ is effective, which implies the decidability of Presburger arithmetic.
The algorithm to decide \PA is the canonical example for the notion of \emph{quantifier elimination procedure}.
The computational complexity of the many variants of this procedure has a long 
history, beginning with Oppen's proof~\cite{Oppen78} that Cooper's procedure~\cite{Cooper72} runs in triply exponential time, 
and followed by refinements from Reddy and Loveland~\cite{Reddy78}, and later Weispfenning~\cite{Weispfenning90}, 
which enable handling formulae with fixed quantifier alternations in doubly exponential time.
Recent research on quantifier elimination aims at narrowing the complexity gap for the existential fragment 
(only last year it was discovered that quantifier elimination can be performed in exponential time in this case~\cite{ChistikovMS24,HaaseKMMZ24}) 
and on extending the procedure to handle additional predicates and functions~\cite{Starchak21,BenediktCM23,ChistikovMS24,KarimovLNO025},
or other forms of quantification~\cite{ChistikovHM22,Habermehl2023,BergstrasserGLZ24}.
The reader can find an extensive bibliography on Presburger arithmetic, and quantifier elimination,
in the survey papers by Haase~\cite{Haase18} and Chistikov~\cite{Chistikov24}. 

This paper addresses the open problem raised in the papers~\cite{BogartGW17} and \cite{Goodrick18} 
regarding the existence of a quantifier elimination procedure for the theory
Th${\langle\Z;\,0,\,1,\,+,\,x \mapsto t \cdot x,\,\leq\rangle}$, known 
as \emph{one-parametric Presburger arithmetic}~(\PPA). 
In this theory, the structure of Presburger arithmetic is extended with the function~$x \mapsto t \cdot x$ for multiplication by a single parameter $t$ ranging over $\Z$. 
Every \PPA formula $\phi(\vec x)$ defines a \emph{parametric Presburger family} 
${\mathbb{S}(\phi)\coloneqq\big\{\!\sem{\phi}_k : k \in \Z \big\}}$\!, where $\sem{\phi}_k$ is the
set of (integer) solutions of the Presburger arithmetic formula obtained from $\phi$ by replacing 
the parameter~$t$ with the integer $k$. 

\begin{example}
    \label{intro:example-2}
    Consider the statement ``for every two successive positive integers $t$ and $t+1$, and for all integers $a$ and $b$, 
    there is an integer $x$ in the interval $[0,t(t+1)-1]$ that is congruent to $a$ modulo $t$, and to $b$ modulo $t+1$''. 
    The truth of this sentence follows from the Chinese remainder theorem together with the fact that successive positive integers are always coprime.
    We can encode this statement in~\PPA with the formula~$\forall a \forall b : \chi(a,b)$, where
    \[ 
        \chi \ \coloneqq \ t \geq 1 \implies \exists x : \ 0 \leq x \land x \leq t^2 + t - 1 
        \land (\exists y : x - a = t \cdot y) \land (\exists z : x - b = (t+1) \cdot z).
    \]
    From the validity of the statement, we find $\sem{\chi}_k = \Z^2$ for every $k \in \Z$. 
    That is to say, for every instantiation of~$t$, both~$\phi$ and~$\chi$ are tautologies of~\PA.
    \qed
\end{example}
There are many natural decision problems regarding~\PPA (all taking a formula~$\phi$ as input):

\begin{itemize}
    \item \emph{satisfiability}: Is $\phi$ satisfiable for an instantiation of the parameter 
    (i.e.,~$\mathbb{S}(\phi) \neq \{\emptyset\}$)\,?
    \item \emph{universality}: Is $\phi$ satisfiable for all instantiations of the parameter 
    (i.e.,~$\emptyset \not\in \mathbb{S}(\phi)$)\,?
    \item \emph{finiteness}: Is $\phi$ satisfiable for only finitely many instantiations of the parameter~$t$\,?
\end{itemize}

In~\cite{BogartGW17}, Bogart, Goodrick, and Woods consider a search problem that generalises 
all the problems above: they show how to compute, from 
an input formula $\phi$, a closed expression for the function $f(k) \coloneqq \card{\sem{\phi}_k}$, 
where $\card S$ stands for the cardinality of a set~$S$. By relying on properties of this function, 
one can solve satisfiability, universality, and finiteness.  
In their proof, the first ingredient is given by Goodrick's bounded quantifier elimination procedure~\cite{Goodrick18}.
In contrast to the quantifier elimination procedures for \PA, in this procedure every quantified variable $x$ is not completely eliminated from the formula~$\phi$, but acquires instead a bound $0\leq x\leq f(t)$, for 
some univariate polynomial $f(t)$. 
An example of this is given by the variable $x$ in~\Cref{intro:example-2}, 
which is bounded in $[0,t(t+1)-1]$.
Closely related bounded quantifier elimination procedures were also developed in~\cite{Lasaruk09,Weispfenning97}. 
The second ingredient of the construction is given by a method developed by Chen, Li, and Sam for the study of parametric polytopes~\cite{ChenLS12}, 
and dubbed ``base $t$ division method'' in~\cite{BogartGW17}.
This method produces a quantifier-free \PPA formula $\psi$ satisfying $\card{\sem{\psi}_k} = \card{\sem{\phi}_k}$ for every $k \in \Z$.

The combination of the two ingredients has a drawback: the equivalence of the initial formula $\phi$ 
with the quantifier-free formula~$\psi$ over $\Z$ is not preserved. 
In~\cite[p.~13]{BogartGW17}, the authors ask whether this issue can be fixed: ``\emph{one might try to show that any formula} [\,of \PPA] 
\emph{is logically equivalent to a quantifier-free formula in a slightly larger language with additional ``well-behaved'' function and 
relation symbols} [\,\dots\,] \emph{But we already know that quantifier elimination in the original language} [\,of the structure $\langle\Z;\,0,\,1,\,+,\,x \mapsto t \cdot x,\,\leq\rangle$\,] 
\emph{is impossible, and finding a reasonable language for quantifier elimination seems difficult''.} 
A candidate for the extended structure was suggested by Goodrick in~\cite[Conjecture 2.6]{Goodrick18}: 
\PPA must be extended with all \emph{integer division functions} $x \mapsto \sfloor{\frac{x}{|f(t)|}}$, one for each integer polynomial $f(t)$ (these functions are assumed to occur in a formula only under the proviso that $f(t) \neq 0$). This is a rather tight conjecture, as all added functions are 
trivially definable in~\PPA: the equality $y = \sfloor{\frac{x}{\abs{f(t)}}}$ holds if and only if
$f(t) \neq 0 \,\land\, \abs{f(t)} \cdot y \leq x \,\land\, x < \abs{f(t)} \cdot y + \abs{f(t)}$.

We give a positive answer to Goodrick's conjecture:
\begin{theorem}\label{theorem:main}
    One-parametric Presburger arithmetic
    admits effective quantifier elimination in the extended structure $\langle\Z;\,0,\,1,\,+,\, x \mapsto t \cdot x,\, 
    x \mapsto \sfloor{\frac{x}{\abs{f(t)}}},\,\leq\rangle$.
\end{theorem}

Above, the adjective ``effective'' reflects the fact that there is a quantifier elimination procedure for constructing, 
given an input formula $\phi$, an equivalent quantifier-free formula~$\psi$.
The main contribution towards the proof of~\Cref{theorem:main} 
is a procedure for removing bounded quantifiers while preserving formula equivalence; hence obtaining 
$\sem{\psi}_k = \sem{\phi}_k$ for all $k \in \Z$, 
instead of the weaker $\card{\sem{\psi}_k} = \card{\sem{\phi}_k}$ obtained with the ``base $t$ division method'' from~\cite{BogartGW17}.

As mentioned above, an active area of research in quantifier elimination focuses on existential Presburger 
arithmetic ($\exists$\PA) and its extensions. This interest is motivated, on the one hand, by the goal of 
improving both the performance and expressiveness of SMT solvers, mostly targeting existential theories~\cite{BarrettT18,FrohnG24}. On the other hand, recent work has revisited the computational complexity of quantifier elimination procedures. For many years, these procedures were regarded as inefficient 
when applied to $\exists$\PA. Notably, Weispfenning's classic approach~\cite{Weispfenning90} yields 
only a \nexptime upper bound for satisfiability, despite fundamental results from integer 
programming~\cite{BoroshT76,vonzurGathenS78} establishing that $\exists$\PA is in~\np. 
It was not until 2024 that two independent works~\cite{ChistikovMS24,HaaseKMMZ24} provided quantifier elimination 
procedures with matching~\np upper bounds. The approach in~\cite{HaaseKMMZ24} builds on the geometric ideas 
from 
\cite{vonzurGathenS78}, while~\cite{ChistikovMS24} adapts Bareiss' fraction-free Gaussian elimination procedure~\cite{Bareiss68} 
to $\exists$\PA. Starting from the fact that Bareiss' algorithm works in any integral domain, 
we show that the second approach extends naturally to~$\PPA$. 
By analysing the runtime of our procedure, we derive: 

\begin{theorem}\label{theorem:complexity}
    For the class of all existential formulae of \PPA, the following holds:  
    \begin{center}
        \begin{tabularx}{0.9\textwidth}{|>{\centering\arraybackslash}X|>{\centering\arraybackslash}X|>{\centering\arraybackslash}X|} 
            \hline
            \cellcolor{lipicsLightGray} Satisfiability      & \cellcolor{lipicsLightGray} Universality & \cellcolor{lipicsLightGray} Finiteness \\ 
            \np-complete  & \conexp-complete           & \conp-complete      \\
            \hline
        \end{tabularx}
      \end{center}
      \vspace{0pt}
\end{theorem}
Our result on the satisfiability problem generalises the feasibility in \np of non-linear integer programs 
$A\cdot \vec x \,\leq\, \vec b(t)$, where $\vec b(t)$ is a vector of quotients of integer polynomials in $t$, 
proved by Gurari and Ibarra~\cite{GurariI79}. 
We remark that both problems are \np-hard in fixed dimension: 
solvability of systems $x \geq 0 \land a \cdot t^2 + b \cdot x = c$ is a well-known 
\np-complete problem~\cite{MandersA78}.

\subparagraph*{Future work.} 
This paper does not provide a complexity analysis for \emph{full}~\PPA. 
A back-of-the-envelope calculation of the runtime of the procedures in~\cite{BogartGW17} and~\cite{Goodrick18} 
suggests that the satisfiability problem for~\PPA is in elementary time (potentially in~\threeexptime).
However, these procedures do not yield an~\np upper bound for the existential fragment. 
In contrast, the procedure we introduce shows $\exists$\PPA in~\np, but it may in principle run in non-elementary 
time on arbitrary formulae. 
Unifying these procedures into a single ``optimal'' one seems possible and will be addressed in a 
forthcoming extended version of this paper. This would also provide an extension to the~\threeexptime 
quantifier-elimination procedure for \emph{almost linear arithmetic} proposed by Weispfenning in~\cite{Weispfenning90}.


Most of the literature on~\PPA focuses on computing the function~${f(k) \coloneqq \card{\sem{\phi}}_k}$, as introduced in~\cite{BogartGW17}. Not much is known regarding the complexity of computing this function. To our knowledge, the most significant result in this direction is the one in~\cite{BogartGNW19}, where Bogart, Goodrick, Nguyen, and Woods show that $f(k)$ can be computed in polynomial time (in the bit-length of $k$ given as part of the input) 
whenever $\phi$ is a \emph{fixed} \PPA-formula. We believe \Cref{theorem:complexity} 
to be a good starting point for further research in this direction.

While our work shows that the feasibility problem for integer programs in which 
a single variable occurs non-linearly is in~\np, the paper does not discuss
related optimisation problems of minimisation/maximisation. From our quantifier elimination procedure, we can show
that if there are optimal solutions to a linear polynomial with coefficients in $\Z[t]$
subject to a formula in $\exists$\PPA, then one is of polynomial bit size. 
Generalising this result to other non-convex objectives is an interesting avenue for future research.
\section{Preliminaries}
\label{sec:preliminaries}

We write $\mathbb{N}$ for the non-negative integers, 
and $\Z[t]$ for the set of univariate polynomials ${f(t) = \sum_{i=0}^d a_i \cdot t^i}$, 
where the coefficients $a_1,\dots,a_d$ and the constant $a_0$ are over the integers~$\Z$.
The \emph{height}~$h(f)$, \emph{degree} $\deg(f)$ and \emph{bit size} $\bitlength{f}$
of~$f$ are defined as 
${h(f) \coloneqq \max\{\abs{a_i} : i \in [0,d]\}}$, 
$\deg(f) \coloneqq \max\{0, i \in [0,d] : a_i \neq 0 \}$, 
and $\bitlength{f} \coloneqq (\deg(f) + 1) \cdot (\ceil{\log_2(h(f)+1)} + 1)$, 
respectively. For example, $f(t) = 2 \cdot t^2 - 3$ has degree $2$, height $3$ and bit size $9$.
Vectors of variables are denoted by $\vec x, \vec y, \vec z$, \emph{etc}.;
we write $\floor{\cdot}$ for the floor function.

\subparagraph*{On extending the structure of~\PPA.}
As discussed in~\Cref{sec:intro}, the paper concerns the extension of the first-order theory of ${\langle\Z;\,0,\,1,\,+,\, x \mapsto t \cdot x,\, \leq \rangle}$ by all \emph{integer division functions} $x \mapsto \sfloor{\frac{x}{\abs{\,f\,}}}$, where $f \in \Z[t]$.
However, in the context of our quantifier elimination procedure, it is more natural to work within the (equivalent) first-order theory of the structure:
\[
    \textstyle{\big\langle\Z;\,0,\,1,\,+,\, x \mapsto t \cdot x,\, \big\{ x \mapsto \floor{\frac{x}{t^d}} \big\}_{d \in \N},\, \big\{\,x \mapsto (x \bmod f)\, \big\}_{f \in \Z[t]},\, \big\{\, f \divides x \, \big\}_{f \in \Z[t]},\,=,\, \leq \big\rangle}
\]
where:
\begin{itemize}
    \item The integer division $x \mapsto \floor{\frac{x}{t^d}}$ is only defined for $t \neq 0$, with the obvious interpretation.
    \item The \emph{integer remainder function} $x \mapsto (x \bmod f(t))$ is defined following the equivalence
    \[ 
        \textstyle{(x \bmod f(t) = y)
        \iff 
        (f(t) = 0 \land y = x)
        \lor 
        (f(t) \neq 0 \land 
        y = x - \abs{f(t)} \cdot \sfloor{\frac{x}{\abs{f(t)}}})}.
    \]
    We remark that whenever $f(t) \neq 0$ the result of $(x \bmod f(t))$ belongs to~$[0,\abs{f(t)}-1]$.

    (Also note that the absolute value function $\abs{.}$ is easily definable in Presburger arithmetic.)
    
    \item The \emph{divisibility relation} $f(t) \divides x$ is a unary relation, and is defined following the equivalence
    \[ 
        (f(t) \divides x) \iff (f(t) = 0 \land x = 0) \lor (f(t) \neq 0 \land (x \bmod f(t) = 0)).
    \]
    We remark that the divisibility relations and integer remainder functions are defined 
    to satisfy the equivalence $f(t) \divides x \iff f(t) \divides (x \bmod f(t))$ 
    also when $f(t)$ evaluates to $0$.
\end{itemize}
For simplicity, \emph{we still denote this first-order theory with~\PPA}. 
Observe that we have ultimately defined $(f(t) \divides \cdot)$ and $x \mapsto (x \bmod f(t))$ 
in terms of $x \mapsto \sfloor{\frac{x}{\abs{\,f\,}}}$. 
As a result, every formula in this first-order theory can be translated into a formula from the theory in~\Cref{theorem:main}. 
This translation can be performed in polynomial time by introducing new existential quantifiers, 
or in exponential time without adding quantifiers (the blow-up is only due to the disjunctions in the definitions of $(f(t) \divides \cdot)$ and $x \mapsto (x \bmod f(t))$). 

The \emph{terms} of \PPA are built from the constants $0,1$, integer variables, and the functions of the structure.
Without loss of generality, we restrict ourselves to (finite) terms of the form 
\begin{equation}
    \label{eq:normal-term}
    \textstyle\tau \ \coloneqq \ f_0(t) \,+\, \sum\nolimits_{i=1}^n f_i(t) \cdot x_i \, +\,  
    \sum\nolimits_{i=n+1}^m f_i(t) \cdot \floor{\frac{\tau_i}{t^{d_i}}} \, +\, \sum\nolimits_{i=m+1}^k f_i(t) \cdot (\tau_i \bmod g_i(t)),
    \hspace{0.5cm}
\end{equation}
where 
all $f_i$ and $g_i$ belong to~$\Z[t]$, all $d_i$ belong to $\N$, and each $\tau_i$ is another term of this form.
The term $\tau$ is said to be \textit{linear}, if $f_i(t)=0$ for $i\in[n+1,k]$ (i.e., it does not contain 
integer division, nor integer remainder functions), and  \textit{non-shifted} whenever $f_0(t)=0$.
Above, the terms $\sfloor{\frac{\tau_i}{t^{d_i}}}$ and 
$(\tau_i \bmod g_i(t))$ are \emph{linear occurrences} of the integer division functions and integer remainder 
functions, and are said to \emph{occur linearly} in $\tau$. 
When every $f_i$ is an integer (a degree $0$ polynomial),
we define $1$-norm of $\tau$ as $\onenorm{\tau} \coloneqq \sum_{i=0}^k \abs{f_i}$.

Moving to the atomic formulae of the theory, it is easy to see that equalities, inequalities and divisibility relations can be rewritten (in polynomial time) 
to be of the form $\tau = 0$, $\tau \leq 0$ and $f(t) \divides \tau$, respectively, where $\tau$ 
is a term of the form given in~\Cref{eq:normal-term}. 
Syntactically, we will only work with atomic formulae of these forms, which we call \PPA \emph{constraints}. 
However, for readability, we will still sometimes write (in)equalities featuring non-zero terms on both sides (i.e.,~$\tau_1 \leq \tau_2$), and strict inequalities $\tau_1 < \tau_2$ and $\tau_2 > \tau_1$; both are short for $\tau_1 - \tau_2 + 1 \leq 0$.
A \PPA constraint is said to be \emph{linear} if the term $\tau$ featured in it is linear.

We restrict ourselves to formulae in prenex normal form $\exists \vec x_1 \forall \vec x_2 \dots \exists \vec x_n : \phi$ where $\phi$ is a \emph{positive} Boolean combination of \PPA constraints; 
that is, the only Boolean connectives featured in $\phi$ are conjunctions~$\land$ and disjunctions~$\lor$. 
This restriction is without loss of generality, as De Morgan's laws allow to push all negations at the level of literals, 
which can then be removed with the equivalences $\lnot(\tau = 0) \iff {{\tau < 0} \lor {\tau > 0}}$, \
${\lnot(\tau \leq 0) \iff {\tau > 0}}$ and ${\lnot (f(t) \divides \tau) \iff (f(t) = 0 \land (\tau < 0 \lor \tau > 0)) \lor (\tau \bmod f(t) > 0)}$.

For two terms $\tau_1$ and $\tau_2$, we write $\sub{\tau_2}{\tau_1}$ for 
a term \emph{substitution}. We see these substitutions as functions from terms to terms or from formulae to formulae: $\tau\sub{\tau_2}{\tau_1}$ is the term obtained by replacing, in the term $\tau$, every occurrence of $\tau_1$ with $\tau_2$. Analogously, $\phi\sub{\tau_2}{\tau_1}$ is the formula obtained from $\phi$ by replacing the term $\tau$ with $\tau\sub{\tau_2}{\tau_1}$
in every atomic formula $\tau = 0$, $\tau \leq 0$, or $f(t) \divides \tau$. 
In~\Cref{sec:gaussian-elimination} we will also need a stronger notion of substitution, called \emph{vigorous substitution} in~\cite{ChistikovMS24}; we defer its definition to that section.

\section{Outline of the quantifier elimination procedure}
\label{sec:FO-procedure}

In this section, we provide a high-level overview of our quantifier elimination procedure, 
highlighting the interactions among its various components. A detailed analysis of the two main components will be provided in the subsequent sections of the paper.

Let us consider a formula $\psi(\vec x_0) \coloneqq \exists \vec x_1 \forall \vec x_2 \dots \exists \vec x_n : \phi(\vec x_0,\dots,\vec x_n)$,
where $\phi$ is a positive Boolean combination of \PPA constraints. 
Our quantifier elimination procedure will work under the assumption that the parameter $t$ is greater than or equal to~$2$.
This simplifying assumption is without loss of generality, as the general problem is then solved as follows:
\begin{itemize}
    \item For every $k \in \{-1,0,1\}$, call a quantifier elimination procedure for Presburger arithmetic (e.g., the one in~\cite{Weispfenning90}) on the formula $\psi\sub{k}{t}$, obtaining $\psi_k'$. Let $\psi_k \coloneqq \psi_k' \land (t=k)$.
    \item Call our procedure on $\psi$, obtaining a formula $\psi^{+}$. Let $\psi_{\geq 2} \coloneqq \psi^+ \land t \geq 2$.
    \item Call our procedure on $\psi\sub{-t}{t}$, obtaining~$\psi^{-}$.
        Let $\psi_{\leq -2} \coloneqq \psi^{-}\sub{-t}{t} \land t \leq -2$.
\end{itemize} 
Then, the formula $\psi_{\leq -2} \lor \psi_{-1} \lor \psi_{0} \lor \psi_{1} \lor \psi_{\geq 2}$ is quantifier-free, and equivalent to $\psi$.

As a second assumption, we only consider the case where $\psi$ has a single block of existential quantifiers: $\psi(\vec x_0) = \exists \vec x_1 : \phi(\vec x_0,\vec x_1)$.
A procedure for this type of formulae can be iterated bottom-up to eliminate arbitrarily many blocks 
of quantifiers (rewriting~$\forall \vec x$ as $\lnot \exists \vec x \lnot$).

    \begin{algorithm}[t]
        \caption{\Master: A quantifier elimination procedure for $\PPA$.}
        \label{algo:master}
        \begin{algorithmic}[1]
            \Require $\exists \vec x :\phi(\vec x, \vec z)$ where $\phi$ is a positive Boolean combination of $\PPA$ constraints.
            \NDBranchOutput
            positive Boolean combination $\psi_\beta(\vec z)$ of $\PPA$ constraints.
            \GlobalSpec
            $\bigvee_{\beta}\psi_\beta$ is equivalent to $\exists \vec x:\phi$.
            \vspace{3pt}
            \While{$\phi$ contains a subterm of the form $\big\lfloor\frac{\tau}{t^d}\big\rfloor$}
            \label{line:master:while:frac}
                \State append a fresh variable $x$ to $\vec x$
                \label{line:master:append-y-frac}
                \State $\phi\gets\phi\sub{x}{\big\lfloor\frac{\tau}{t^d}\big\rfloor} \land (t^d \cdot x \leq \tau)\land( \tau < t^d \cdot x + t^d)$
                \label{line:master:replace-frac}
            \EndWhile
            \State $\vec y\gets \varnothing$; \quad $B\gets \varnothing$
            \Comment{variables and map used to remove occurrences of $(\cdot \bmod{f(t)})$}
            \While{$\phi$ contains a subterm of the form $(\tau \bmod {f(t)})$}
                \If{$\ast$}
                    \label{line:master:first-nondet-branch}
                \Comment{non-deterministic choice: skip or execute} 
                    \State $\phi \gets \phi\sub{\tau}{\tau \bmod f(t)} \land f(t) = 0$ 
                    \label{line:master:replace-mod-zero}
                    \State \textbf{continue}
                    \label{line:master:mod-0}
                \EndIf
                \State \myguess ${\pm} \gets$ symbol in $\{+,-\}$
                \label{line:master:mod-sign}
                \State append a fresh variable $y$ to $\vec y$ and update $B$ : add the key-value pair $(y,\, \pm f(t) -1)$
                \label{line:master:bounded-quantifier}
                \State $\phi \gets \phi\sub{y}{\tau \bmod {f(t)}} \land (f(t) \divides \tau - y)$
                \label{line:master:replace-mod}
            \EndWhile
            \State \textbf{return} $\ElimBounded(\BoundedElimDiv(\exists\vec y \leq B : \BoundedQE( \exists \vec x : \phi(\vec x, \vec y, \vec z))))$
            \label{line:master:return}
        \end{algorithmic}
        \vspace{-1pt}
    \end{algorithm}%

The pseudocode of our procedure is given in~\Cref{algo:master} (\Master).
It describes a non-deterministic algorithm: 
for an input $\exists \vec x : \phi$,
each non-deterministic execution of \Master returns a positive Boolean combination 
of \PPA constraints $\psi(\vec z)$. The disjunction obtained by aggregating all output formulae
is equivalent to $\exists \vec x : \phi$; so it is this disjunction that must ultimately be used to perform quantifier elimination.
The choice to present the procedure in this manner is not merely stylistic: 
it automatically implements Reddy and Loveland's optimisation for Presburger arithmetic~\cite{Reddy78}.
In quantifier elimination procedures for \PA, eliminating a single variable $x$ from an existential 
block~$\exists \vec y \exists x$  produces a formula~${\bigvee_i \gamma_i}$ with a DNF-like structure. 
Reddy and Loveland observed that pushing the remaining existential quantifiers~$\exists \vec y$ inside the scope 
of the disjunctions, i.e., rewriting $\exists \vec y \bigvee_i \gamma_i$ into $\bigvee_i \exists \vec y \gamma_i$, and then performing quantifier elimination locally to each disjunct leads to a faster procedure.
By keeping variables local to a non-deterministic branch, one achieves the same effect. 

We now describe the four components that make up~\Master, 
which can be summarized under the following titles: 
pre-processing (lines~\ref{line:master:while:frac}--\ref{line:master:replace-mod}), bounded quantifier elimination (call to~\BoundedQE), elimination of divisibility constraints (call to~\BoundedElimDiv),
and elimination of all bounded quantifiers (call to~\ElimBounded).
For the rest of the section, let $\exists \vec x : \phi(\vec x, \vec z)$ be the input to~\Master, 
where $\phi$ is a positive Boolean combination of \PPA constraints.

\subparagraph*{Pre-processing (lines~\ref{line:master:while:frac}--\ref{line:master:replace-mod}).}
These lines remove all occurrences of the integer division functions $x \mapsto \floor{\frac{x}{t^d}}$ and of the integer remainder functions ${x \mapsto (x \bmod f(t))}$, at the expense of adding new existentially quantified variables that are later eliminated. 
After this step, the formula is a positive Boolean combination of \emph{linear} constraints.
For the integer division function, the algorithm simply adds to the sequence of variables $\vec x$ 
to be eliminated a fresh variable $x$ to proxy a term ${\floor{\frac{\tau}{t^d}}}$ (line~\ref{line:master:append-y-frac}). It then relies on the equivalence 
$x = \floor{\frac{\tau}{t^d}} \iff t^d \cdot x \leq \tau \,\land\, \tau < t^d \cdot (x + 1)$
to replace ${\floor{\frac{\tau}{t^d}}}$ with $x$ (line~\ref{line:master:replace-frac}).
The removal of integer remainder functions is performed differently. First, let us observe that the following equivalence holds: 
\[ 
    y = (\tau \bmod f(t)) \iff  (f(t) = 0 \land y = \tau) \lor (0 \leq y \land y < \abs{f(t)}-1 \land f(t) \divides \tau - y).
\]
The formula $f(t) = 0 \land y = \tau$ on the right-hand side of the equivalence is considered in 
line~\ref{line:master:mod-0}.
This line is executed conditionally to a non-deterministic branching (line~\ref{line:master:first-nondet-branch}).
If it is not executed, then lines~\ref{line:master:mod-sign}--\ref{line:master:replace-mod} are executed instead; these correspond to the formula ${0 \leq y \land y < \abs{f(t)}-1 \land f(t) \divides \tau - y}$.
The interesting property of this formula is that the variable $y$ appears \emph{bounded} by $0$ from below, and by either $f(t)-1$ or $-f(t)-1$ from above (following the sign of $f(t)$). Instead of quantifying $y$ using standard existential quantifiers (as done for the variables replacing $\floor{\frac{\tau}{t^d}}$), in line~\ref{line:master:bounded-quantifier} we use a bounded quantifier:

\begin{definition}
    A block of bounded quantifiers $\exists \vec w \leq B$ is given by a sequence of variables $\vec w  = (w_1,\dots,w_m)$ and a map $B$ assigning to each variable in $\vec w$ a polynomial in~$\Z[t]$. Its semantics is given by the equivalence
        $\exists \vec w \leq B : \psi 
        \iff 
        \exists \vec w : \bigwedge_{i=1}^m (0 \leq w_i \leq B(w_i)) \land \psi$.
\end{definition}


Let us write $\vec x_\beta, \vec y_\beta, B_\beta$ and $\phi_\beta$ for the values taken by $\vec x$, $\vec y$, $B$ and $\phi$ 
in the non-deterministic branch $\beta$, when the control flow of the program reaches line~\ref{line:master:return}. 
The following equivalence holds, where the disjunction $\bigvee_{\beta}$ ranges across all non-deterministic branches:
\begin{equation} 
    \label{eq:post-pre-processing}
    \exists  \vec x : \phi(\vec x, \vec z) \iff \textstyle\bigvee_{\!\beta}\, \exists \vec y_{\beta} \leq B_{\beta} \, \exists \vec x_{\beta} : \phi_{\beta}(\vec x_\beta, \vec y_\beta, \vec z). 
\end{equation}

\subparagraph*{Bounded quantifier elimination.} 
Once reaching line~\ref{line:master:return}, the algorithm proceeds by calling the procedure \BoundedQE. 
We will discuss this procedure in~\Cref{sec:gaussian-elimination}. In a nutshell, 
its role is to replace the quantifiers $\exists \vec x_\beta$ on the right-hand side of~\Cref{eq:post-pre-processing} 
with bounded quantifiers, that are merged with the already existing block of bounded quantifiers $\exists \vec y_{\beta} \leq B_{\beta}$.
The formal specification of \BoundedQE is given in the next lemma.


\begin{lemma}\label{lemma:gaussianqe}
    There is a non-deterministic procedure with the following specification: 

    \begin{center}
        \vspace{-5pt}%
        \begin{minipage}{0.95\linewidth}
        \begin{algorithmic}[1]
            \Require $\exists\vec x :\phi(\vec x, \vec z)$, with $\phi$  
            positive Boolean combination of linear \PPA constraints.
            \NDBranchOutput
            a formula ${\exists\vec w_{\beta} \leq B_{\beta}} :\psi_\beta(\vec w_\beta, \vec z)$,
            where $\psi_\beta$ is a positive Boolean combination of linear \PPA constraints.
        \end{algorithmic}
        \end{minipage}
    \end{center}
    The algorithm ensures that the disjunction $\bigvee_{\beta}\exists\vec w_\beta \leq B_{\beta} : \psi_\beta$ 
    of output formulae ranging over all non-deterministic branches is equivalent to  $\exists \vec x :\phi$.
\end{lemma}

As stated, the lemma above is also proved by Goodrick in~\cite{Goodrick18}, who introduced the first bounded quantifier elimination procedure specifically for~\PPA. 
When applied to existential formulae, that procedure 
requires doubly-exponential time, making it unsuitable for establishing~\Cref{theorem:complexity}. In contrast, \BoundedQE runs in non-deterministic polynomial time.
Due to this difference, we cannot rely directly on~\cite{Goodrick18} and must 
thus re-establish~\Cref{lemma:gaussianqe}. 

\begin{algorithm}[t]
    \caption{\BoundedElimDiv: Elimination of divisibility constraints.}
    \label{algo:elimdiv}
    \begin{algorithmic}[1]
        \Require $\exists\vec w \leq B :\psi(\vec w, \vec z)$, with $\psi$ positive Boolean combination of linear \PPA constraints.
        \NDBranchOutput a formula $\exists\vec w_\beta \leq B_\beta :\psi_{\beta}(\vec w_\beta, \vec z)$ in~\PPA, where $\psi_\beta$ is a positive Boolean combination of linear equalities and inequalities,
        and equalities of the form ${\sigma(\vec w) + (\tau(\vec z) \bmod{f(t)}) = 0}$,
        with $\sigma$ linear, and $\tau$ linear and non-shifted.
        \GlobalSpec
        $\bigvee_{\beta} (\exists\vec w_\beta \leq B_\beta :\psi_{\beta})$ is equivalent to $\exists\vec w \leq B : \psi$.
        \vspace{3pt}
        \ForEach{divisibility $f(t) \divides \sigma(\vec w) + \tau(\vec z)$ in $\psi$, where $\tau$ is non-shifted}
        \label{algo:bounded-elimi-div:pick-divisibility}
            \State \textbf{let} $\sigma(\vec w)$ be the term $f_0(t) + \sum_{i=1}^n f_i(t) \cdot w_i$, where $\vec w = (w_1,\dots,w_n)$
            \State $d \gets (n + 3) \cdot \max\{\bitlength{f}, \bitlength{f_0}, \bitlength{f_i} \cdot \bitlength{B(w_i)}\, : i \in [1,n]\}$
            \label{algo:bounded-elimi-div:max-degree}
            \State append a fresh variable $y$ to $\vec w$ and update $B$ : add the key-value pair $(y,\, t^d)$
            \label{algo:bounded-elimi-div:bounded-quantifier}
            \State \myguess ${\pm} \gets$ symbol in $\{+,-\}$
            \label{algo:bounded-elimi-div:guess-sign}
            \State update $\psi$ : replace $(f(t) \divides \sigma(\vec w) + \tau(\vec z))$ with  $\pm f(t) \cdot y + \sigma(\vec w) + (\tau\bmod{f(t)}) = 0$
            \label{algo:bounded-elimi-div:update-psi}
        \EndFor
        \State \textbf{return} $\exists\vec w \leq B : \psi$
    \end{algorithmic}
    \vspace{-1pt}
\end{algorithm}

\subparagraph*{Removing divisibility constraints: more bounded quantifiers.}
\BoundedQE~introduces divisibility constraints $f(t) \divides \tau$. 
The next step, detailed in~\Cref{algo:elimdiv} (\BoundedElimDiv), eliminates all divisibility constraints in favour, once more, of bounded quantifiers.

The idea behind~\Cref{algo:bounded-elimi-div:pick-divisibility} is as follows. 
Let $f(t) \divides \sigma(\vec w) + \tau(\vec z)$ be a constraint 
from the input formula, where $\vec w$ are the variables 
in the block of bounded quantifiers (these correspond to the variables $\vec y_\beta$ from~\Cref{eq:post-pre-processing} and those introduced by~\BoundedQE), and $\vec z$ are the free variables.
Notice that this constraint is equivalent to $f(t) \divides \sigma(\vec w) + (\tau(\vec z) \bmod f(t))$, 
which is in turn equivalent to the existential formula 
$\exists y : f(t) \cdot y + \sigma(\vec w) + (\tau(\vec z) \bmod f(t)) = 0$, 
where $y$ is a fresh variable ranging over $\Z$.
Since $\vec w$ is constrained by bounded quantifiers, 
we can upper-bound the number of digits in the base $t$ encoding of the linear term $\sigma(\vec w)$ 
(recall: $t \geq 2$). 
When $f(t) \neq 0$, the same applies to $(\tau(\vec z) \bmod f(t))$, which ranges between $0$ and $f(t)-1$; and this in turn imposes a bound on the base $t$ representation of $y$.
When ${f(t) = 0}$ instead, the truth of $f(t) \cdot y + \sigma(\vec w) + (\tau(\vec z) \bmod f(t)) = 0$
only depends on whether $\sigma(\vec w) + (\tau(\vec z) \bmod f(t)) = 0$, 
and we can thus restrict $y$ to any non-empty interval.
This allows us to replace the quantifier $\exists y$ with a bounded quantifier (lines~\ref{algo:bounded-elimi-div:max-degree} and~\ref{algo:bounded-elimi-div:bounded-quantifier}). 
Since $y$ ranges over $\Z$, whereas bounded quantifiers use non-negative ranges, 
the algorithm explicitly guesses the sign of $y$ in line~\ref{algo:bounded-elimi-div:guess-sign}, 
allowing $y$ to only range over $\N$ instead.
Formalising these arguments yields the following lemma.

\begin{restatable}{lemma}{LemmaBounedElimDivCorrect}\label{lemma:boundedelimdiv}
    \Cref{algo:elimdiv} (\BoundedElimDiv) complies with its specification.
\end{restatable}

\subparagraph*{Elimination of all bounded quantifiers.}
From the output of~\BoundedElimDiv, the final operation by~\Master is a call to~\ElimBounded, which removes all bounded quantifiers. This algorithm is detailed in~\Cref{sec:division}. 
Its specification is given in the next lemma.

\begin{lemma}\label{lemma:elimbounded}
    There is a non-deterministic procedure with the following specification: 

    \begin{center}
        \vspace{-5pt}%
        \begin{minipage}{\linewidth}
        \begin{algorithmic}[1]
            \Require ${\exists\vec w \leq B} : \phi(\vec w,\!\vec z)$, with $\phi$  positive Boolean combination 
            of linear~\PPA (in)equalities 
            and constraints $\sigma(\vec w) + (\tau(\vec z) \bmod{f(t)}) = 0$,
            with $\sigma$ linear, and $\tau$ linear and non-shifted.
            \NDBranchOutput
            a positive Boolean combination $\psi_\beta(\vec z)$ of $\PPA$ constraints.
            In all divisibility constraints $f(t) \divides \tau$, 
            the divisor $f(t)$ is an integer.
        \end{algorithmic}
        \end{minipage}
    \end{center}
    The algorithm ensures that the disjunction $\bigvee_{\beta} \psi_\beta$ 
    of output formulae ranging over all non-deterministic branches is equivalent to  ${\exists \vec w \leq B} : \phi$.
\end{lemma}
Together, \Cref{eq:post-pre-processing} and~\Cref{lemma:gaussianqe,lemma:boundedelimdiv,lemma:elimbounded}
show that~\Master meets its specification; thus showing~\Cref{theorem:main} 
conditionally to the correctness of~\BoundedQE and~\ElimBounded.
\section{Efficient bounded quantifier elimination in \PPA}
\label{sec:gaussian-elimination}

This section outlines the arguments leading to the procedure \BoundedQE.  
Its pseudocode is given in~\Cref{algo:gaussianqe},
and technical details can be found in Appendix~\ref{app:gaussian-elimination}.
We start with an example demonstrating the key idea used to develop \underline{a} version 
of bounded quantifier elimination in \PPA.  
These arguments are sufficient for establishing~\Cref{lemma:gaussianqe}; although they do not result in an 
optimal procedure complexity-wise. We will then recall the main arguments used in~\cite{ChistikovMS24} 
to obtain an optimal procedure, which, when implemented, yield~\BoundedQE. 

Let us consider a formula $\exists x : \phi(x, \vec z)$ where, for simplicity, $\phi$ is of the form:  
\begin{align*}
    \tau(\vec z) \leq a \cdot x \,\land\, b \cdot x \leq \rho(\vec z) 
    \,\land\, \big(m \divides c \cdot x + \sigma(\vec z)\big) \,\land\, 
    a > 0 \,\land\, b > 0 \,\land\, m > 0, 
\end{align*}
where $a$, $b$, $c$ and $m$ are polynomials from $\Z[t]$, and $\tau$, $\rho$ and
$\sigma$ are \emph{linear} \PPA terms.
For the time being, we invite the reader to pick some values for~$t$ and the free variables~$\vec z$, so that the formula~$\phi$ 
becomes a formula from Presburger arithmetic in a single variable~$x$.
The standard argument for eliminating $x$ in \PA goes as follows (see, e.g.,~\cite{Weispfenning90}). 
We first update the inequalities to ensure that all coefficients of $x$ are equal; 
this results in the inequalities $b \cdot \tau \leq a \cdot b \cdot x$ and 
$a \cdot b \cdot x \leq a \cdot \rho$.  
The quantification $\exists x$ expresses that there is $g \in \Z$~such~that 
\textbf{(i)} $g$ is a multiple of $a \cdot b$ that belongs to the interval $[b \cdot \tau, a \cdot \rho]$; and
\textbf{(ii)} $m$ divides $c \cdot \frac{g}{a \cdot b} + \sigma$.
The key observation is that such an integer (if~it~exists) can be found by only looking at elements of 
$[b \cdot \tau, a \cdot \rho]$ that are ``close'' to $b \cdot \tau$. 
More precisely, the properties \textbf{(i)}~and~\textbf{(ii)}
must be simultaneously satisfied by $b \cdot \tau + r$, for some $r \in  [0, a \cdot b \cdot m]$.
We can thus restrict $x$ to satisfy an additional constraint $a \cdot b \cdot x = b \cdot \tau + r$. 
A small refinement: since this equality is unsatisfiable when the shift $r$ is not a multiple of $b > 0$, we can rewrite it as $a \cdot x = \tau + s$, where the shift $s$ now ranges in $[0, a \cdot m]$.
Observe that $s$ lies in an interval that is independent of the values picked for $\vec z$;
as $a$ and $m$ were originally polynomials in $t$. 

Let us keep assigning a value to the parameter $t$ (so, $a$ and $m$ are still integers), 
but reinstate the variables~$\vec z$.
From the above argument, the formula $\exists x : \phi(x, \vec z)$ 
is equivalent to $\bigvee_{s = 0}^{a \cdot m} \,\exists x\, (\phi(x,\vec z) 
\,\land\, a \cdot x = \tau + s)$.
It is now straightforward to eliminate $x$ from each disjunct $\exists x\, (\phi(x,\vec z) 
\,\land\, a \cdot x = \tau + s)$: we simply ``apply'' the equality $a \cdot  x =  \tau + s$, substituting 
$x$ for $\frac{\tau+r}{a}$, and add a divisibility constraint forcing $\tau + s$ to be a multiple of $a$. 
After this substitution, both $a \cdot x = \tau + s$ and $\tau(\vec z) \leq a \cdot x$ become $\top$. 
The resulting disjunct is 
\begin{equation*}
   \psi(s,\vec z) \coloneqq b \cdot (\tau + s) \leq a \cdot \rho \,\land\,
    (m \cdot a \divides c \cdot (\tau + s) + a \cdot \sigma) \,\land\,
    (a \divides  \tau + s) \,\land\, a > 0 \,\land\, b > 0 \,\land\, m > 0,
\end{equation*}
and $\bigvee_{s=0}^{a \cdot m} \psi(s,\vec z)$ is equivalent to $\exists x : \phi(x, \vec z)$. 
(For~\PA, this concludes the quantifier elimination procedure.) 
When restoring the parameter $t$, these two formulae are still equivalent, 
but the number of disjunctions~$\bigvee_{s=0}^{a \cdot m}$
now depends on~$t$. We replace them with a bounded quantifier, 
rewriting $\bigvee_{s = 0}^{a \cdot m} \psi(s,\vec z)$ as 
${\exists s \leq B : \psi(s,\vec z)}$, where ${B(s) \coloneqq a(t) \cdot m(t)}$. 
This is, in a nutshell, the \emph{bounded quantifier elimination procedure} from~\cite{Goodrick18,Weispfenning97}.

When the signs of $a$, $b$, and $m$ are unknown 
(i.e., $\phi$ does not feature the constraints $a > 0$, $b > 0$ and $m > 0$), 
we must perform a ``sign analysis'': we write a disjunction (or guess) over all possible 
signs of the three polynomials. 
In~\Cref{algo:gaussianqe}, the lines marked in \colorbox{lip-yellow!50}{yellow} 
are related to this analysis; e.g., line~\ref{gauss:guess-sign} guesses the sign of the (non-zero) coefficient $a$ of $x$.

    \begin{algorithm}[t!]
        \caption{\textsc{BoundedQE}: A bounded quantifier elimination procedure for~$\PPA$.}
        \label{algo:gaussianqe}
        \begin{algorithmic}[1]
            \Require $\exists \vec x: \phi(\vec x, \vec z)$ where $\phi$ is a  
                positive Boolean combination of linear \PPA constraints.
            \NDBranchOutput
            a formula $\exists\vec w_{\beta} \leq B_{\beta} :\psi_\beta(\vec w_\beta, \vec z)$ 
            where $\psi_\beta$ is a positive Boolean combination of linear \PPA constraints.
            \GlobalSpec
            $\bigvee_{\beta}\exists\vec w_\beta \leq B_{\beta} : \psi_\beta$ is equivalent to $\exists \vec x :\phi$.  
            \vspace{3pt}
            \algrenewcommand{\alglinenumber}[1]{\colorbox{lipicsYellow!50}{\footnotesize#1:}}
            \State \myguess $Z$ $\gets$ subset of $\{f(t) : \text{the relation $(f(t) \divides \cdot)$ occurs in $\phi$}\}$ 
            \label{gauss:set-Z}
            \ForEach{$f(t)$ in $Z$}\label{gauss:Z-for}
                \State update $\phi$ : replace each divisibility $f(t) \divides \tau$ with $\tau = 0$\label{gauss:set-Z-replace}  
                \State $\phi \gets \phi \land (f(t) = 0)$\label{gauss:set-Z-end}
            \EndFor
            \State \myguess $\pm \gets$ symbol in $\{-,+\}$
            \Comment{sign required to make $m(t)$ below positive}
            \label{gauss:guess-sign-mt}
            \algrenewcommand{\alglinenumber}[1]{\colorbox{white}{\footnotesize#1:}}
            \State $m(t) \gets \pm \prod\{ f(t) : \text{the relation $(f(t) \divides \cdot)$ occurs in $\phi$} \}$
            \label{gauss:guess-mt}
            \algrenewcommand{\alglinenumber}[1]{\colorbox{lipicsYellow!50}{\footnotesize#1:}}
            \State $\chi \gets (m(t) > 0)$
            \label{gauss:mt-is-positive}
            \algrenewcommand{\alglinenumber}[1]{\colorbox{RoyalBlue!25}{\footnotesize#1:}}
            \State $(\pm,\lead) \gets (+,1)$;
            \quad $B \gets \varnothing$ \label{gauss:bareissfactors}
            \Comment{$B$: map from variables to upper bounds}
            \algrenewcommand{\alglinenumber}[1]{\colorbox{white}{\footnotesize#1:}}
            \State update $\phi$ : replace each inequality $\tau \le 0$ with $\tau + y = 0$, where $y$ is a fresh slack variable \label{gauss:introduce-slack}
            \ForEach{$x$ in $\vec x$} \label{gauss:mainloop}
                \algrenewcommand{\alglinenumber}[1]{\colorbox{lipicsYellow!50}{\footnotesize#1:}}
                \If{$\ast$} \label{gauss:ast}
                \Comment{non-deterministic choice: skip or execute} 
                    \State update $B$ : add the key-value pair $(x,m(t)-1)$ \label{gauss:subst-rem}
                    \State \mycontinue \label{gauss:guess-all-zeros}
                \EndIf
                \algrenewcommand{\alglinenumber}[1]{\colorbox{white}{\footnotesize#1:}}
                \State \myguess $f(t) \cdot x + \tau = 0 \ \gets \text{equality in $\phi$ that contains $x$}$ \label{gauss:guess-equation}
                \algrenewcommand{\alglinenumber}[1]{\colorbox{RoyalBlue!25}{\footnotesize#1:}}
                \State ${\prevlead} \gets \lead$; \quad $\lead \gets f(t)$   \label{gauss:update-factor}
                \algrenewcommand{\alglinenumber}[1]{\colorbox{lipicsYellow!50}{\footnotesize#1:}}
                \Comment{previous and current leading coefficients}
                \State $\pm \gets$ \myguess a symbol in $\{-,+\}$
                \Comment{sign of $f(t)$}
                \label{gauss:guess-sign}
                \State $\chi \gets \chi \land (\pm f(t) > 0)$
                \label{gauss:assert-sign-ft}
                \State $m(t) \gets \pm f(t) \cdot m(t)$ 
                \algrenewcommand{\alglinenumber}[1]{\colorbox{white}{\footnotesize#1:}}
                \label{gauss:new-factor}
                \If{$\tau$ contains a slack variable $y$ such that $B(y)$ is undefined}\label{gauss:branch}
                    \State update $B$ : add the key-value pair $(y,\, m(t)-1)$\label{gauss:append-seq}
                \EndIf
                \State $\phi \gets \phi \vigsub{\frac{-\tau}{f(t)}}{x}$ \label{gauss:vig}\label{gauss:vigorous}
                \algrenewcommand{\alglinenumber}[1]{\colorbox{RoyalBlue!25}{\footnotesize#1:}}
                \State update $\phi$ : divide all constraints by \prevlead \label{gauss:divide}
                \algrenewcommand{\alglinenumber}[1]{\colorbox{white}{\footnotesize#1:}}
                \Comment{both sides for divisibility constraints}
                \State $\phi \gets \phi \land ( f(t) \divides \tau)$ \label{gauss:restore}
            \EndFor \label{gauss:endmainloop} 
            \ForEach{equality $\eta = 0$ of $\phi$ with a slack variable $y$ such that $B(y)$ is undefined} \label{gauss:unslackloop}
                \State update $\phi$ : replace $\eta = 0$ with $\eta\sub{0\!}{\!y} \le 0$ \myif the sign $\pm$ is plus \myelse with $\eta\sub{0\!}{\!y} \ge 0$ \label{gauss:drop-slack}%
            \EndFor%
            \State \textbf{return} $\exists\vec w \leq B : \phi \land \chi$ 
            \textbf{where} $\vec w$ is the sequence of keys of the map $B$\label{gauss:return}
        \end{algorithmic}
        \vspace{-1pt}
        \end{algorithm}%

\subparagraph*{Taming the complexity of the procedure.}
Problems arise when looking at the complexity of the procedure outlined above. 
To understand this point, consider the inequality ${b \cdot (\tau + s) \leq a \cdot \rho}$, 
which was derived by substituting $\frac{\tau + s}{a}$ for $x$ in $b \cdot x \leq \rho$, 
and suppose $\tau = c \cdot y + \tau'$ and $\rho = d \cdot y + \rho'$,
for some variable $y$.
This inequality can be rewritten as $(b \cdot c - a \cdot d) \cdot y + b \cdot \tau' - a \cdot \rho' + b \cdot s \leq 0$. 
When looking at the coefficient $(b \cdot c - a \cdot d)$ of $y$ one deduces that, if quantifier elimination is performed carelessly on a block $\exists \vec x$ of 
multiple existential quantifiers, the coefficients of the variables in the formula will grow quadratically with each eliminated variable.
Then, by the end of the procedure, their binary bit size will be exponential in the number of variables in~$\vec x$.
However, this explosion can be avoided by noticing that coefficients
are updated following the same pattern as in Bareiss' polynomial-time Gaussian elimination procedure~\cite{Bareiss68}.
This insight was highlighted in~\cite{ChistikovMS24}, building upon an earlier observation from~\cite{Weispfenning97}. 
In Bareiss' algorithm, the key to keeping coefficients polynomially bounded is given by the Desnanot--Jacobi identity.
Consider an~$m \times d$ matrix~$A$.
Let us write $A[i_1,\dots,i_r; j_1,\dots,j_\ell]$ for the $r \times \ell$ sub-matrix of $A$ made of the rows with indices ${i_1,\dots,i_r \in [1,m]}$ and columns with indices $j_1,\dots,j_\ell \in [1,d]$.
For $i,j,\ell \in \N$ with $0 \leq \ell \leq \min(m,d)$,  $1 \leq i \leq m$ and $1 \leq j \leq d$,
we define $a_{i,j}^{(\ell)} \coloneqq \det A[1,\dots,\ell,i; 1,\dots,\ell,j]$.

\begin{proposition}[Desnanot--Jacobi identity]
    \label{thm:desnanot-jacobi}
    For every $i,j,\ell \in \N$ with $\ell \geq 2$, $\ell < i \leq m$ and $\ell < j \leq d$,
    we have
    $(a_{\ell,\ell}^{(\ell-1)} \cdot  a_{i,j}^{(\ell-1)} - a_{\ell,j}^{(\ell-1)} \cdot a_{i,\ell}^{(\ell-1)})
    =
    a_{\ell,\ell}^{(\ell-2)} \cdot a_{i,j}^{(\ell)}$.
\end{proposition}
The Desnanot--Jacobi identity is true for all matrices with entries over an integral domain 
(a non-zero commutative ring in which the product of non-zero elements is non-zero), 
and therefore we can take the entries of $A$ to be polynomials in $\Z[t]$. 

Returning to our informal discussion, we now see that the coefficient $(b \cdot c - a \cdot d)$ 
of~$y$ is oddly similar to the left-hand side of the Desnanot--Jacobi identity.
Suppose that the elements $a_{i,j}^{(\ell-1)}$ are 
the coefficients of the variables in the formula currently being processed by the quantifier elimination procedure, 
and we are eliminating the $\ell$-th quantifier. \Cref{thm:desnanot-jacobi} tells us that all coefficients produced by 
the na\"ive elimination (left-hand side of the identity) 
have~$a_{\ell,\ell}^{(\ell-2)}$ as a common factor. 
By dividing through by this common factor, we obtain smaller coefficients for the next step of variable elimination ---namely, 
$a_{i,j}^{(\ell)}$. When eliminating the first variable ($\ell = 1$), the common factor is~$1$. 
Otherwise, it is the coefficient~$a$ that the $(\ell-1)$-th eliminated variable $x$ has in the equality 
$a \cdot x = \tau + s$ used for the elimination.
In~\Cref{algo:gaussianqe}, the lines marked in \colorbox{RoyalBlue!25}{blue} implement Bareiss' optimisation: 
line~\ref{gauss:bareissfactors} initialises the common factor $\ell(t)$ and its sign, line~\ref{gauss:update-factor} updates it, 
and line~\ref{gauss:divide} performs the division.

\subparagraph*{Some details on~\BoundedQE.}


Lines~\ref{gauss:set-Z}--\ref{gauss:set-Z-end} handle the divisibility constraints $f(t) \divides \tau$ 
with the divisor $f(t)$ equal to $0$. Such constraints are equalities in disguise, and the procedure replaces 
them with~$\tau = 0$. 
When the procedure reaches line~\ref{gauss:guess-sign-mt}, all divisors in the divisibility constraints 
are assumed non-zero.
Following the example from the previous paragraph, recall that the shifts~$s$ 
belong to intervals that depend on these divisors; when multiple divisors occur, the procedure for 
\PA takes their $\lcm$ (instead of just $m$ as in our example). 
For simplicity, instead of $\lcm$, \BoundedQE considers the absolute value~$m(t)$ of their product (line~\ref{gauss:guess-mt}). 
After guessing the sign~$\pm$ of this product, the procedure enforces $m(t) > 0$ in line~\ref{gauss:mt-is-positive}. This information is stored in the formula $\chi$, which accumulates all sign guesses made by the algorithm; these are conjoined to $\phi$ when the procedure returns.

Line~\ref{gauss:introduce-slack} replaces all inequalities with equalities by introducing slack variables 
ranging over~$\N$. (This step is inherited from~\cite{ChistikovMS24}.) 
Slack variables represent the shifts~$s$ from the quantifier elimination procedure for \PA.
Line~\ref{gauss:subst-rem} covers the corner cases of $x$ not appearing in equalities, or $t$ being such that all the coefficients $f(t)$ of $x$ evaluate to zero. 
After guessing an equality $f(t) \cdot x + \tau = 0$ to perform the substitution (line~\ref{gauss:guess-equation}), 
line~\ref{gauss:branch} checks whether $\tau$ features a slack variable~$y$ (i.e., the equality was originally an 
inequality). If so, the procedure generates a bounded quantifier for~$y$.
The elimination of $x$ (line~\ref{gauss:vigorous}) is performed with the \emph{vigorous substitution}~$\phi\vigsub{\frac{-\tau}{f(t)}}{x}$ 
which works as follows: 
\textbf{1:}~Replace every equality $\rho = 0$ with $f(t) \cdot \rho = 0$, and every divisibility 
$g(t) \divides \rho$ with $f(t) \cdot g(t) \divides f(t) \cdot \rho$; 
this is done also for constraints where $x$ does not occur.
\textbf{2:}~Replace every occurrence of $f(t) \cdot x$ with~$\tau$
(from step~\textbf{1}, each coefficient of $x$ in the system can be factored as $f(t) \cdot h(t)$ 
for some $h \in \Z[t]$).

After applying the vigorous substitution, the procedure divides all coefficients of the inequalities 
and divisibility constraints in $\phi$ by the common factor of the Bareiss' optimisation (line~\ref{gauss:divide}). 
In the case of divisibility constraints, divisors are also affected.
\Cref{thm:desnanot-jacobi} ensures that these divisions are all without remainder. 
In practice, the traditional Euclidean algorithm for polynomial division can be used to construct 
the quotient in polynomial time. As a result of these divisions, 
throughout the procedure 
all polynomials~$f(t)$ guessed in line~\ref{gauss:guess-equation} 
have polynomial bit sizes in the size of the input formula. 


After the \textbf{foreach} loop of line~\ref{gauss:mainloop} completes, all variables from $\vec x$ have been eliminated (line~\ref{gauss:vigorous}) or bounded (line~\ref{gauss:subst-rem}). 
A benefit of translating inequalities into equalities in line~\ref{gauss:introduce-slack} is that $x$ can be eliminated independently of the sign of its coefficient $f(t)$; inequalities would need to flip for $f(t)$ negative instead. 
The final step (lines~\ref{gauss:unslackloop} and~\ref{gauss:drop-slack}) drops all slack variables for which no bound was assigned in line~\ref{gauss:append-seq}, 
reintroducing the inequalities (the sign stored in~$\pm$ tells us the direction of these inequalities). This step is also inherited 
from~\cite{ChistikovMS24}. 

By fully developing the arguments above, one shows that~\BoundedQE is correct:
\begin{restatable}{lemma}{LemmaBoundedQECorrectness}
    \label{lemma:boundedqe:correctness}
    \Cref{algo:gaussianqe} (\BoundedQE) complies with its specification.
\end{restatable}

\noindent
This lemma implies~\Cref{lemma:gaussianqe}. 
In the sequel we will also need the next lemma, discussing 
properties of the outputs of \BoundedQE
for ``Presburger-arithmetic-like'' inputs.

\begin{restatable}{lemma}{LemmaBoundedQEInt}\label{lemma:boundedqe:int}
    Let $\exists \vec x : \phi(\vec x, \vec z)$ be a formula input of \Cref{algo:gaussianqe}, 
    in which all coefficients of the variables in $\vec x$, and all divisors $f(t)$ in relations $(f(t) \divides {\cdot})$,
    are integers.
    The map~$B_\beta$ in the output of each non-deterministic branch $\beta$ ranges over the integers.
\end{restatable}

\section{Elimination of polynomially bounded quantifiers}
\label{sec:division}

\begin{algorithm}[t!]
    \caption{\textsc{ElimBounded}: Elimination of polynomially bounded quantifiers.}
    \label{algo:qe}
    \begin{algorithmic}[1]
        \Require $\exists\vec w \leq B :\phi(\vec w, \vec z)$, with $\phi$ positive Boolean combination of linear~\PPA (in)equalities, 
        and constraints $\sigma(\vec w) + (\tau(\vec z) \bmod{f(t)}) = 0$,
        with $\sigma$ linear, and $\tau$ linear and non-shifted.
        \NDBranchOutput
        a positive Boolean combination $\psi_\beta(\vec z)$ of $\PPA$ constraints.
        \GlobalSpec
        $\bigvee_{\beta}\psi_\beta$ is equivalent to $\exists\vec w \leq B : \phi$.
        \vspace{2pt}
        \algrenewcommand{\alglinenumber}[1]{\colorbox{Mahogany!25}{\footnotesize#1:}}
        \State $\phi\gets\phi\land\bigwedge_{w \in \vec w}(0\leq w)\land(w \leq B(w))$
        \label{line:elimbounded:add-bounds} 
        \State $M \gets \max\{ \bitlength{B(w)} \ : \ w \in \vec w\}$
        \label{line:elimbounded:m}
        \State $\vec y\gets\varnothing$
        \label{line:elimbounded:intro-temp}
        \Comment{$\vec y$ is a vector of variables used to ``bit blast'' bounded variables}
        \ForEach{$w$ in $\vec w$}
        \label{line:elimbounded:foreach-div-x}
            \State append fresh variables $y_0,\dots,y_{M}$ to $\vec y$ 
            \label{line:elimbounded:append-y}
            \State $\phi \gets \phi\sub{(t^{M} \cdot y_{M} + \dots + t \cdot y_1 + y_0)}{w}\,\land\,\bigwedge_{i=0}^{M} ((0\leq y_i)\land (y_i\leq t - 1))$
            \label{line:elimbounded:div-x}
        \EndFor
        \algrenewcommand{\alglinenumber}[1]{\colorbox{lipicsYellow!25}{\footnotesize#1:}}
        \While{a variable from $\vec y$ has a non-integer coefficient in a constraint $(\eta \sim 0)$ of $\phi$}
            \label{line:elimbounded:while}   
            \If{the symbol $\sim$ is $\leq$} $\eta \gets \eta - 1$
            \label{line:elimbounded:sim}
            \Comment{we work with $\eta - 1 < 0$; else $\sim$ is $=$}
            \EndIf
            \State \textbf{let} $\eta$ be $(\sigma(\vec y)\cdot t + \rho(\vec y) + \tau(\vec z))$, 
                where $\rho$ does not contain $t$, and $\tau$ is non-shifted
            \label{line:elimbounded:eta}
            \State $\rho \gets \rho(\vec y) + (\tau(\vec z) \bmod t)$
            \label{line:elimbounded:rho}
            \Comment{add to $\rho$ the unbounded part modulo $t$}
            \State \myguess $r\gets$ integer in $[-\onenorm{\rho},\onenorm{\rho}]$
            \label{line:elimbounded:guess-r}
            \Comment{quotient of the division of $\rho$ by $t$} 
            \If{the symbol $\sim$ is $=$}
                \label{line:elimbounded:start-construction-gamma}
                \State $\gamma\gets(t \cdot r=\rho)$
                \label{line:elimbounded:add-eq-to-psi}
            \Else 
                \State $\gamma\gets(t\cdot r\leq \rho)\land(\rho \leq t\cdot(r+1)-1)$
                \label{line:elimbounded:add-neq-to-psi}
                \State $r \gets r+1$
                \label{line:elimbounded:increment-r}
            \EndIf
            \State update $\phi$ : replace $(\eta \sim 0)$ with $\gamma \,\land\, 
            (\sigma + r + \sfloor{\frac{\tau}{t}} \sim 0\big)$
            \label{line:elimbounded:replace-in-phi}
        \EndWhile
        \algrenewcommand{\alglinenumber}[1]{\colorbox{ForestGreen!25}{\footnotesize#1:}}
        \State $\vec z'\gets\varnothing$; \quad $S \gets \varnothing$
        \label{line:elimbounded:intro-s}
        \Comment{$\vec z'$ are used to rewrite $\phi$ as a combination of linear constraints}
        \ForEach{constraint $(\rho(\vec y) + \tau(\vec z) \sim 0)$ of $\phi$, where $\tau$ is non-shifted}
        \label{line:elimbounded:foreach-constraint}
            \State append a fresh variable $z'$ to $\vec z'$ and update $S$ : add the key-value pair $(z',\, \tau(\vec z))$
            \label{line:elimbounded:z-primes}
            \State $\phi \gets \phi\sub{z'}{\tau(\vec z)}$
            \label{line:elimbounded:replace-tau}
        \EndFor
        \State $\exists\vec w' \leq B' :\psi(\vec w',\vec z')\gets$ \textsc{BoundedQE}$(\vec y, \phi(\vec y,\vec z'))$
        \label{line:elimbounded:gauss}
        \ForEach{$w$ in $\vec w'$}
        \label{line:elimbounded:foreach-w-to-int}
            \Comment{now every $B'(w)$ is an integer}
            \State \myguess $g \gets \text{integer in $[0,B'(w)]$}$ 
            \label{line:elimbounded:guess-g}
            \State $\psi\gets\psi\sub{g}{w}$
            \label{line:elimbounded:replace-w}
        \EndFor
        \State \textbf{return} $\psi[\sub{S(z')}{z'} : z' \in \vec z']$
        \label{line:elimbounded:return}
        \algrenewcommand{\alglinenumber}[1]{\colorbox{white}{\footnotesize#1:}}
    \end{algorithmic}
    \vspace{-1pt}
\end{algorithm}%

We move to~\Cref{algo:qe} (\textsc{ElimBounded}), which eliminates the bounded quantifiers in three steps: \colorbox{Mahogany!25}{replacement} of bounded variables by their $t$-ary expansions;
``\colorbox{lipicsYellow!25}{divisions} by $t$'' until all $t$-digits have integer coefficients; 
\colorbox{ForestGreen!25}{elimination} of $t$-digits via~\BoundedQE.

\subparagraph*{Base $t$ expansion (lines~\ref{line:elimbounded:add-bounds}--\ref{line:elimbounded:div-x}).}
Following the semantics of bounded quantifiers,
line~\ref{line:elimbounded:add-bounds} adds to~$\phi$ the bounds 
$0\leq w \land w\leq B(w)$, for each bounded variable $w$.
The subsequent lines ``bit blast'' $w$ 
into its $t$-ary expansion $t^{M} \cdot y_{M} + \dots + t \cdot y_1 + y_0$, 
where $M$ is the largest bit size of the bounds in $B$ (line~\ref{line:elimbounded:m}). 
All added variables~$\vec y$ are $t$-digits, i.e., they range in $[0,t-1]$.

\begin{example}\label{example:division}
    Let us see this step in action on a bounded version of the formula 
    in~\Cref{intro:example-2}:
    \[ 
        \exists (x,y,z) \leq B 
        \ : \ \big(t \cdot y = x + (-a \bmod t)\big) \land \big((t+1) \cdot z = x + (-b \bmod t+1)\big), 
    \]
    where $B(x) = t^2 + t - 1$, and $B(y) = B(z) = t+2$. 
    By~``bit blasting'' the bounded variables into $t$-digits $\vec x=(x_0,x_1,x_2)$, $\vec y=(y_0,y_1,y_2)$, and
    $\vec z=(z_0,z_1,z_2)$, we obtain the formula
    \begin{align*} 
        \exists \vec x,\vec y,\vec z \ : \  
        \bigwedge\nolimits_{w \in \{x,y,z\}} \Big(0 \leq  (w_2 \cdot t^2 + w_1 \cdot t + w_0) \leq B(w) &\land \bigwedge\nolimits_{i\in\{0,1,2\}} 0 \leq w_i < t\Big)\\
        {}\land\, t \cdot (y_2 \cdot t^2 + y_1 \cdot t + y_0) = (x_2 \cdot t^2 + x_1 \cdot t + x_0) &+ (-a \bmod t)\\ 
        {} \land\, (t+1) \cdot (z_2 \cdot t^2 + z_1 \cdot t + z_0) = (x_2 \cdot t^2 + x_1 \cdot t + x_0) &+ (-b \bmod t+1).
        &\hspace{-7pt}\qed
    \end{align*}
\end{example}

\begin{restatable}{lemma}{LemmaElimBoundedPreProcessing}\label{lemma:elimbounded:preprocessing}
    Let $\phi_{\varnothing}(\vec y, \vec z)$ be the formula obtained from $\phi$ by performing lines~\ref{line:elimbounded:add-bounds}--\ref{line:elimbounded:div-x} of \Cref{algo:qe}. 
    Then, the input formula $\exists\vec w \leq B :\phi(\vec w, \vec z)$ is equivalent to 
    $\exists \vec y : \phi_{\varnothing}(\vec y, \vec z)$ and every ($t$-digit) variable in $\vec y$ has 
    only linear occurrences in $\phi_{\varnothing}$.  
\end{restatable}

\vspace{-8pt}
\subparagraph*{The coefficients of the $t$-digits become integers (lines~\ref{line:elimbounded:while}--\ref{line:elimbounded:replace-in-phi}).}
This step is defined by the \textbf{while} loop of line~\ref{line:elimbounded:while}, whose goal is
to transform the formula $\phi_{\varnothing}$ into an equivalent positive Boolean combination of equalities and inequalities in which all coefficients of $\vec y$ are integers.

\begin{example}\label{example:division-two}
    Before delving into the details, 
    let us illustrate the transformation 
    on the equality $(t+1) \cdot (z_2 \cdot t^2 + z_1 \cdot t + z_0) = (x_2 \cdot t^2 + x_1 \cdot t + x_0) + (-b \bmod t+1)$ 
    from~\Cref{example:division}.
    Grouping terms according to powers of $t$, we obtain: 
    \[
        {-z_2 \cdot t^3 + (x_2-z_1-z_2) \cdot t^2} + {(x_1-z_0-z_1) \cdot t + (x_0 - z_0) + (-b \bmod t+1) = 0}.
    \]
    We symbolically perform a division with remainder on the sub-term $(-b \bmod t+1)$ concerning the free variables, 
    rewriting it as $\floor{\frac{-b \bmod t+1}{t}} \cdot t + ((-b \bmod t+1) \bmod t)$. In the resulting equality, we notice that $(x_0-z_0) + ((-b \bmod t+1) \bmod t)$ 
    must be divisible by $t$. Since both $x_0$ and $z_0$ belong to $[0,t-1]$, 
    only two multiples of $t$ are possible: $0$ and $t$. 
    Consider the latter case: we can rewrite the equality as the conjunction of $t = (x_0 - z_0) + ((-b \bmod t+1) \bmod t)$ and 
    $-z_2 \cdot t^3 + (x_2-z_1-z_2) \cdot t^2 + (x_1-z_0-z_1) \cdot t +\floor{\frac{-b \bmod t+1}{t}} \cdot t + t = 0$. 
    By dividing the second equality by $t$, the variables $x_0$, $x_1$ and $z_0$ 
    end up appearing with integer coefficients only.
    Repeating this process guarantees that all quantified variables satisfy this property:
    the second iteration ``frees'' $x_2$ and $z_1$, 
    and the third iteration handles the variable~$z_2$.
    \qed
\end{example}

The~\textbf{while} loop guesses some integers in line~\ref{line:elimbounded:guess-r}.
Let $R_i$ be the (finite) set of all sequences~$\vec s$ 
of guesses from the first $i$ iterations of the loop 
(so,~$\vec s$ has length~$i$), 
and let $\phi_{\vec s}$ be the unique formula obtained from 
$\phi_{\varnothing}$ after iterating the loop $i$ times, using~$\vec s$ as the sequence of guesses. 
(The subscript $\varnothing$ corresponds to the empty sequence of guesses; the only element in $R_0$.)

Together with proving that the \textbf{while} loop preserves formula equivalence (across non-deterministic branches),
the critical parameter to track during the execution of the loop 
is the degrees of all coefficients of the $t$-digits $\vec y$.
Showing that this parameter reaches~$0$ implies loop termination, 
and correctness of this step of the procedure. More formally, we inductively define 
the $\vec y$-\emph{degree} $\deg(\vec y,\phi)$ of a positive Boolean combination of \PPA constraints $\phi(\vec y, \vec z)$, 
where the variables $\vec y=(y_1,...,y_\ell)$ occur only linearly, as follows (below, ${\sim} \in \{{\leq},{=}\}$):

\begin{itemize}
    \item $\deg(\vec y, \phi) \,=\, \max\{ \deg(f_i) : i \in [1,\ell] \}$ \tab\hspace{-7pt} if $\phi$ is an (in)equality $\sum_{i=1}^\ell f_{i}(t) \cdot y_i + \tau(\vec z) \sim 0$; 
    \item $\deg(\vec y, \phi) \,=\, \deg(\vec y, \phi_1) + \deg(\vec y, \phi_2)$ \tab\hspace{-8pt} if $\phi$ is either $(\phi_1 \,\land\, \phi_2)$ or $(\phi_1 \,\lor\, \phi_2)$.
\end{itemize} 

\noindent
Then, the defining property of the \textbf{while} loop of line~\ref{line:elimbounded:while} can be stated as follows:

\begin{restatable}{lemma}{LemmaElimBoundedDtwo}\label{lemma:elimbounded:Dtwo}
    Consider $\vec s \in R_i$ with $\deg(\vec y, \phi_{\vec s}) > 0$, and the set ${G \coloneqq \{ \vec s r \in R_{i+1} : r \in \Z\}}$.
    Then, \labeltext{\textup{\textbf{(i)}}}{lemma:elimbounded:Dtwo:i1} $\phi_{\vec s}$ is equivalent to~$\bigvee_{\vec r \in G}\phi_{\vec r}$, and
    \labeltext{\textup{\textbf{(ii)}}}{lemma:elimbounded:Dtwo:i2} $\deg(\vec y, \phi_{\vec s}) > \deg(\vec y, \phi_{\vec r})$ for every $\vec r \in G$.
\end{restatable} 

\begin{proof}[Proof sketch.]
    The \textbf{while} loop considers an (in)equality $\eta \sim 0$ from $\phi_{\vec s}$ (line~\ref{line:elimbounded:while}); 
    inequalities are transformed into strict inequalities $\eta - 1 < 0$
    in line~\ref{line:elimbounded:sim}, as in this case the latter are easier to work with. Line~\ref{line:elimbounded:eta} 
    represents the term of the (in)equality as $\sigma(\vec y) \cdot t + \rho(\vec y) + \tau(\vec z)$, 
    where $\sigma$ is linear, $\rho$ is a linear term with coefficients in~$\Z$,
    and $\tau(\vec z)$ is non-shifted. 
    
    \begin{remark*} Observe that line~\ref{line:elimbounded:replace-in-phi}  
    will later replace $\eta \sim 0$ with 
    $\sigma + r + \sfloor{\frac{\tau}{t}} \sim 0$. 
    If the latter (in)equality is considered again by the \textbf{while} loop at a later iteration, this replacement will produce the term $\sfloor{\frac{\sfloor{\frac{\tau}{t}}}{t}}$, which can be rewritten as $\sfloor{\frac{\tau}{t^2}}$. We assume the algorithm to implicitly perform this rewriting, so that the term above can in fact be written as 
    \begin{align}\label{eq:elimbounded:constraint:1}
        \textstyle\sigma(\vec y)\cdot t + \rho(\vec y) + \sfloor{\frac{\tau(\vec z)}{t^k}}, 
        \qquad\quad &\text{where $k\geq0$, such that $\textstyle\sfloor{\frac{\tau}{t^0}}\coloneqq \tau$}.        
    \end{align}
    \end{remark*}

    Let us define ${\eta' = \sigma(\vec y) + \sfloor{\frac{\tau(\vec z)}{t^{k+1}}}}$ and 
    ${\rho' = \rho(\vec y) +\big(\sfloor{\frac{\tau(\vec z)}{t^k}} \bmod t\big)}$, so that the term in~\Cref{eq:elimbounded:constraint:1} 
    is then equal to $\eta' \cdot t + \rho'$. 
    (Note: $\rho'$ is exactly the term in line~\ref{line:elimbounded:rho}.)
    We are now ready to perform the symbolic division by~$t$.
    Indeed, since all variables in~$\vec y$, as well as the term $\big(\sfloor{\frac{\tau(\vec z)}{t^k}} \bmod t\big)$, belong to $[0,t-1]$, 
    we conclude that $\rho' \in \big[{-t\cdot N},t \cdot N]$ where $N \coloneqq \onenorm{\rho'}$.
    Line~\ref{line:elimbounded:guess-r} guesses an integer $r$ from the 
    segment $[{-N},N]$, which stands for the quotient 
    of the division of $\rho'$ by $t$. 
    Each guess corresponds to a disjunct from the following two equivalences:
    \[
    \begin{array}{rcl>{\columncolor{lipicsYellow!25}}ll>{\columncolor{ForestGreen!25}}ll}
        \eta' \cdot t + \rho' = 0 &\iff &\bigvee\nolimits_{r=-N}^{N}\big(&\eta' + r = 0  &\land& r \cdot t = \rho'&\big),\\
        \eta' \cdot t + \rho' < 0 &\iff &\bigvee\nolimits_{r=-N}^{N}\big( 
        &\eta' + r + 1 \leq 0  &\land& 
        t \cdot r \leq \rho' \, \land \, \rho' < t \cdot (r+1) &\big).\\
        &&&\scriptstyle{\text{quotient of the division}}
        &&\scriptstyle{\text{formula $\gamma$ in the pseudocode}}
    \end{array}
    \]
    These equivalences ``perform'' the symbolic division by $t$.
    In line~\ref{line:elimbounded:replace-in-phi},
    the algorithm substitutes the constraint $\eta \sim 0$ 
    with the result of the division, 
    that is, the conjunction of the~\colorbox{lipicsYellow!25}{quotient}
    and the~\colorbox{ForestGreen!25}{remaining part} that is stored the formula $\gamma$ 
    (lines~\ref{line:elimbounded:start-construction-gamma}--\ref{line:elimbounded:increment-r}).
    Observe a key property of~$\gamma$: in it, all the 
    coefficients of the $t$-digits $\vec y$ only have integer coefficients, 
    i.e.,~$\deg(\vec y, \gamma) = 0$.

    Let $\chi_r$ be the formula obtained from $\phi$ by performing the replacement in line~\ref{line:elimbounded:replace-in-phi}.
    This formula belongs to $R_{i+1}$ and, moreover,
    $\phi_{\vec s} \iff \bigvee_{r = -N}^N \chi_r$. 
    We have ${G = \{\vec sr : r \in [-N,N]\}}$, 
    and Item~\ref{lemma:elimbounded:Dtwo:i1} is proved. To prove Item~\ref{lemma:elimbounded:Dtwo:i2}, observe that
    \begin{align*}
        \deg(\vec y,\chi_{r}) &\,=\, \deg(\vec y,\phi_{\vec s}) - \deg(\vec y,\eta \sim 0) + \deg(\vec y,\eta' \sim 0) + \deg(\vec y, \gamma)\\
        &\,=\, \deg(\vec y,\phi_{\vec s}) - \deg(\vec y,\sigma \cdot t \sim 0) + \deg(\vec y,\sigma \sim 0) = \deg(\vec y, \phi_{\vec s}) - 1.   
        &&\qedhere 
    \end{align*}
\end{proof}

\subparagraph*{Elimination of $t$-digits (lines \ref{line:elimbounded:intro-s}--\ref{line:elimbounded:return}).}
By inductively applying \Cref{lemma:elimbounded:Dtwo}, we deduce that the \textbf{while} loop performs at most
$\deg(\vec y, \phi_{\varnothing})$ iterations,
and that the disjunction (over all non-deterministic branches) of formulae~$\phi_{\vec r}$ obtained at the end of this loop 
is equivalent to $\phi_{\varnothing}$.
The constraints in each~$\phi_{\vec r}$ 
are (in)equalities $f(t) +\tau(\vec z) + \sum_{i=1}^{\ell} a_i \cdot y_i \sim 0$, where $a_{1},\dots,a_\ell \in \Z$, $f \in \Z[t]$, all $y_1,\dots,y_\ell$ 
are $t$-digits, and $\tau$ is a non-shifted term of \PPA. 

The last step is to remove the $t$-digits $\vec y$ by appealing 
to~\BoundedQE (line~\ref{line:elimbounded:gauss}).
Recall that this algorithm requires all \PPA constraints in the input to be linear, while terms $\tau(\vec z)$ in $\phi_{\vec r}$ 
may contain (nested) occurrences of the functions $\sfloor{\frac{\cdot}{t^d}}$  and $(\cdot \bmod f(t))$. 
In order to respect this specification, \ElimBounded first replaces each non-shifted term $\tau(\vec z)$ in~$\phi_{\vec r}$ 
with a fresh variable $z'$, storing the substitution $[\tau(\vec z)/z']$ in the map~$S$ (line~\ref{line:elimbounded:z-primes}). 
These terms are restored at the end of the procedure (line~\ref{line:elimbounded:return}). Since the variables~$\vec z'$ 
occur free in the formula in input to~\BoundedQE, and this procedure preserves 
formula equivalence (\Cref{lemma:boundedqe:correctness}), the overall process remains sound.

The formula in input of~\BoundedQE has no divisibility constraints, 
and the eliminated variables~$\vec y$ have integer coefficients. By~\Cref{lemma:boundedqe:int} the output of each non-deterministic branch 
is a formula $\exists\vec w' \leq B' : \psi(\vec w',\vec z')$ where, for every variable $w$ in $\vec w'$, the bound $B'(w)$ is an integer.  
We can thus replace $w$ with a (guessed) integer $g \in [0,B'(w)]$ 
(lines~\ref{line:elimbounded:foreach-w-to-int}--\ref{line:elimbounded:replace-w}). After restoring the terms stored in $S$, 
the resulting formula is quantifier-free and the disjunction over all outputs of~\ElimBounded is equivalent to the input formula. 

\begin{lemma}\label{lemma:elimbounded:correctness}
    \Cref{algo:qe} (\ElimBounded) complies with its specification.
\end{lemma}
This lemma implies~\Cref{lemma:elimbounded}, which was the last missing piece required to complete the proof of~\Cref{theorem:main}. To simplify the complexity arguments in the next section, 
we make use of the two observations in the following lemma, concerning the output of \ElimBounded. 

\begin{restatable}{lemma}{LemmaPropertiesElimBounded}
    \label{lemma:properties-elimbounded}
    In every output of~\ElimBounded, 
    all functions $\intdiv{\cdot}{t^d}$ and $(\cdot \bmod f(t))$ are applied to non-shifted terms. 
    In divisibility relations $(f(t) \divides \cdot)$, the divisor $f(t)$ is an integer. 
\end{restatable}

\section{Solving satisfiability, universality and finiteness}\label{sec:complexity}

Now that we have established that~\PPA admits quantifier elimination, 
let us discuss the decision problems of \emph{satisfiability}, \emph{universality} 
and \emph{finiteness} defined in~\Cref{sec:intro}. 
Without loss of generality, we add to the formula $\phi$ in input to these problems a prefix of existential quantifiers over all its 
free variables; thus assuming that $\phi$ is a sentence. 

Applying our quantifier elimination procedure to the sentence $\phi$ 
results in a variable-free formula $\psi$ whose truth only depends on the value 
taken by the parameter $t$.
Furthermore, from~\Cref{lemma:properties-elimbounded}, 
all occurrences of the functions $\intdiv{\cdot}{t^d}$ and $(\cdot \bmod f(t))$ 
are applied to the constant $0$ (the only variable-free non-shifted term) 
and can thus be replaced with $0$; and all divisibility relations $(f(t) \divides \cdot)$ 
are such that $f(t)$ is an integer. 
That is to say, $\psi$ is a positive Boolean combination of univariate polynomial inequalities 
$g(t) \leq 0$, equalities $g(t) = 0$ and divisibility constraints $d \divides g(t)$, 
where $g \in \Z[t]$ and $d \in \Z$. 

We study the solutions to such a univariate formula $\psi(t)$. 
First, recall that computing the set of all integer roots of a 
polynomial in~$\Z[t]$ can be done in polynomial~time:

\begin{theorem}[\!\!{\cite[Theorem~1]{CuckerKS99}}]
    \label{theorem:cucker}
    There is a polynomial time algorithm 
    that returns the set of integer roots of an input polynomial~$f \in \Z[t]$.
\end{theorem}
This theorem implies that the every integer root of $f \in \Z[t]$
has bit~size polynomial in~$\bitlength{f}$.
Let $r_1 < \dots < r_n$ be the roots of all the polynomials occurring in (in)equalities of $\psi$. 
These roots partition $\Z$ in $2n+1$ regions ${(-\infty,r_1-1],\{r_1\},[r_1+1,r_2-1],\{r_2\},\dots,\{r_n\},[r_n+1,\infty)}$. 
The truth of all (in)equalities in $\psi$ remains invariant for integers within the same region. 
Furthermore, the solutions to the divisibility constraints in $\psi$ are periodic with period 
${p \coloneqq {\lcm\{d : \text{$(d \divides \cdot)$ occurs in $\psi$}\}}} > 0$, 
i.e.,~setting ${t = b}$ satisfies the same divisibility constraints as ${t = b+p}$, for every integer~$b$.
Then, a solution to $\psi$ (if it exists) can be found in the interval $[r_1-p,r_n+p]$. 
Moreover, $\psi$ has infinitely many solutions if and only if one solution lies in intervals $[r_1-p,r_1-1]$ 
or $[r_n+1,r_n+p]$. To recap:

\begin{lemma}\label{lemma:small-solutions}
    Let $\psi$ be a positive Boolean combination of polynomial inequalities $g \leq 0$, equalities $g = 0$ and divisibility constraints $d \divides g$, where $g \in \Z[t]$ and $d \in \Z$. 
    Then, 
    \begin{enumerate}
        \item\label{lemma:small-solutions:i1} If $\psi$ has a solution, then it has one of bit size polynomial in the size of $\psi$.
        \item\label{lemma:small-solutions:i2} If $\psi$ has a polynomial bit size solution that is either larger or smaller than all roots of the polynomials in (in)equalities~of~$\psi$, 
        then $\psi$ has infinitely many solutions, 
        and vice versa.
    \end{enumerate}
\end{lemma}

\subparagraph*{The complexity of the existential fragment.}
\Cref{lemma:small-solutions} implies decidability of all the decision problems of  \emph{satisfiability}, \emph{universality} and \emph{finiteness}. 
We now analyse their complexity for the existential fragment of \PPA, 
establishing~\Cref{theorem:complexity}. 
For simplicity of the exposition, we keep assuming $t \geq 2$. 
Our reasoning can be extended in a straightforward way to all $t \in \Z$, 
following the discussion given at the beginning of~\Cref{sec:FO-procedure}.

First and foremost, we study the complexity of our quantifier elimination procedure.

\begin{lemma}\label{lemma:small-formula}
    \Cref{algo:master} (\Master) runs in non-deterministic polynomial time. 
\end{lemma}

\begin{proof}[Proof idea.]
For the proof we track the evolution of the following parameters as~\Cref{algo:master} executes, where $\phi$ is a formula, and $B$ is a map used for the bounded quantification: 
\begin{itemize}
    \item $\paratom(\phi) \coloneqq \big(\text{number of occurrences of atomic formulae in $\phi$}\big)$,
    \item $\parvars(\phi) \coloneqq \big(\text{number of variables in $\phi$}\big)$,
    \item $\parfunc(\phi) \coloneqq \big(\text{number of occurrences of $\floor{\frac{\cdot}{t^d}}$ and $(\cdot \bmod f(t))$  in an atomic formula of $\phi$}\big)$,
    \item $\parconst(\phi) \coloneqq \max\{ \bitlength{f} : \text{$f  \in \Z[t]$ occurs in $\phi$}\}$,
    \item $\parbound \coloneqq \max\{ 0, \bitlength{B(w)} : \text{$w$ is in the domain of $B$}\}$.
\end{itemize}
For example, for the formula $\gamma \coloneqq (g \divides f_1 \cdot x + \sfloor{\frac{f_2 \cdot x + f_3 \cdot \sfloor{\frac{y}{t}} + f_4}{t}} + f_5)$, 
we have ${\paratom(\gamma) = 1}$, $\parvars(\gamma) = 2$, 
$\parfunc(\gamma) = 2$ (two occurrences of $\floor{\frac{\cdot}{t}}$), 
$\parconst(\gamma) = \max\{\bitlength{g},\bitlength{f_1},\dots,\bitlength{f_5},\bitlength{t}\}$.
The bit size of~$\phi$ is polynomial in the values of these parameters.

We prove (in~Appendices~\ref{subapp:gaussian-elimination:A3}, \ref{subapp:complexity-elimbounded}, and~\ref{app:complexity-master}) 
that throughout the procedure each of these parameters remain 
polynomially bounded with respect to all the parameters of the input formula. 
For instance, during the first two steps of~\ElimBounded, the number of variables in the manipulated formulae~($\phi$ and~$\gamma$) increases at most polynomially in~$\parbound$, due to the \mbox{bit blasting} of the bounded variables. However, by the end of the procedure, it reduces to the number of free variables ---as expected by a quantifier elimination procedure.

While tedious, this proof is not dissimilar to the other complexity analyses of quantifier elimination procedures for~\PA; see, e.g.,~\cite{Weispfenning90}.
Once the evolution of the parameters is known, it is simple to show that~\Cref{algo:master} runs in non-deterministic polynomial time.
\end{proof}

Using this lemma, we can now obtain our results on the decision problems for existential formulae of \PPA.

\begin{proof}[Proof of~\Cref{theorem:complexity}]
    \begin{description}
        \item[Satisfiability.] 
            By~\Cref{lemma:small-formula} and~\Cref{lemma:small-solutions}.\ref{lemma:small-solutions:i1}, 
            if the input sentence $\gamma \coloneqq \exists \vec x : \phi(\vec x)$ is satisfiable (equivalently, valid), 
            then $\sem{\gamma}_k \neq \emptyset$ for some $k$ of bit size polynomial in the size of $\gamma$. 
            Observe that replacing $t$ for $k$ in~$\gamma$ yields an existential sentence of~Presburger arithmetic
            of size polynomial in the size of~$\gamma$. 
            Checking satisfiability for $\exists$\PA is a well-known~\np-complete problem~\cite{vonzurGathenS78}. 
            We conclude that the satisfiability problem for existential formulae of~\PPA is also~\np-complete.
        \item[Universality.] 
            By~\Cref{lemma:small-formula},
            \Cref{algo:master} can be implemented by a deterministic exponential time Turing machine $T$: 
            given an input sentence~$\gamma$, 
            $T$ computes an equivalent disjunction $\psi \coloneqq \bigvee_\beta \psi_\beta$ of exponentially many 
            polynomial-size formulae~$\psi_\beta$. 
            By~\Cref{lemma:properties-elimbounded}, 
            each $\psi_\beta$ is a positive Boolean combination of (in)equalities ${g(t) \sim 0}$ and divisibility constraints~${d \divides g(t)}$, with $g \in \Z[t]$ and ${d \in \N \setminus \{0\}}$. 
            From~\Cref{theorem:cucker}, all roots~$r_1 < \dots < r_n$ of polynomials in (in)equalities of~$\psi$ are of bit size polynomial
            in the size of $\gamma$.  
            However, the period ${p \coloneqq {\lcm\{d : \text{$(d \divides \cdot)$ occurs in $\psi$}\}}}$ may be of exponential bit size. 
            \smallbreak
            
            As a certificate asserting that $\gamma$ is a \emph{negative} instance, 
            we can take $\psi$, the sequence of configurations reached by $T$ when computing $\psi$ from $\gamma$, and a number $k \in [r_1-p,r_n+p]$. 
            The certificate is verified in polynomial time (in its size) by checking that the sequence of configurations is a valid run of~$T$ computing~$\psi$ from~$\gamma$, and that $k$ is not a solution to $\psi(t)$.
            Since the certificate has exponential size, universality is in \conexp. 
            (\conexp-hardness follows from the $\conexp$-hardness of the $\forall \exists^*$ fragment of \PA~\cite{Gradel89,Haase14}.) 
        \item[Finiteness.]   
            Equivalently, we show that the problem of checking whether a sentence~$\gamma$ is satisfiable for \emph{infinitely} many instantiations of $t$ is~\np-complete.
            The proof of~\mbox{\np-hardness} is trivial: for sentences without $t$, this problem 
            is equivalent to the satisfiability problem for $\exists\PA$. 
            For the \np membership, let us consider the formula $\bigvee_\beta \psi_\beta$ 
            computed from an input sentence $\gamma$ via~\Cref{algo:master}. 
            If $\gamma$ is satisfied by infinitely many values of~$t$, 
            then the same holds for least one of the formulae~$\psi_\beta$.
            By~\Cref{lemma:small-solutions}.\ref{lemma:small-solutions:i2}, 
            $\psi_\beta$ has infinitely many solutions if and only if 
            it has a solution~$k$ of bit size polynomial in the size of~$\gamma$, such that~$k$ is either larger or smaller than all roots of polynomials appearing in (in)equalities of $\psi_\beta$. 
            Then, as a certificate asserting that~$\gamma$ has infinitely many solutions 
            we can provide~$\psi_\beta$, the sequence of non-deterministic guesses made by~\Cref{algo:master} to compute $\psi_\beta$ from $\gamma$, and the value~$k$. This 
            certificate can be verified in polynomial time: first, run~\Cref{algo:master} using the provided sequence of guesses, and check that the output is $\psi_\beta$. Then, compute (in polynomial time)
            all roots of the polynomials appearing in the (in)equalities of $\psi_\beta$, 
            and verify that $k$ is a solution to $\psi_{\beta}$ that is either larger or smaller than all of them. \qedhere
    \end{description}
\end{proof}

\bibliography{bibliography}

\clearpage 
\appendix
\section{Extended material for Section~\ref{sec:gaussian-elimination}}
\label{app:gaussian-elimination}

The core of the arguments required for the correctness and complexity of~\BoundedQE stems directly from Appendix B in the full version of~\cite{ChistikovMS24}. 
(The key being that the arguments in that appendix work on any integral domain, not only the integers.) 
Throughout the appendix, let $\exists \vec x : \phi(\vec x, \vec z)$ be the formula in input of~\BoundedQE, 
where $\phi$ is a positive Boolean combination of linear \PPA constraints, $\vec x = (x_1,\dots,x_n)$
and $\vec z = (x_{n+1},\dots,x_g)$.

\subsection{How coefficients evolve throughout the procedure}
\label{subapp:gaussian-elimination:A1}
We start by studying the evolution of the coefficients in the formula $\phi$ throughout the procedure. 
For simplicity, fix a non-deterministic branch~$\beta$ of the algorithm, and let $\exists \vec w_\beta \leq B_\beta : \psi_\beta(\vec w_\beta, \vec z)$ be its output. Observe that the branch $\beta$ can be identified by at most $3n+2$ guesses:
\begin{itemize}
    \item The guess of a set $Z \subseteq \{f(t) : \text{the relation $(f(t) \divides \cdot)$ occurs in $\phi$}\}$ (line~\ref{gauss:set-Z}).
    \item The guess of a sign $\pm$ among $-$ or $+$ (line~\ref{gauss:guess-sign-mt}).
    \item For every $i \in [0,n-1]$, consider the $(i+1)$th iteration of the \textbf{foreach} loop of line~\ref{gauss:mainloop}. 
    This loop performs one guess in line~\ref{gauss:guess-all-zeros}, followed (if the execution reaches line~\ref{gauss:guess-equation}) by two guesses: one for the equality $f_i(t) \cdot x_{i+1} + \tau_i = 0$ (line~\ref{gauss:guess-equation}), and one for the sign $\pm_{i+1}$ among $-$ and $+$ (line~\ref{gauss:guess-sign}).
\end{itemize}
We can thus represent $\beta$ pictorially as follows:
\begin{align*} 
    \phi \xrightarrow{(Z,\pm)} (\phi_{0},m_0,\ell_0,\pm_0,B_0,\chi_0) \xrightarrow{\delta_0} (\phi_{1},m_{1},\ell_{1},\pm_1,B_1,\chi_1) \xrightarrow{\delta_1} \dots {}\\ 
    &\hspace{-4cm}{} \dots \xrightarrow{\delta_{n-1}} (\phi_{n},\chi_n,m_{n},\ell_{n},\pm_n,B_n,\chi_n) \to \psi_\beta.
\end{align*}
Here, each $\delta_i$ is $\top$ if the~$(i+1)$th iteration of the \textbf{foreach} loop executes the~\textbf{continue} statement of line~\ref{gauss:guess-all-zeros}, and otherwise it is the pair ${(f_i(t) \cdot x_{i+1} + \tau_i = 0,\, \pm_{i+1})}$.
The tuple $(\phi_0,m_0,\ell_0,\pm_0,\chi_0,B_0)$ refers to the state of the program after the execution of line~\ref{gauss:introduce-slack} ---hence after inserting the slack variables. In particular:
\begin{itemize}
    \item $\phi_0$ is the formula $\phi$ before the first iteration of the \textbf{foreach} loop of line~\ref{gauss:mainloop}. 
    \item $m_0(t)$ is the polynomial $m(t)$ constructed in line~\ref{gauss:guess-mt}.
    \item $\ell_0(t)$ is the polynomial $\ell(t)$ initialized as $1$ in line~\ref{gauss:bareissfactors}. Following the discussion in~\Cref{sec:gaussian-elimination}, is the current ``common factor'' implementing the optimization of Bareiss's algorithm.
    \item $\pm_0$ is the sign $+$, i.e., the value of the ``sign variable'' $\pm$ from line~\ref{gauss:bareissfactors}. Throughout the procedure this variable matches the sign of the polynomial $\ell(t)$.
    \item $B_0$ is the empty map (line~\ref{gauss:bareissfactors}). Recall that $B$ is the map used for the bounded quantifiers.
    \item $\chi_0$ is the formula $m(t) > 0$ defined in line~\ref{gauss:mt-is-positive}. As explained in~\Cref{sec:gaussian-elimination}, $\chi$ contains all signs guessed by the algorithm.
\end{itemize}
For every $i \in [1,n]$, the tuple $(\phi_i,m_i,\ell_i,\pm_i,B_i,\chi_i)$ contains (in order) the 
formula~$\phi$, the polynomial $m(t)$, the polynomial $\ell(t)$, the value of the variable $\pm$, the map $B$, and the formula $\chi$, that are obtained just \emph{after} the $i$th iteration of the body of the \textbf{foreach} loop of line~\ref{gauss:mainloop}. Some observations: 
\begin{itemize}
    \item If $\delta_{i} = (f_i(t) \cdot x_{i+1} + \tau_i = 0,\, \pm_{i+1})$, then $\ell_{i+1}(t)$ is the polynomial $f_i(t)$ (see line~\ref{gauss:update-factor}), and $(\pm_{i+1}f_i(t)) > 0$ is one of the conjuncts in $\chi_{i+1}$ (see line~\ref{gauss:assert-sign-ft}).
    \item If $\delta_{i-1} = \top$, then $(\phi_{i},m_{i},\ell_{i},\pm_i,\chi_i) = (\phi_{i-1},m_{i-1},\ell_{i-1},\pm_{i-1},\chi_{i+1})$, that is, only the map $B$ is updated.
\end{itemize}

Below, let $\vec y = (y_1,\dots,y_r)$ be the slack variables added in line~\ref{gauss:introduce-slack}.
For simplicity, we work under the following assumption: 
\begin{description}
    \item[($\ast$)]\label{gauss:an-assumption} The variables $x_1,\dots,x_n$ are rearranged in such a way that the $\delta_i$ that are equal to $\top$ all follow those of the form ${(f_i(t) \cdot x_{i+1} + \tau_i = 0,\, \pm_i)}$. That is to say, we assume that, for some $q \in [0,n]$, 
    the \textbf{continue} statement of line~\ref{gauss:guess-all-zeros} is only executed in the last $n-q$ iterations of the \textbf{foreach} loop of line~\ref{gauss:mainloop}. 
\end{description}
Since for the moment we are only interested in how the formula $\phi$ is updated during the procedure, and $\delta_{i-1} = \top$ implies that
$(\phi_{i},m_{i}(t),\ell_{i}(t),\pm_i,\chi_i) = (\phi_{i-1},m_{i-1}(t),\ell_{i-1}(t),\pm_{i-1},\chi_{i-1})$, 
the assumption \textbf{($\ast$)} is without loss of generality.
Therefore, $\delta_i = {(f_i(t) \cdot x_{i+1} + \tau_i = 0,\, \pm_i)}$ 
for every~${i \in [0,q-1]}$, and $\phi_q = \phi_{q+1} = \dots  = \phi_n$. 
Observe also that, then, the variables $x_1,\dots,x_q$ are eliminated via substitution in line~\ref{gauss:vigorous} (opportunely bounding a slack variable if necessary, see line~\ref{gauss:branch}), 
whereas variables $x_{q+1},\dots,x_n$ will be bounded in line~\ref{gauss:subst-rem}.

Suppose that $\phi_0$ has $e$ equalities and $d$ divisibility constraints. 
Our aim is to associate to each formula $\phi_i$ among $\phi_0,\dots,\phi_q$ a matrix $M_i$ with entries in $\Z[t]$. Intuitively, $M_i$ will be storing all the atomic formulae of $\phi_i$ (with repetitions) in some specific order. First, let us note some straightforward observations on the updates performed in lines~\ref{gauss:vigorous}--\ref{gauss:divide}:
\begin{enumerate}
    \item\label{gauss:proof:i1} For every equality $f(t) \cdot x + \tau = 0$ the equality $(f(t) \cdot x + \tau = 0)\vigsub{\frac{-\tau}{f(t)}}{x}$ is $(0 = 0)$.
    \item\label{gauss:proof:i2} There is a one-to-one mapping between the equalities and divisibility constraints in $\phi$ and those in $\phi\vigsub{\frac{-\tau}{f(t)}}{x}$ (i.e., the substitutions are applied locally to each atomic formula).
    \item\label{gauss:proof:i3} The divisions performed in line~\ref{gauss:divide} preserves the one-to-one mapping of the substitutions (they are also applied locally to each atomic formula).
\end{enumerate}
Because of these properties, we can fix an enumeration 
\begin{equation}
    \label{gauss:proof:enum}
    (\rho_1 = 0),\ \dots,\ (\rho_e = 0),\ (h_{e+1}(t) \divides \rho_{e+1}),\ \dots,\ (h_{e+d}(t) \divides \rho_{e+d}),
\end{equation}
of all the atomic formulae of $\phi_0$. The last $d$ elements of the enumeration correspond to the divisibility constraints in $\phi_0$, in any order. The first $e$ elements of the enumeration correspond to the equalities in $\phi_0$, 
picked in such a way that, for every $i \in [1,k]$, $\rho_i = 0$ is in one-to-one correspondence with the equality $f_{i-1}(t) \cdot x_i + \tau_{i-1} = 0$ from $\phi_{i-1}$ 
that is guessed in line~\ref{gauss:guess-equation}. 
In other words, executing the following program on $\rho_i = 0$ yields $f_{i-1}(t) \cdot x_i + \tau_{i-1} = 0$:
\begin{algorithmic}
    \For{$j$ from $1$ to $i-1$}
        \State apply $\vigsub{\frac{-\tau_{j-1}}{f_{j-1}(t)}}{x_j}$ to $\rho_i = 0$
        \State divide $\rho_i = 0$ by the common factor $\ell_{j-1}(t)$.
    \EndFor
\end{algorithmic}
In particular, $\rho_1 = 0$ is equal to $f_0(t) \cdot x_1 + \tau_0 = 0$. 
We denote the enumeration in~\Cref{gauss:proof:enum} by $E^{(0)}$, and write 
$E_j^{(0)}$ for its $j$th entry ($j \in [1,e+d]$). 

Starting from the enumeration in~\Cref{gauss:proof:enum}, we build enumerations~$E^{(1)},\dots,E^{(q)}$ for all the constraints in~$\phi_1,\dots,\phi_q$, perhaps excluding some trivial equalities $0 = 0$.
Let us first explicitly derive $E^{(1)}$ as an example. An enumeration for the atomic formulae of $\phi_1$ can be obtained by applying to each element of the enumeration in~\Cref{gauss:proof:enum} the substitution $\vigsub{\frac{-\tau_0}{f_0(t)}}{x_1}$ (in general we also need to divide by the common factor $\ell_{i-1}(t)$, but for $i = 1$ we have $\ell_{0}(t) = 1$), \underline{except} that the resulting enumeration misses the divisibility $f_0(t) \divides \tau_0$ added in line~\ref{gauss:restore}. 
However, from~\Cref{gauss:proof:i1} above we know that $(\rho_1 = 0)\vigsub{\frac{-\tau_0}{f_0(t)}}{x_1}$ is equal to $(0 = 0)$. 
Hence, instead of adding a constraint to the enumeration, let us simply replace $\rho_1 = 0$ with $f_0(t) \divides \tau_0$.
(This makes perfect sense in relation with Gaussian Elimination: we are using $\rho_1 = 0$ to remove $x_1$ from \emph{all other} constraint; except for $\rho_1 = 0$ itself which we recast in terms of the divisibility~${f_0(t) \divides \tau_0}$.)
The enumeration for $\phi_1$ is then 
\[
    \textstyle(f_0(t) \divides \tau_0),\ E_2^{(0)}\vigsub{\frac{-\tau_0}{f_0(t)}}{x_1},\ \dots,\ E_{e+d}^{(0)}\vigsub{\frac{-\tau_0}{f_0(t)}}{x_1}.
\]
In general, for $i \in [1,q]$, let $E^{(i-1)} \coloneqq (E_1^{(i-1)},\dots,E_{e+d}^{(i-1)})$ be the enumeration associated to~$\phi_{i-1}$.
The enumeration $E^{(i)} \coloneqq (E_1^{(i)},\dots,E_{e+d}^{(i)})$ associated to $\phi_i$ is given by the  program:
\begin{algorithmic}
    \For{$j$ from $1$ to $e+d$, except for $j = i$} 
        \State $E_j^{(i)} \gets E_j^{(i-1)}\vigsub{\frac{-\tau_{j-1}}{f_{j-1}(t)}}{x_j}$
        \State divide $E_j^{(i)}$ by $\ell_{i-1}(t)$
    \EndFor
    $E_i^{(i)} \gets (f_{i-1}(t) \divides \tau_{i-1})$
\end{algorithmic}
The following properties follows from the definition of $E^{(i)}$ and the updates done in lines~\ref{gauss:vigorous}--\ref{gauss:restore}: 
\begin{enumerate}
    \setcounter{enumi}{3}
    \item Every constraint from $\phi_i$, except for trivial constraints $0=0$, occur in $E^{(i)}$.
    \item Each constraint in $E^{(i)}$ occurs in $\phi_i$.
    \item[] (Therefore, $E^{(i)}$ only features variables from $x_{i+1},\dots,x_g$ and $\vec y$.)
    \item The first $i$ entries of $E^{(i)}$, as well as the last $d$ entries, are divisibility constraints. All remaining entries are equalities.
\end{enumerate}

We now move from the ``enumeration view'' of the constraints in $\phi_0,\dots,\phi_q$ to a ``matrix view'' 
from which we will be able to directly appeal to the results in~\cite{ChistikovMS24}.
For $i \in [0,q]$, we associate to $\phi_i$ the matrix $M_i = \begin{bmatrix} G_i \mid A_i \mid \vec c_i \mid S_i \mid D_i \end{bmatrix}$ (with entries over $\Z[t]$), where
\begin{itemize}
    \item $G_i$ is a $(e+ d) \times i$ (rectangular) diagonal matrix having in position $(j,j) \in [1,i] \times [1,i]$ the divisor of the divisibility constraint $E_j^{(i)}$.
    \item $A_i$ is a $(e+ d) \times (g-i)$ matrix. For every $j \in [1,e+d]$ and $k \in [1,g-i]$, the entry of $A_i$ in position $(j,k)$ is the coefficient of the variable $x_{i+k}$ in the constraint $E_j^{(i)}$. 
    \item $\vec c_i$ is a $(e + d)$ column vector. For every $j \in [1,e+d]$, the $j$th entry of $\vec c_i$ correspond to the constant (i.e., the element in $\Z[t]$ that does not appear as a coefficient of any variable) of the term in the constraint~$E_j^{(i)}$. 
    This term is $\tau$ for divisibility constraints $h(t) \divides \tau$.
    \item $S_i$ is a $(e + d) \times r$ matrix. For every $j \in [1,e+d]$ and $k \in [1,r]$, the entry of $S_i$ in position $(j,k)$ is the coefficient of the slack variable $y_k$ in the constraint $E_j^{(i)}$. 
    \item $D_i$ is a $(e + d) \times d$ matrix. Its first $e$ rows are all zero. The last $d$ rows form a diagonal matrix: the entry in position $(e+j,j)$ with $j \in [1,d]$ is the divisor of the divisibility constraint~$E_{e+j}^{(i)}$.
\end{itemize}
The matrices $M_0,\dots,M_q$ are just another way of encoding the constraints in the enumerations $E^{(0)},\dots,E^{(q)}$. 
The key observation from~\cite{ChistikovMS24} is that $M_1,\dots,M_q$ also corresponds to the sequence of matrices obtained 
by performing Bareiss algorithm on the first $q$ columns of the matrix, diagonalizing them (equivalently, removing $x_1,\dots,x_q$).
Following~\cite{Bareiss68}, \mbox{\cite[Appendix B]{ChistikovMS24}} characterizes the entries of $M_1,\dots,M_q$ from the entries of $M_0$. 
This in turn characterizes all polynomials in $\Z[t]$ occurring in the formulae $\phi_1,\dots,\phi_q$.

\begin{proposition}[{\!\!\cite[Appendix B]{ChistikovMS24}}]
    \label{prop:gauss-fundamental}
    For every $i \in [1,q]$, the matrix $M_i$ satisfies: 
    \begin{enumerate}[I.]
        \item\label{prop:gauss-fundamental:i1} For every row $j \in [1,i]$ and column $k \in [1,g+1+r+d]$, the entry in position $(j,k)$ of $M_i$ is equal to the determinant of the matrix $M_0[1,\dots,i;\ 1,\dots,j-1,k,j+1,\dots,i]$.
        \item[] (Hence, $G_i$ is a diagonal matrix whose non-zero entries are all $\det(M_0[1,\dots,i; 1,\dots,i])$.)

        \item\label{prop:gauss-fundamental:i2} For row $j \in [i+1,e+d]$ and column $k \in [1,i]$, the entry in position $(j,k)$ of $M_i$ is $0$.
        \item\label{prop:gauss-fundamental:i3} For every row $j \in [i+1,e+d]$ and column $k \in [i+1,g+1+r+d]$, the entry in position $(j,k)$ of $M_i$ is equal to the determinant of the matrix $M_0[1,\dots,i,j;\ 1,\dots,i,k]$.
    \end{enumerate}
    Moreover: 
    \begin{enumerate}[I.]
        \setcounter{enumi}{3}
        \item\label{prop:gauss-fundamental:i4} The coefficient $f_{i}(t)$ of $x_{i+1}$ in the equality guessed in line~\ref{gauss:guess-equation} during the $(i+1)$th iteration of the \textbf{foreach} loop of line~\ref{gauss:mainloop} 
        is $\det(M_0[1,\dots,i+1; 1,\dots,i+1])$.
        \item\label{prop:gauss-fundamental:i5} The factor used for the division of line~\ref{gauss:divide} during the $(i+1)$th iteration of the \textbf{foreach} loop of line~\ref{gauss:mainloop} is~$f_{i-1}(t)$ 
        (postulating $f_{-1}(t) \coloneqq 1$). These divisions are without remainder.
    \end{enumerate}
\end{proposition}

\begin{proof}[Where to find the proof.]
    Item~(a) and Item~(b) in Lemma 9 of the full version of~\cite{ChistikovMS24} 
    imply
    \Cref{prop:gauss-fundamental:i3}, \Cref{prop:gauss-fundamental:i4}, 
    and the fact that factor used for the division of line~\ref{gauss:divide} during the $(i+1)$th iteration of the \textbf{foreach} loop of line~\ref{gauss:mainloop} is~$f_{i-1}(t)$.
    The proofs of Item~(a) and Item~(b) follow~\cite{Bareiss68}, which iteratively uses the Desnanot--Jacobi identity.
    Lemma 9 also shows that the divisions are all without remainder for all constraints in the formulae $\phi_0,\dots,\phi_q$, except for the divisibility constraints added in line~\ref{gauss:restore} (see Item~(c) of Lemma 9). 
    Under the assumption that the divisibilities performed in line~\ref{gauss:divide} are also without remainder for the added divisibility constraints, Item~(a) from Lemma~14 of~\cite{ChistikovMS24} establishes \Cref{prop:gauss-fundamental:i1}. The assumption is ultimately discharged in Lemma~15 of~\cite{ChistikovMS24}: all divisions performed in line~\ref{gauss:restore} are without remainder (which completes the proof of~\Cref{prop:gauss-fundamental:i5}).
    Lastly, \Cref{prop:gauss-fundamental:i2} follows directly from the definition of~$G_i$, which is a $(e+d) \times i$ diagonal matrix.

    A note: for the most part, it is straightforward to see that the aforementioned lemmas in~\cite[Appendix B]{ChistikovMS24} work on any integral domain (a non-zero commutative ring in which the product of non-zero elements is non-zero). The only non-trivial step is given in the proof of~\cite[Lemma~14]{ChistikovMS24}. 
    There, the proof requires taking an inverse $N^{-1}$ of some integer matrix~$N$ with non-zero determinant. Recall that a matrix is invertible over the integers only if the determinant is $\pm 1$. However, because we only know $\det N \neq 0$, this step of the proof considers matrices over the rationals~$\Q$ instead, so that $N^{-1}$ is guaranteed to exist. (This is not a problem, because in any case the proof of~\cite[Lemma~14]{ChistikovMS24} then continues by showing that the entries of~$M_i$ are indeed the ring elements given in the statement of~\Cref{prop:gauss-fundamental}.)
    For an arbitrary integral domain (such as $\Z[t]$), instead of $\Q$ we can use its \emph{field of fractions}, that is the smallest field in which the integral domain can be embedded.
    Such a field exists for any integral domain.
\end{proof}

\Cref{prop:gauss-fundamental} will be very useful for both the proofs of correctness and complexity of the algorithm. We will also need the following property on the coefficients of the slack variables.

\begin{lemma}
    \label{lemma:same-slack-everywhere}
    For all $i \in [0,q]$, all slack variables in $\phi_i$ that are not assigned a value by $B_i$
    occur in $\phi_i$ with identical coefficients, namely $f_{i-1}(t)$ (setting $f_{-1}(t) \coloneqq 1$).
    Moreover, each of these slack variables occur in a single constraint of $\phi_i$, it is an equality.
\end{lemma}

\begin{proof}
    See Lemmas~11 and~12 from~\cite{ChistikovMS24} a fully detailed proof. 
    Briefly, the second statement of the lemma is simple to show: 
    it is clearly true for $\phi_0$, and in all substitutions~$\vigsub{\frac{-\tau}{f(t)}}{x}$ performed by the algorithm, the term $\tau$ only feature slack variables (if any) for which the map $B$ already assigns a value. Since substitutions are local to each atomic formula, it is thus not possible
    for an ``unassigned'' slack variable~$y$ to end up in multiple constraints. 
    Then, in the column~$k$ of $M_i$ corresponding to the variable~$y$, there is at most one non-zero coefficient, and it is located in a row~$j$ among the first $e$ ($y$ occurs in an equality). More precisely, because of how the matrix $S_i$ is defined, $j > i$ ($y$ is unassigned). 
    Then, we are in the case of~\Cref{prop:gauss-fundamental}.\ref{prop:gauss-fundamental:i3}, 
    and the coefficient of $y$ in $\phi_i$ is $\det (M_0[1,\dots,i,j;\ 1,\dots,i,k])$. 
    The last column of $M_0[1,\dots,i,j;\ 1,\dots,i,k]$ is $(0,\dots,0,1)$, 
    and so this determinant is equal to $\det (M_0[1,\dots,i;\ 1,\dots,i])$. 
    By~\Cref{prop:gauss-fundamental}.\ref{prop:gauss-fundamental:i4}, this is equal to $f_{i-1}(t)$.
\end{proof}

\LemmaBoundedQEInt* 

\begin{proof}
    We look at the polynomial $m(t)$ initially defined in line~\ref{gauss:guess-mt} 
    and used to define the range of all maps~$B_\beta$. 
    In any run of the algorithm, this polynomial is built from polynomials appearing 
    as divisors of divisibility constraints in~$\phi$ (which are here assumed to be integers), 
    together with polynomials $f_i(t)$ considered in~\Cref{prop:gauss-fundamental:i4} of~\Cref{prop:gauss-fundamental}. 
    Since in $\phi$ all coefficients of variables in $\vec x$ are integers, 
    from~\Cref{prop:gauss-fundamental} and by definition of the matrix~$M_q$ 
    we conclude that all $f_i(t)$ are integers. 
    Therefore, $m(t)$ is an integer throughout the procedure, and the statement follows.
\end{proof}

\subsection{Correctness of~\BoundedQE (proof of Lemma~\ref{lemma:boundedqe:correctness})}
\label{subapp:gaussian-elimination:A2}

Recall that the input of~\BoundedQE is a $\exists \vec x : \phi(\vec x, \vec z)$, 
where $\phi$ is a positive Boolean combination of linear \PPA constraints, $\vec x = (x_1,\dots,x_n)$
and $\vec z = (x_{n+1},\dots,x_g)$. 

For the proof of correctness, we would like to proceed as in~\Cref{subapp:gaussian-elimination:A1} and fix a non-deterministic branch~$\beta$ of the algorithm, which produces a path of the form:
\begin{align*} 
    \phi \xrightarrow{(Z,\pm)} (\phi_{0},m_0,\ell_0,\pm_0,B_0,\chi_0) \xrightarrow{\delta_0} (\phi_{1},m_{1},\ell_{1},\pm_1,B_1,\chi_1) \xrightarrow{\delta_1} \dots{}\\ 
    &\hspace{-4cm} \dots{} \xrightarrow{\delta_{n-1}} (\phi_{n},m_{n},\ell_{n},\pm_n,B_n,\chi_n) \to \psi_\beta
\end{align*}
where the various objects in it are as defined in~\Cref{subapp:gaussian-elimination:A1}. 
Correctness is however a global property that ranges all non-deterministic executions, 
and we will thus need to consider the whole computation tree of the procedure. 
We will still use the path above as a reference (it will soon be clear in what sense).
An important detail: in this appendix, we are \underline{not} working under the assumption~\textbf{($\ast$)}
used to derive the bounds on $\phi$.

We divide the proof of correctness into three steps:
\begin{itemize}
    \item Step 1 involves the first edge of paths as the one above, concerning the guesses of $(Z,\pm)$.
    \item Step 2 concerns the iterations of the \textbf{foreach} loop of line~\ref{gauss:mainloop}, 
    involving the guesses~$\delta_i$.
    \item Step 3 concerns all that happens after the \textbf{foreach} loop of line~\ref{gauss:mainloop} completes. 
    \item[] (In this step the algorithm is deterministic.)
\end{itemize}

\subparagraph*{Step 1 (lines~\ref{gauss:set-Z}--\ref{gauss:introduce-slack}).}
Let us focus on the first edge of the path, that is, 
\[ \phi \xrightarrow{(Z,\pm)} (\phi_{0},m_0,\ell_0,\pm_0,B_0,\chi_0).\]
We refer the reader to~\Cref{subapp:gaussian-elimination:A1} 
for the description of $(\phi_{0},m_0,\ell_0,\pm_0,B_0,\chi_0)$ (\Cref{lemma:corrgauss:straightforward} below recalls properties of this tuple that suffice for this appendix). 
In the guess $(Z,\pm)$, recall that $Z$ is a subset of $E \coloneqq \{f(t) : \text{the relation $(f(t) \divides \cdot)$ occurs in $\phi$}\}$, and $\pm$ is the sign used to define $m_0(t) \coloneqq \pm \prod_{g \in E \setminus Z} g$. The formula $\phi_0$ already contains all slack variables $\vec y = (y_1,\dots,y_r)$ added in line~\ref{gauss:introduce-slack} (note that the number of slack variables does not depend on the guessed $Z$ and $\pm$, because the number of inequalities in $\phi$ is the same across non-deterministic branches).

Let us write $S_0$ for the set of all tuples $(\phi_0,m_0(t),\ell_0(t),\pm_0,B_0)$ obtained 
by considering all the different choices for the pair $(Z,\pm)$. 
We start with some simple observations:
\begin{lemma}\label{lemma:corrgauss:straightforward}
    for every $(\phi_0,m_0,\ell_0,\pm_0,B_0,\chi_0) \in S_0$:
    \begin{enumerate}
        \item\label{lemma:corrgauss:straightforward:i1} $\phi_0$ is a positive Boolean combination linear of equalities and divisibility constraints. 
        \item\label{lemma:corrgauss:straightforward:i2} The formula $\chi_0$ is the inequality $m_0(t) > 0$.
        \item\label{lemma:corrgauss:straightforward:i3} For every divisibility $(f(t) \divides \cdot)$ in $\phi_0$, the divisor $f(t)$ is a factor of $m_0(t)$.
        \item\label{lemma:corrgauss:straightforward:i4} The polynomial $\ell_0(t)$ is the constant $1$, and
        $\pm_0$ is the symbol $+$. 
        \item\label{lemma:corrgauss:straightforward:i5} The map $B_0$ is empty.
    \end{enumerate}
\end{lemma}

\begin{proof}
    The proof is straightforward (see lines~\ref{gauss:guess-mt}--\ref{gauss:introduce-slack}, or the discussion in~\Cref{subapp:gaussian-elimination:A1}). 
\end{proof}

The following lemma express the key ``equivalence'' between $\phi$ and the elements in $S_0$:

\begin{lemma}\label{lemma:corrgauss:first-step} The formula $\phi$ is equivalent to $\bigvee\nolimits_{(\phi_0,m_0,\ell_0,\pm_0,B_0,\chi_0) \in S_0} \vec y \in \N : \phi_0 \land \chi_0.$
\end{lemma}

\begin{proof}
    Below, given $Z \subseteq E$, let us write $\phi_Z$ for the formula obtained from $\phi$ by replacing every divisibility $f(t) \divides \tau$, where $f \in Z$, with $\tau = 0$. 
    We also write $\phi_Z'$ for the formula obtained from $\phi_Z$ by replacing each inequality $\tau \leq 0$ 
    with $\tau + y = 0$, where $y$ is a slack variable from~$\vec y$ (distinct inequalities receive distinct slack variables). 
    Following lines~\ref{gauss:set-Z}--\ref{gauss:introduce-slack}, 
    we see that for every $(\phi_0,m_0,\ell_0,\pm_0,B_0,\chi_0) \in S_0$, 
    the formula $\phi_0$ is of the form 
    $(\phi_Z' \land \bigwedge_{f \in Z} f(t) = 0)$, 
    for some $Z \subseteq E$, and moreover $m_0(t)$ is equal to $\pm(\prod_{g \in E \setminus Z} g(t)) > 0$, 
    for some sign $\pm$ among $+$ and $-$. 
    The following chain of equivalences then proves the lemma:
    \begingroup 
    \allowdisplaybreaks
    \begin{align*} 
        & \phi\\ 
        \iff{}& \textstyle\bigvee_{Z \subseteq E} \Big(\phi_Z \land \bigwedge_{f \in Z} f(t) = 0 \land \bigwedge_{g \in E \setminus Z} g(t) \neq 0\Big)\\
        &&\hspace{-3.4cm}\text{(by definition~of divisibility constraint)}\\
        \iff{}& \textstyle\bigvee_{Z \subseteq E} \bigvee_{\pm \in \{+,-\}} \Big(\phi_Z \land \bigwedge_{f \in Z} f(t) = 0 \land (\prod_{g \in E \setminus Z} g(t)) \neq 0 \Big)\\
        &&\hspace{-3.4cm}\text{(since $a \neq 0 \land b \neq 0 \iff a \cdot b \neq 0$)}\\
        \iff{}& \textstyle\bigvee_{Z \subseteq E} \bigvee_{\pm \in \{+,-\}} \Big(\phi_Z \land \bigwedge_{f \in Z} f(t) = 0 \land m(t) > 0 \Big)\\
        &&\hspace{-7cm}\text{(guessing a sign for $\textstyle\prod_{g \in E \setminus Z} g(t)$, and by definition of $m(t)$)}\\
        \iff{}& \textstyle\bigvee_{Z \subseteq E} \bigvee_{\pm \in \{+,-\}} \Big((\exists \vec y \in \N : \phi_Z') \land \bigwedge_{f \in Z} f(t) = 0 \land m(t) > 0 \Big)\\
        &&\hspace{-8cm}\text{(because $\tau \leq 0 \iff \exists y \in \N : \tau + y = 0$, and moreover $\phi_Z$ is a}\\ 
        &&\hspace{-8cm}\text{positive Boolean combination of linear \PPA constraints, hence}\\ 
        &&\hspace{-8cm}\text{we can push $\exists y$ outside the scope of all its Boolean connectives)}\\ 
        \iff{}&\textstyle\bigvee_{(\phi_0,m_0,\ell_0,\pm_0,B_0,\chi_0) \in S_0} \vec y \in \N : \phi_0 \land \chi_0
        &\hspace{-8cm}\text{(by def.~of~$S_0$, $\chi_0$, and $\phi_0$)}\hspace{0.5cm}~
        \qedhere
    \end{align*}
    \endgroup
\end{proof}

\subparagraph*{Step 2 (lines~\ref{gauss:mainloop}--\ref{gauss:endmainloop}).}
Let us write $S_i$ for the set of all tuples $(\phi_i,m_i,\ell_i,\pm_i,B_i,\chi_i)$ obtained across all non-deterministic branches 
of~\Cref{algo:gaussianqe}, after the \textbf{foreach} loop of line~\ref{gauss:mainloop} has been executed $i$ times. Note that $S_0$ coincides with the homonymous set defined before~\Cref{lemma:corrgauss:straightforward}.
Recall that this loop executes $n$ times, considering the variables $x_1,\dots,x_n$ (in this order). 

Let $i \in [0,n-1]$. Back to our ``path view'', we will now look at edges of the form  
\[
    (\phi_{i},m_i,\ell_i,\pm_i,B_i,\chi_i) \xrightarrow{\delta} (\phi_{i+1},m_{i+1},\ell_{i+1},\pm_{i+1},B_{i+1},\chi_{i+1}),
\]
where:
\begin{itemize}
    \item the tuple $(\phi_{i},m_i,\ell_i,\pm_i,B_i,\chi_i)$ belongs to $S_i$, 
    \item  $\delta$ is either $\top$ (this case occurs when the~$(i+1)$th iteration of the \textbf{foreach} loop executes the \textbf{continue} statement of line~\ref{gauss:guess-all-zeros}) or a pair ${(f_i(t) \cdot x_{i+1} + \tau_i = 0,\, \pm_{i+1})}$ (corresponding to the two guesses of lines~\ref{gauss:guess-equation} and~\ref{gauss:guess-sign}), and 
    \item $(\phi_{i+1},m_{i+1},\ell_{i+1},\pm_{i+1},B_{i+1},\chi_{i+1}) \in S_{i+1}$. 
\end{itemize}

Let $i \in [0,n]$, and consider $\Phi \coloneqq (\phi_{i},m_i,\ell_i,\pm_i,B_i,\chi_i) \in S_i$.
Let us denote by $\vec y'$ the slack variables from $\vec y$ that are \emph{unassigned} in $B_i$ 
(that is, not in its domain), and $\vec v$ for all the variables that are assigned in $B_i$.
We write $F(\Phi)$ for the formula: 
\[ 
    F(\Phi) \coloneqq \exists \vec y' \in \N\, \exists \vec v \leq B\, (\phi_i \land \chi_i).
\]
Observe that then the equivalence in~\Cref{lemma:corrgauss:first-step} can be states as 
$\phi \iff \bigvee_{\Phi \in S_0} F(\Phi)$, 
since in this case $B_0$ is the empty map 
(\Cref{lemma:corrgauss:straightforward}).

The goal for this step is to prove the following result: 
\begin{restatable}{lemma}{LemmaCorrGaussKeySecondStep}
    \label{lemma:corrgauss:key-second-step}
    Let $i \in [0,n-1]$ and $\Phi \in S_i$. There is a subset $G \subseteq S_{i+1}$ 
    satisfying 
    \[ \exists x_{i+1} F(\Phi) \iff \bigvee\nolimits_{\Psi \in G} F(\Psi).\]
\end{restatable}

As done in Step 1, let us start with some simple facts:
\begin{lemma}\label{lemma:corrgauss:straightforward:second}
    Let $i \in [0,n]$ and consider a tuple $(\phi_{i},m_i,\ell_i,\pm_i,B_i,\chi_i) \in S_i$. We have:
    \begin{enumerate}
        \item\label{lemma:corrgauss:straightforward:second:i1} $\phi_i$ is a positive Boolean combination of linear equalities and divisibility constraints.
        \item\label{lemma:corrgauss:straightforward:second:i2} For every divisibility $(g(t) \divides \cdot)$ in $\phi_i$, the divisor $g(t)$ is a factor of $m_i(t)$.
        \item\label{lemma:corrgauss:straightforward:second:i3} 
            The formula $\chi_i$ is a conjunction of inequalities of the form $h(t) > 0$, with $h \in \Z[t]$. 
        \item\label{lemma:corrgauss:straightforward:second:i4}  
            The formula $\chi_i$ entails both $m_i(t) > 0$ and $\pm_i \ell_i(t) > 0$.
    \end{enumerate}
\end{lemma}

\begin{proof}
    The proof is by induction on $i$.
    \begin{description}
        \item[base case: $i  = 0$.] Directly from~\Cref{lemma:corrgauss:straightforward}.
        \item[induction step: $i \geq 1$.] Consider $(\phi_{i},m_i,\ell_i,\pm_i,B_i,\chi_i) \in S_i$, as well as the label $\delta$ of an edge 
        such that there is $(\phi_{i-1},m_{i-1},\ell_{i-1},\pm_{i-1},B_{i-1},\chi_{i-1}) \in S_{i-1}$ 
        satisfying  
        \[
            (\phi_{i-1},m_{i-1},\ell_{i-1},\pm_{i-1},B_{i-1},\chi_{i-1}) \xrightarrow{\delta} (\phi_{i},m_i,\ell_i,\pm_i,B_i,\chi_i).
        \]
        By induction hypothesis, the tuple $(\phi_{i-1},m_{i-1},\ell_{i-1},\pm_{i-1},B_{i-1},\chi_{i-1})$ satisfies Items~\ref{lemma:corrgauss:straightforward:second:i1}--\ref{lemma:corrgauss:straightforward:second:i4}. 
        We proceed by cases on $\delta$:
        \begin{description}
            \item[case: $\delta = \top$.] We have $(\phi_{i},m_{i},\ell_{i},\pm_i,\chi_i) = (\phi_{i-1},m_{i-1},\ell_{i-1},\pm_{i-1},\chi_{i-1})$ (because the \textbf{foreach} loop executes the~\textbf{continue} statement of line~\ref{gauss:guess-all-zeros}). Items~\ref{lemma:corrgauss:straightforward:second:i1}--\ref{lemma:corrgauss:straightforward:second:i4} are thus statisfied.
            \item[case: $(f(t) \cdot x_{i} + \tau = 0,\, \pm)$.]
                Following lines~\ref{gauss:update-factor}, \ref{gauss:assert-sign-ft} and~\ref{gauss:new-factor}, we have:
                \begin{align*}
                    \ell_i &\coloneqq f(t), & m_i & \coloneqq \pm f(t) \cdot m_{i-1}(t),\\
                    \pm_i &\coloneqq \pm, & \chi_i &\coloneqq (\chi_{i-1} \land \pm f(t) > 0).
                \end{align*} 
                \begin{description}
                    \item[proof of~\Cref{lemma:corrgauss:straightforward:second:i1}:]  
                    By induction hypothesis, $\phi_{i-1}$ is a positive Boolean combination of linear equalities and divisibility constraints. The formula $\phi_{i}$ is obtained from $\phi_{i-1}$ by applying lines~\ref{gauss:vigorous}--\ref{gauss:restore}, which preserve this property. 
                    \item[proof of~\Cref{lemma:corrgauss:straightforward:second:i2}:] By induction hypothesis, for every divisibility $(g(t) \divides \cdot)$ in $\phi_{i-1}$, the divisor $g(t)$ is a factor of $m_{i-1}(t)$. 
                    Line~\ref{gauss:vigorous} multiplies each divisor $g(t)$ of $\phi_{i-1}$ by~$f(t)$. 
                    The resulting polynomial $g(t) \cdot f(t)$
                    is a factor of $m_i(t) = \pm f(t) \cdot m_{i-1}(t)$. 
                    Line~\ref{gauss:divide} divides $g(t) \cdot f(t)$ by the common factor $\ell_{i-1}(t)$. From~\Cref{prop:gauss-fundamental}.\ref{prop:gauss-fundamental:i5} these polynomial divisions are without remainder, and therefore the result is still a factor of $m_i(t)$. 
                    Lastly, line~\ref{gauss:restore} adds the divisibility $(f(t) \divides \tau)$; whose divisor is again a factor of $m_i(t) = \pm f(t) \cdot m_{i-1}(t)$.
                    \item[proof of~\Cref{lemma:corrgauss:straightforward:second:i3}:] 
                        Follows by definition of $\chi_i$ and the fact that $\chi_{i-1}$ satisfies~\Cref{lemma:corrgauss:straightforward:second:i3}. 
                    \item[proof of~\Cref{lemma:corrgauss:straightforward:second:i4}:] Directly from the its definition, $\chi_i$ implies both $\pm_i \ell_i > 0$ 
                    and $m_i > 0$ (the latter because $\chi_{i-1}$ implies $m_{i-1} > 0$ by induction hypothesis).
                    \qedhere 
                \end{description}
        \end{description}
    \end{description}
\end{proof}

Let us consider a tuple $(\phi_{i},m_i,\ell_i,\pm_{i},B_i,\chi_{i}) \in S_i$. 
For the time being, it is more convenient to reason purely in terms of formulae equivalences, neglecting their connections with the operations performed by~\Cref{algo:gaussianqe}. We will come back to~\Cref{algo:gaussianqe} later. 

We rely on the following notion:
\begin{definition}[Partial formulae of $\phi_i$]
    \label{definition:partial-subformulae}
    Let $I,A_0$ and $A_1$ be a partition of the set $T \coloneqq \{(f(t),\tau) : \text{the equation $f(t) \cdot x_{i+1} + \tau = 0$ occurs in $\phi_i$}\}$.  
    Let us write $(\phi_i|_{I,T})$ for the formula obtained from $\phi_i$ 
    by replacing with $\bot$ every equation $f(t) \cdot x_{i+1} + \tau = 0$ such that $(f(t),\tau) \in I$, and with $\top$ every equation $f(t) \cdot x_{i+1} + \tau = 0$ such that $(f(t),\tau) \in T \setminus I$. 

    The following formula is said to be a \emph{partial formula} of $\phi_i$:  
    \[ 
        (\phi_i|_{I,T}) \land \bigwedge_{(f(t),\tau) \in A_0} (f(t) = 0 \land \tau = 0) \land \bigwedge_{(f(t),\tau) \in A_1} (f(t) \neq 0 \land f(t) \cdot x_{i+1} + \tau = 0).
    \] 
    We denote by $\mathcal{P}(\phi_i)$ the set of \emph{all partial formulae} of~$\phi_i$.
    Given $\gamma \in \mathcal{P}(\phi_i)$, we write $I(\gamma)$, $A_0(\gamma)$ and $A_1(\gamma)$ 
    for the sets $I$, $A_0$ and $A_1$ used to define $\gamma$, respectively. 
\end{definition}
The idea behind the notion of partial formulae is the following. 
In any solution~$\nu$ to $\phi_i$, some atomic formulae are satisfied (below, ``active'') and some are not (below, ``inactive'').
(We recall in passing that $\nu$ is a map assigning an integer to each variable in $\vec x$, $\vec y$, and $\vec z$, 
and to the parameter $t$, such that the integers assigned to the slack variables $\vec y$ are non-negative.)  
The formula $(\phi_i|_{I,T})$ in~\Cref{definition:partial-subformulae} guesses the ``activity'' 
of each equality featuring~$x_{i+1}$; replacing inactive equalities with $\bot$ and active equalities with $\top$.
Among the active equalities $f(t) \cdot x_{i+1} + \tau = 0$ (with respect to $\nu$), 
we distinguish two types, depending on whether $f(\nu(t)) = 0$. 
In~\Cref{definition:partial-subformulae}, the first type corresponds to the set $A_0$: 
here $f(t)$ is guessed to be $0$, and so $f(t) \cdot x_{i+1} + \tau = 0$ 
becomes equivalent to $f(t) = 0 \land \tau = 0$. 
The second type corresponds to the set $A_1$: 
here $f(t)$ is guessed non-zero, and therefore $f(t) \cdot x_{i+1} + \tau = 0$ 
becomes (trivially) equivalent to $f(t) \neq 0 \land f(t) \cdot x_{i+1} + \tau = 0$. 

The next two lemmas show that no information is lost when moving from $\phi_i$ to $\mathcal{P}(\phi_i)$.

\begin{lemma}
    \label{lemma:corrgauss:gamma-contains-phi}
    Every formula $\gamma \in \mathcal{P}(\phi_i)$ implies $\phi_i$.
\end{lemma}

\begin{proof}
    Let $\gamma \in \mathcal{P}(\phi_i)$, and $I \coloneqq I(\gamma)$, $A_0 \coloneqq A_0(\gamma)$, $A_1 \coloneqq A_1(\gamma)$ and $T \coloneqq I \cup A_0 \cup A_1$.
    Following~\Cref{definition:partial-subformulae}, 
    the formula $(\phi_i|_{I,T})$ has the same Boolean structure as $\phi_i$, the only difference is that 
    all equalities from $T$ are replaced with $\bot$ or $\top$. 
    Moreover, since $\phi_i$ is a positive Boolean combination, both $\phi_i$ and $(\phi_i|_{I,T})$ only contain the connectives $\lor$ and $\land$. Let us then refer to subformulae of $\phi_i$ and $(\phi_i|_{I,T})$ 
    by using finite words from the alphabet $\{L,R\}$, where $L$ selects the left disjunct/conjunct of a Boolean connective, and $R$ selects the right one. 
    Given a word~$w \in \{L,R\}^*$, we write $[\phi_i]_w$
    for the subformula of $\phi_i$ corresponding to $w$ (and proceed similarly for $(\phi_i|_{I,T})$).
    For example, $\big[(a \lor b) \land ((c \lor d) \land e)  \big]_{RL} = c \lor d$.
    Therefore, $[\phi_i]_\epsilon = \phi_i$ and $[(\phi_i|_{I,T})]_\epsilon = (\phi_i|_{I,T})$, where $\epsilon$ stands for the empty word.
    We leave $[\phi_i]_w$ undefined whenever $w$ does not lead to a subformula of $\phi_i$ 
    (e.g., $\big[(a \lor b) \land ((c \lor d) \land e)  \big]_{LLL}$ is undefined), 
    and write $\mathcal{W}$ for the set of words $w \in \{L,R\}^*$ for which $[\phi_i]_w$ is defined. 
    This also corresponds to the set of words $w \in \{L,R\}^*$ for which $[(\phi_i|_{I,T})]_w$ is defined.

    Let $\nu$ be a solution to $\gamma$.
    We prove that for every $w \in \mathcal{W}$, if $\nu$ satisfies $[(\phi_i|_{I,T})]_w$ then it also satisfies $[\phi_i]_w$. Since $\nu$ satisfies $(\phi_i|_{I,T})$ (as it satisfies $\gamma$), this suffices to conclude that $\nu$ satisfies $\phi_i$, as required. The proof is by induction on the words of $\mathcal{W}$, 
    with induction hypothesis asserting that the statement holds for all strictly longer words in this set. 
    
    \begin{description}
        \item[base case: $w \in \mathcal{W}$ is not a strict prefix of some word in $\mathcal{W}$.]
           In this case, $\alpha \coloneqq [\phi_i]_w$ and $\beta \coloneqq [(\phi_i|_{I,T})]_w$ are atomic formulae. 
           We split the proof into four cases: 
           \begin{description}
                \item[case: $\alpha$ is not an equation $f(t) \cdot x_{i+1} + \tau = 0$.] 
                We have $\alpha = \beta$ (by def.~of~$(\phi_i|_{I,T})$), 
                and therefore the statement trivially follows.
                \item[case: $\alpha$ is an equation $f(t) \cdot x_{i+1} + \tau = 0$ and $(f(t),\tau) \in I$.]
                We have $\beta = \bot$ (by definition of $(\phi_i|_{I,T})$). 
                Hence, $\nu$ does not satisfy $\beta$, and the statement follows.
                \item[case: $\alpha$ is an equation $f(t) \cdot x_{i+1} + \tau = 0$ and $(f(t),\tau) \in A_0$.]
                By def.~of $(\phi_i|_{I,T})$ we have $\beta = \top$.
                Since $\nu$ satisfies $\gamma$, it also satisfies $f(t) = 0 \land \tau = 0$ 
                (by definition of~$\gamma$). This conjunction entails $f(t) \cdot x_{i+1} + \tau = 0$; and so $\nu$ satisfies $\alpha$. 
                \item[case: $\alpha$ is an equation $f(t) \cdot x_{i+1} + \tau = 0$ and $(f(t),\tau) \in A_1$.]
                Since $\nu$ satisfies $\gamma$, it also satisfies $f(t) \neq 0 \land f(t) \cdot x_{i+1} + \tau = 0$ (by definition of $\gamma$). Hence, $\nu$ satisfies $\alpha$. 
           \end{description}
        \item[induction step: $w \in \mathcal{W}$ is a strict prefix of some word in $\mathcal{W}$.]
           From the definition of $\mathcal{W}$, we have $[\phi_i]_w = [\phi_i]_{wL} \oplus [\phi_i]_{wR}$ and $[(\phi_i|_{I,T})]_w = [(\phi_i|_{I,T})]_{wL} \oplus [(\phi_i|_{I,T})]_{wR}$, where $\oplus \in \{\land,\lor\}$.
           We split the proof into two cases, depending on $\oplus$. 
           \begin{description}
            \item[case: $\oplus$ is $\land$.] Suppose that $\nu$ satisfies $[(\phi_i|_{I,T})]_{wL} \land [(\phi_i|_{I,T})]_{wR}$. By induction hypothesis, $\nu$ satisfies $[\phi_i]_{wL}$ and $[\phi_{i}]_{wR}$. 
            That is, $\nu$ satisfies $[\phi_i]_{wL} \land [\phi_{i}]_{wR}$. 
            \item[case: $\oplus$ is $\lor$.] Suppose that $\nu$ satisfies $[(\phi_i|_{I,T})]_{wL} \lor [(\phi_i|_{I,T})]_{wR}$. There is $D \in \{L,R\}$ such that $\nu$ satisfies $[(\phi_i|_{I,T})]_{wD}$.
            By induction hypothesis, $\nu$ satisfies $[\phi_i]_{wD}$. 
            Therefore, $\nu$ satisfies $[\phi_i]_{wL} \lor [\phi_{i}]_{wR}$. 
            \qedhere
           \end{description}
    \end{description}
\end{proof}

\begin{lemma}\label{lemma:corrgauss:second-step}
    The formula
    $\phi_i$
    implies $\bigvee_{\gamma \in \mathcal{P}(\phi_i)} \gamma$.
\end{lemma}

\begin{proof}
    Consider a solution $\nu$ to $\phi_i$. 
    Let $T$ be defined as in~\Cref{definition:partial-subformulae} 
    as the set $\{(f(t),\tau) : \text{the equation $f(t) \cdot x_{i+1} + \tau = 0$ occurs in $\phi_i$}\}$.  
    Let $I$ and $A_0$ and $A_1$ be the following partition of $T$, induced by $\nu$:
    \begin{align*} 
        I &\coloneqq \{ (f(t),\tau) \in T : \text{$\nu$ does not satisfy~$f(t) \cdot x_{i+1} + \tau = 0$}\}\\
        A_0 &\coloneqq \{ (f(t),\tau) \in T : \text{$\nu$ satisfies~$f(t) \cdot x_{i+1} + \tau = 0 \land f(t) = 0$}\}\\
        A_1 &\coloneqq \{ (f(t),\tau) \in T : \text{$\nu$ satisfies~$f(t) \cdot x_{i+1} + \tau = 0 \land f(t) \neq 0$}\}.
    \end{align*}
    From~\Cref{definition:partial-subformulae}, 
    it is simple to see that $\nu$ satisfies the partial formula~$\gamma$ of $\phi_i$
    such that $I(\gamma) = I$, $A_0(\gamma) = A_0$ and $A_1(\gamma) = A_1$.
    Indeed, the fact that $\nu$ satisfies $(\phi_i|_{I,T})$ follows from the fact that 
    $(\phi_i|_{I,T})$ is obtained from $\phi_i$ by replacing each equality featuring $x_{i+1}$ 
    with $\top$ (when $\nu$ satisfies the equality) or $\bot$ (when $\nu$ does not satisfy the equality). Moreover, for every $(f(t),\tau) \in A_0$,  $\nu$ satisfies~$f(t) \cdot x_{i+1} + \tau = 0 \land f(t) = 0$, which in turn means that $\nu$ satisfies $\tau = 0 \land f(t) = 0$. For every $(f(t),\tau) \in A_1$, $\nu$ trivially satisfies $f(t) \cdot x_{i+1} + \tau = 0 \land f(t) \neq 0$.
\end{proof}

We now discuss three cases, depending on the set $A_1(\gamma)$. 
These cases will later corresponds to different branches in~\Cref{algo:gaussianqe}, 
from which we will prove~\Cref{lemma:corrgauss:key-second-step}.

Here is the first case:

\begin{lemma}
    \label{lemma:corrgauss:third-step:a}
    Let $\gamma \in \mathcal{P}(\phi_i)$ such that $A_1(\gamma) = \emptyset$.
    Then, 
    \begin{center}
    $\exists x_{i+1} (\gamma \land \chi_i)$  \ implies \ $\exists x_{i+1} (0 \leq x_{i+1} \leq m_i(t) - 1 \land \gamma \land \chi_i)$.
    \end{center}
\end{lemma}

\begin{proof}
    Since $A_1(\gamma) = \emptyset$, in $\gamma$ the variable $x_{i+1}$ only occurs in divisibility constraints.
    Consider then a solution $\nu$ to $\gamma \land \chi_i$. Let $g_1(t),\dots,g_k(t)$ be the list of divisors occurring in divisibility constraints of $\phi_i$. These are also all the divisors occurring in~$\gamma$ (by definition of~$\gamma$). 
    Each divisor~$g_j$ ($j \in [1,k]$) is evaluated to an integer $g_j(\nu(t))$. 
    From~\Cref{lemma:corrgauss:straightforward:second}:
    \begin{itemize}
        \item $\chi_i$ implies $m(t) > 0$, and the variable $x_{i+1}$ does not occur in this formula.
        \item $g_j(t)$ is a factor of $m_i(t)$.
    \end{itemize} 
    Therefore, $m(\nu(t))$ is positive, and it is divided by $g_j(\nu(t))$
    (in particular, $g_j(\nu(t))$ is non-zero).
    Let $L \coloneqq \lcm\{g_j(\nu(t)) : j \in [1,k]\}$ (recall: $\lcm$ returns a non-negative integer). 
    All numbers of the form $\nu(x_{i+1}) + j \cdot L$, for $j \in \Z$, have the same remainder 
    modulo every $g_j(\nu(t))$. 
    Since $x_{i+1}$ only occurs in divisibility constraints, updating $\nu(t)$ to any of these values still yields a solution to $\gamma \land \chi_i$. In particular, one such value must lie in $[0,L-1]$.
    Then, to conclude the proof, it suffices to observe that $L \leq m(\nu(t))$, 
    as $L$ divides $m(\nu(t)) > 0$.
\end{proof}

Let us denote 
by $\vec y'$ the slack variables from $\vec y$ that are unassigned in $B_i$.
Given ${\gamma \in \mathcal{P}(\phi_i)}$, we ``divide'' the set $A_1(\gamma)$ into two sets $A_{1}^{\text{eq}}(\gamma)$ and $A_1^{\text{slk}}(\gamma)$: 
\begin{align*}
    A_{1}^{\text{eq}}(\gamma) &\coloneqq \{ (f(t),\tau) \in A_1(\gamma) : \text{$\tau$ does not contain a slack variable from $\vec y'$}\},\\
    A_1^{\text{slk}}(\gamma) &\coloneqq \{(\sigma, y) : \text{$\sigma = (f(t),\tau) \in A_1(\gamma) \setminus A_{1}^{\text{eq}}(\gamma)$, and $y$ occurs in $\vec y'$ and $\tau$}\}.
\end{align*}

For the second of the three cases, let us consider when $A_1^{\text{eq}}(\gamma) \neq \emptyset$. 
Let $(f(t),\tau) \in A_1^{\text{eq}}(\gamma)$.
Any solution to $\gamma$ satisfies ${f(t) \cdot x_{i+1} + \tau = 0 \land f(t) \neq 0}$, and this equality does not feature any unassigned slack variable. There is then little to do, 
we can simply substitute $x_{i+1}$ for $\frac{-\tau}{f(t)}$, eliminating this variable without producing any new bounded quantifier (more on that later). We are left with the third case (the most interesting one):

\begin{lemma}
    \label{lemma:corrgauss:third-step:b}
    Let $\gamma \in \mathcal{P}(\phi_i)$ with $A_1^{\text{slk}}(\gamma) \neq \emptyset$ and $A_{1}^{\text{eq}}(\gamma) = \emptyset$. Then, 
    \[ 
        \exists \vec y' \in \N\, \exists x_{i+1} (\gamma \land \chi_i)  \ \text{ implies } \ 
            \hspace{-10pt}\bigvee_{\substack{((f(t),\tau),y) \in A_{1}^{\text{slk}}(\gamma)\\ \pm \in \{+,-\}}} 
            \hspace{-10pt}\exists \vec y' \in \N\, \exists x_{i+1}
            (0 \leq y \leq \pm f(t) \cdot m_i(t)-1 \land \gamma \land \chi_i).
    \]
\end{lemma}

\begin{proof}
    We follow arguments in~\cite[Lemma 17]{ChistikovMS24}.
    Let $A_1^{\text{slk}}(\gamma) = \{((g_1,\tau_1),y_1'),\dots,((g_q,\tau_1),y_q')\}$.
    Let $\nu$ be a solution to $\gamma \land \chi_i$, which in particular assigns an integer to $x_{i+1}$ and non-negative integers to each of the slack variables $y_1',\dots,y_q'$. 
    Observe that each integer $g_j(\nu(t))$ is non-zero, because $(g_j,\tau_j)$ is an element of $A_1(\gamma)$.

    Consider the auxiliary rational vector $\vec v^{\nu} = (v_1,\dots,v_q)$ where $v_j \coloneqq \frac{\nu(y_j')}{\abs{g_j(\nu(t))}}$. Each $v_j$ is positive.
    Suppose now that the values given by $\nu$ to the variables $x_{i+1},y_1',\dots,y_q'$ are those 
    that \emph{minimize} the smallest possible component of $\vec v^{\nu}$. Without loss of generality, assume $v_1$ to be such component. In other words, no matter how we change the values given to $x_{i+1},y_1',\dots,y_q'$ in the solution $\nu$, if the resulting map~$\xi$ is still a solution, then the minimal component in its auxiliary rational vector~$\vec v^{\xi}$ is greater or equal than $v_1$.

    We show that $v_1 < m_i(\nu(t))$ or, in other words, $\nu(y_1') \leq \abs{g_1(\nu(t))} \cdot m_i(\nu(t))-1$, proving the lemma. 
    (Recall that $\chi_i$ implies $m(t) > 0$ by~\Cref{lemma:corrgauss:straightforward:second}, hence $m_i(\nu(t))$ is positive).
    
    For the sake of contradiction, suppose $v_1 \geq m_i(\nu(t))$; and note that this implies $v_j \geq m_i(\nu(t))$ for all $j \in [1,q]$.
    Let $\mu(t)$ be the coefficient of the slack variable $y_1'$ in~$\tau_1$. 
    From~\Cref{lemma:same-slack-everywhere}, for every $j \in [1,q]$, $\tau_j$ is the only term in which $y_j'$ occurs, and moreover its coefficient is $\mu(t)$. 
    We also know that $\mu(\nu(t))$ cannot be zero. Indeed, from~\Cref{lemma:same-slack-everywhere},
    $\mu(t)$ is the polynomial $f(t)$ from line~\ref{gauss:guess-equation}, with respect to the last time 
    this line was executed by the~\textbf{foreach} loop of line~\ref{gauss:mainloop} (with $f(t) = 1$ if no such previous execution occurred). Let us say that this happened during the $k$th iteration of the loop, with $k \leq i$. Then, $\chi_k$ contains the constraint $\pm_k \mu(t) > 0$.
    Iterations between the $k$th and before the (current) $(i+1)$th one do not change the formula $\chi$ (the \textbf{continue} statement of line~\ref{gauss:guess-all-zeros} is executed instead), hence $\chi_{i}$ still contains the constraint $\pm_k \mu(t) > 0$, which forces $\mu(t)$ to be non-zero.
    
    Let us denote by $\xi$ the assignment that agrees with $\nu$ on all possible variables except $x_{i+1}$ and $y_1',\dots,y_q'$, and such that:
    \begin{align*}
        \xi(x_{i+1}) &= \nu(x_{i+1}) \pm \mu(\nu(t)) \cdot m_i(\nu(t)),\\
        \xi(y_j') &= \nu(y_j') \mp g_j(\nu(t)) \cdot m_j(\nu(t)) & \text{for every } j \in [1,q],
    \end{align*}
    where $\mp$ stands for $-$ whenever $g_1(\nu(t)) > 0$, and for $+$ otherwise, whereas $\pm \in \{+,-\}$ is the sign opposite to $\mp$.
    Clearly, $\frac{\xi(y_1')}{\abs{g_1(\nu(t))}} < v_1$, since $\xi(y_1')$ is obtained by adding a negative number to $\nu(y_1')$.
    Moreover, for every $j \in [1,q]$, we have $\xi(y_j') \geq 0$; 
    because $v_j \geq m_i(\nu(t))$ and we have (in the worse case) only removed $\abs{g_j(\nu(t))} \cdot m_j(\nu(t))$ from~$\nu(y_j')$. 
    
    To conclude the proof it suffices to show that $\xi$ still satisfies $\gamma \land \chi_i$. 
    Indeed, in this way $\frac{\xi(y_1')}{\abs{g_1(\nu(t))}} < v_1$ contradicts the fact that $v_1$ is the minimal possible value that entries in the auxiliary vector~$\vec v^{\nu}$ can take as we change the values of $x_{i+1},y_1',\dots,y_q'$ (while preserving the satisfaction of $\gamma \land \chi_i$).

    Clearly, the map $\xi$ satisfies all equalities and divisibility constraints not featuring $x_{i+1}$ (hence, in particular, it satisfies~$\chi_i$). 
    Similarly, it satisfies also all divisibility constraints featuring $x_{i+1}$, 
    since by~\Cref{lemma:corrgauss:straightforward:second} we have that every divisor in these constraints is a factor of $m_i(t)$, 
    and $\xi(x_{i+1})$ has been defined by shifting $\nu(x_{i+1})$ by a multiple of $m_i(\nu(t))$. 
    Lastly, the equalities in $\gamma$ involving $x_{i+1}$ are still all satisfies.
    Indeed, given $j \in [1,q]$, let us write $\tau_j = \rho_j + \mu(t) \cdot y_j'$. We have:
    \begin{align*}
        &g_j(\xi(t)) \cdot \xi(x_{i+1}) + \rho_j + \mu(\xi(t)) \cdot \xi(y_j')\\
        ={}&g_j(\nu(t)) \cdot \big(\nu(x_{i+1}) \pm \mu(\nu(t)) \cdot m_i(\nu(t))\big) + \rho_j + \mu(\nu(t)) \cdot \big(\nu(y_j') \mp g_j(\nu(t)) \cdot m_i(\nu(t))\big)\\
        ={}&g_j(\nu(t)) \cdot \nu(x_{i+1}) + \rho_j + \mu(\nu(t)) \cdot \xi(y_j)\\
        ={}& 0&\qedhere
    \end{align*}
\end{proof}

We now combine \Cref{lemma:corrgauss:gamma-contains-phi,lemma:corrgauss:second-step,lemma:corrgauss:third-step:a,lemma:corrgauss:third-step:b}:

\begin{lemma}
    \label{lemma:corrgauss:step-two:main-lemma}
    The formula 
    $\exists \vec y' \in \N \, \exists x_{i+1}(\phi_i \land \chi_i)$ is equivalent to 
    \begin{align*} 
        &\Big(\exists \vec y' \in \N \, \exists x_{i+1} \big(0 \leq x \leq m_i(t) - 1 \land \phi_i  \land \chi_i\big) \Big) \lor{} \\
        \bigvee_{\substack{\gamma \in \mathcal{P}(\phi_i)\\ \pm \in \{+,-\}}}
        &\Big(
        \bigvee_{(f(t),\tau) \in A_{1}^{\text{eq}}(\gamma)}
        \hspace{-0.2cm}\Big(\exists \vec y' \in \N \, \exists x_{i+1} \big(f(t) \cdot x_{i+1} + \tau = 0 \land \pm f(t) > 0 \land \phi_i \land \chi_i \big) \Big) \lor{} \\
        &\hspace{-1.8cm}\bigvee_{((f(t),\tau),y) \in A_{1}^{\text{slk}}(\gamma)}
        \hspace{-0.4cm}\Big(\exists \vec y' \in \N \, \exists x_{i+1} \big(f(t) \cdot x_{i+1} + \tau = 0 \land 0 \leq y \leq \pm f(t) \cdot m_i(t)-1 \land \phi_i \land \chi_i \big) \Big)\Big).
    \end{align*}
\end{lemma}

\begin{proof}
    The right to left direction is trivial, since every disjunct ``specializes'' the formula 
    $\exists \vec y' \in \N \, \exists x_{i+1} (\phi_i \land \chi_i)$ with further constraints.
    For the left to right direction, by~\Cref{lemma:corrgauss:second-step} 
    we have that $\exists \vec y' \in \N \, \exists x_{i+1}  (\phi_i \land \chi_i)$ implies $\bigvee_{\gamma \in \mathcal{P}(\phi_i)} \exists \vec y' \in \N \, \exists x_{i+1} (\gamma \land \chi_i)$. 
    Let us then fix $\gamma \in \mathcal{P}(\phi_i)$, 
    and show that $\exists \vec y' \in \N \, \exists x_{i+1} (\gamma \land \chi_i)$ implies one of the disjunct of the right-hand side of the equivalence. We again consider three cases separately, depending on the set $A_1(\gamma)$. 
    \begin{description}
        \item[case: $A_1(\gamma) = \emptyset$.] From~\Cref{lemma:corrgauss:third-step:a}, the formula
        $\exists \vec y' \in \N \, \exists x_{i+1} (\gamma \land \chi_i)$ 
        implies the formula ${\exists \vec y' \in \N \, \exists x_{i+1} (0 \leq x \leq m_i(t) - 1 \land \gamma \land \chi_i)}$. 
        From~\Cref{lemma:corrgauss:gamma-contains-phi},
        the latter formula implies
        \[{\exists \vec y' \in \N \, \exists x_{i+1} (0 \leq x \leq m_i(t) - 1 \land \phi_i \land \chi_i)}.\]
        This is one of the disjuncts on the right-hand side of the equivalence (see first line).
        \item[case: $A_1^{\text{eq}}(\gamma) \neq \emptyset$.] In this case, one of the conjuncts occurring in the formula $\gamma$ is a formula ${f(t) \cdot x_{i+1} + \tau = 0} \land f(t) \neq 0$
        such that $(f(t),\tau) \in A_1^{\text{eq}}(\gamma)$. 
        Since $\gamma$ implies $\phi_i$ (\Cref{lemma:corrgauss:gamma-contains-phi}), 
        $\exists \vec y' \in \N \, \exists x_{i+1} (\gamma \land \chi_i)$ 
        implies $\exists \vec y' \in \N \, \exists x_{i+1} (f(t) \cdot x_{i+1} + \tau = 0 \land f(t) \neq 0 \land \phi_i \land \chi_i)$. By ``guessing'' the sign of $f(t)$, we can rewrite the latter formula as 
        \[ 
            \bigvee_{\pm \in \{+,-\}} \exists \vec y' \in \N \, \exists x_{i+1} (f(t) \cdot x_{i+1} + \tau = 0 \land \pm f(t) > 0 \land \phi_i \land \chi_i).
        \]
        This is one of the disjunct on the right-hand side of the equivalence (see second line).
        \item[case: $A_1^{\text{eq}}(\gamma) = \emptyset$ and $A_1^{\text{slk}}(\gamma) \neq \emptyset$.] 
        From~\Cref{lemma:corrgauss:third-step:b},
        $\exists \vec y' \in \N \, \exists x_{i+1} (\gamma \land \chi_i)$  
        implies
        \[
            \bigvee_{\substack{((f(t),\tau),y) \in A_{1}^{\text{slk}}(\gamma)\\ \pm \in \{+,-\}}} 
            \hspace{-5pt}\exists \vec y' \in \N\, \exists x_{i+1}
            (0 \leq y \leq \pm f(t) \cdot m_i(t)-1 \land \gamma \land \chi_i).
        \]
        In the above formula, we can replace $\gamma$ by $\phi_i$ (as the former implies the latter by~\Cref{lemma:corrgauss:gamma-contains-phi}), to obtain a formula that is still implied by $\exists \vec y' \in \N \, \exists x_{i+1} (\gamma \land \chi_i)$. 
        Each disjunct in the resulting formula appears on the right-hand side of the equivalence (see third line).
        \qedhere
    \end{description}
\end{proof}

The next lemma completes the elimination of $x_{i+1}$ from the second and third case of the formula in~\Cref{lemma:corrgauss:step-two:main-lemma}. 

\begin{lemma}
    \label{lemma:corrgauss:remove-variable}
    Consider a formula $\psi \coloneqq \exists x_{i+1} (f(t) \cdot x_{i+1} + \tau = 0 \land \pm f(t) > 0 \land \phi_i \land \chi_i)$, 
    where $f(t) \cdot x_{i+1} + \tau = 0$ is an equation in $\phi_i$ and $\pm \in \{+,-\}$. 
    Let $\phi_{i+1}$ be the formula obtained from $\phi_i$ by executing lines~\ref{gauss:vigorous}--\ref{gauss:restore} of~\Cref{algo:gaussianqe}, where the division is done with respect to $\ell_i$. 
    Then, $\psi$ is equivalent to $\pm f(t) > 0 \land \phi_{i+1} \land \chi_{i}$.
\end{lemma}

\begin{proof}
    By definition of the vigorous substitution, it is easy to see that 
    $\psi$ is equivalent to $\exists x_{i+1} (f(t) \cdot x_{i+1} + \tau = 0 \land \pm f(t) > 0 \land \phi_i\vigsub{\frac{-\tau}{f(t)}}{x_{i+1}} \land \chi_i)$.
    As explained in~\Cref{subapp:gaussian-elimination:A1}, 
    all polynomials from~$\Z[t]$ occurring in inequalities and divisibility constraints (divisors included) of $\phi_i\vigsub{\frac{-\tau}{f(t)}}{x_{i+1}}$ are divisible by the polynomial $\ell_i$ (without remainder).
    Line~\ref{gauss:divide} performs this division. The resulting formula $\phi_i'$ is equivalent to $\phi_i\vigsub{\frac{-\tau}{f(t)}}{x_{i+1}}$. 
    Lastly, observe that $x_{i+1}$ occurs now only in the equality $f(t) \cdot x_{i+1} + \tau = 0$, 
    which we can rewrite as $f(t) \divides \tau$, since $x_{i+1}$ ranges over the integers. 
    The formula $\phi_{i+1}$ is defined as $\phi_i' \land f(t) \divides \tau$ (see line~\ref{gauss:restore}).
    Therefore, $\psi$ is equivalent to $\pm f(t) > 0 \land \phi_{i+1} \land \chi_{i}$.
\end{proof}

We are now ready to prove~\Cref{lemma:corrgauss:key-second-step}:

\LemmaCorrGaussKeySecondStep*

\begin{proof}
    Let $\Phi = (\phi_{i},m_i,\ell_i,\pm_{i},B_i,\chi_{i}) \in S_i$, $\vec v$ be the variables assigned in $B_i$, and $\vec y'$ be the slack variables that are unassigned in $B_i$. 
    By~\Cref{lemma:corrgauss:step-two:main-lemma}, $\exists x_{i+1} F(\Phi)$ is equivalent to 
    \begin{center}
    \scalebox{0.9}{
        $\begin{aligned} 
            &\Big(\exists \vec v \leq B_i\,\exists \vec y' \in \N \, \exists x_{i+1} \big(0 \leq x \leq m_i(t) - 1 \land \phi_i  \land \chi_i\big) \Big) \lor{} \\
            \bigvee_{\substack{\gamma \in \mathcal{P}(\phi_i)\\ \pm \in \{+,-\}}}
            &\Big(
            \bigvee_{(f(t),\tau) \in A_{1}^{\text{eq}}(\gamma)}
            \hspace{-0.2cm}\Big(\exists \vec v \leq B_i\,\exists \vec y' \in \N \, \exists x_{i+1} \big(f(t) \cdot x_{i+1} + \tau = 0 \land \pm f(t) > 0 \land \phi_i \land \chi_i \big) \Big) \lor{} \\
            &\hspace{-1.8cm}\bigvee_{((f(t),\tau),y) \in A_{1}^{\text{slk}}(\gamma)}
            \hspace{-0.4cm}\Big(\exists \vec v \leq B_i\,\exists \vec y' \in \N \, \exists x_{i+1} \big(f(t) \cdot x_{i+1} + \tau = 0 \land 0 \leq y \leq \pm f(t) \cdot m_i(t)-1 \land \phi_i \land \chi_i \big) \Big)\Big).
        \end{aligned}$
    }
    \end{center}
    We now consider each disjunction~$\psi$ of this formula, and show that $S_{i+1}$ contains an element $\Psi$ 
    such that $F(\Psi)$ is equivalent to~$\psi$. Naturally, the element $\Psi$ will be such that $\Phi \xrightarrow{\delta} \Psi$, for a suitable choice of the guess $\delta$. 
    Once more, we divide the proof in three cases. 
    \begin{description}
        \item[case 1:] Consider the formula $\psi \coloneqq \exists \vec v \leq B_i\,\exists \vec y' \in \N \, \exists x_{i+1} \big(0 \leq x \leq m_i(t) - 1 \land \phi_i  \land \chi_i\big)$. Let $\delta = \top$, that is, consider the case where the \textbf{foreach} loop executes the \textbf{continue} statement of line~\ref{gauss:guess-all-zeros}. Starting from $\Psi$, in this case the algorithm simply adds $x_{i+1}$ to the keys of $B_i$, with value $m(t)-1$. Let $B_{i+1}$ be the resulting map. This completes the iteration of the \textbf{foreach} loop, 
        and therefore $S_{i+1}$ contains the tuple $\Psi \coloneqq (\phi_{i},m_i,\ell_i,\pm_{i},B_{i+1},\chi_{i})$. Clearly, $F(\Psi)$ is equivalent to $\psi$.
        \item[case 2:] Consider the formula 
        \[ 
            \psi \coloneqq \exists \vec v \leq B_i\,\exists \vec y' \in \N \, \exists x_{i+1} \big(f(t) \cdot x_{i+1} + \tau = 0 \land \pm f(t) > 0 \land \phi_i \land \chi_i \big),  
        \]
        where $(f(t),\tau) \in A_1^{\text{eq}}(\gamma)$, $\gamma \in \mathcal{P}(\phi_i)$ and $\pm \in \{+,-\}$.
        Since  $(f(t),\tau) \in A_1^{\text{eq}}(\gamma)$, the formula $\phi_i$ contains the equality $f(t) \cdot x_{i+1} + \tau$, and $\tau$ does not feature any unassigned slack variable.
        Let $\delta = (f(t) \cdot x_{i+1} + \tau = 0, \pm)$, that is, the equation and sign~$\pm$ are guessed in lines~\ref{gauss:guess-equation} and~\ref{gauss:guess-sign}. Note that line~\ref{gauss:assert-sign-ft} updates $\chi$, so that $\chi_{i+1} = \chi_i \land (\pm f(t) > 0)$.
        By~\Cref{lemma:corrgauss:remove-variable}, 
        lines~\ref{gauss:vigorous}--\ref{gauss:restore} updates $\phi_i$ into a formula $\phi_{i+1}$
        such that $\exists x_{i+1} \big(f(t) \cdot x_{i+1} + \tau = 0 \land \phi_i \land \chi_{i+1} \big)$ is equivalent to $(\phi_{i+1} \land \chi_{i+1})$.
        Furthermore, the algorithm gives $m_{i+1} \coloneqq \pm f(t) \cdot m_{i}$, $\pm_{i+1} \coloneq \pm$ 
        and $\ell_{i+1} = f(t)$. The tuple $\Psi \coloneqq (\phi_{i+1},m_{i+1},\ell_{i+1},\pm_{i+1},B_i,\chi_{i+1})$ belongs to $S_{i+1}$, 
        and $F(\Psi) = \exists \vec v \leq B_i\,\exists \vec y' \in \N \,( \phi_{i+1} \land \chi_{i+1})$ is equivalent to $\psi$. 

        \item[case 3:] This case considers formulae 
        \[ 
            \exists \vec v \leq B_i\,\exists \vec y' \in \N \, \exists x_{i+1} \big(f(t) \cdot x_{i+1} + \tau = 0 \land 0 \leq y \leq \pm f(t) \cdot m_i(t)-1 \land \phi_i \land \chi_i \big).
        \]
        Observe that the subformula $0 \leq y \leq \pm f(t) \cdot m_i(t)-1 \land \chi_i$ implies $\pm f(t) > 0$ (recall that $\chi_i$ implies $m_i(t) > 0$). After adding this implies inequality, this case is analogous to the previous one, the only difference being that the map $B_i$ is now updated 
        with a bound $\pm f(t) \cdot m_i(t)-1$ for the slack variable $y$.
        \qedhere
    \end{description}
\end{proof}

\subparagraph*{Step 3 (lines~\ref{gauss:unslackloop}--\ref{gauss:return}).}
The last step of the procedure is deterministic and simply removes the remaining slack variables. 
Here is its ``specification'':

\begin{lemma}
    \label{lemma:corrgauss:step-three}
    Let $(\phi_n,m_n,\ell_n,\pm_n,B_n,\chi_n) \in S_n$, 
    and let $\vec y'$ be the slack variables that are unassigned in $B_n$.
    Let $\psi$ be the formula obtained by applying the \textbf{foreach} loop of line~\ref{gauss:unslackloop} 
    to the formula $\phi_n$, with respect to $B_n$ and $\pm_n$. 
    Then, $\exists \vec y' \in \N (\phi_n \land \chi_n) \iff \psi \land \chi_n$.
\end{lemma}

\begin{proof}[Proof idea.]
    The proof follows essentially verbatim the proof of~\cite[Lemma 18]{ChistikovMS24}
    and it is thus not duplicated here. We give instead the intuition on how~\cite[Lemma 18]{ChistikovMS24} is proven, 
    by considering the case of $n = 1$. The full proof given in~\cite{ChistikovMS24} iterates the same reasoning across all formulae~$\phi_2,\dots,\phi_n$.

    In the case $n = 1$, the \textbf{foreach} loop of line~\ref{gauss:mainloop} iterates only once. 
    Let us suppose that, during this iteration, the procedure guesses an equation 
    $f(t) \cdot x_1 + \tau  = 0$ and $\pm \in \{+,-\}$ (whether $\tau$ has a slack variable or not is irrelevant here, so let us assume it does not for simplicity). 
    The formula $\chi$ contains the constraint $\pm f(t) > 0$ asserting that $\pm$ is the sign of $f(t)$.
    Let $g(t) \cdot x_1 + \rho + y = 0$ be another equation from $\phi$, featuring the unassigned slack variable~$y$.
    ($y$ has coefficient $1$ because we are performing the first iteration of the loop.)

    The substitution $\vigsub{\frac{-\tau}{f(t)}}{x}$ 
    updates $g(t) \cdot x_1 + \rho + y = 0$ 
    to ${-g(t) \cdot \tau + f(t) \cdot \rho + f(t) \cdot y = 0}$. 
    Line~\ref{gauss:drop-slack} modifies this term 
    to ${-g(t) \cdot \tau + f(t) \cdot \rho \sim 0}$, 
    where $\sim$ stands for $\leq$ whenever $\pm$ is the symbol $+$, and otherwise $\sim$ stands for $\geq$.

    First, notice that $\exists y \in \N \,({-g(t) \cdot \tau + f(t) \cdot \rho + f(t) \cdot y = 0} \land \pm f(t) > 0)$ 
    implies the inequality ${-g(t) \cdot \tau + f(t) \cdot \rho \sim 0}$. 
    Indeed, since $y$ is non-negative, the equality is forcing $-g(t) \cdot \tau + f(t) \cdot \rho$ to either be zero (when $y$ is zero) or to have a sign opposite to the one of~$f(t)$.
    The left to right direction of the lemma follows from this reason alone.

    The right-to-left direction is more subtle, 
    because ${-g(t) \cdot \tau + f(t) \cdot \rho \sim 0} \land \pm f(t) > 0$ alone 
    does not imply 
    $\exists y \in \N : {-g(t) \cdot \tau + f(t) \cdot \rho + f(t) \cdot y = 0}$.
    For this direction to hold, it is sufficient (and necessary) to further assume that $-g(t) \cdot \tau + f(t) \cdot \rho$ is a multiple of $f(t)$; which in turns implies that $-g(t) \cdot \tau$ has to be a multiple of $f(t)$. But indeed, the procedure adds in line~\ref{gauss:restore} 
    the division $f(t) \divides \tau$, which is preserved by the updates performed in Line~\ref{gauss:drop-slack}.
    This is the key observation used to prove the right to left direction of the lemma.
\end{proof}

\subparagraph*{Putting all together.}
We can now complete the proof of correctness: 

\LemmaBoundedQECorrectness*

\begin{proof}
    Let  $\exists x_1,\dots,x_n \phi(x_1,\dots,x_n, \vec z)$ be the input formula.
    By~\Cref{lemma:corrgauss:first-step}, 
    across non-deterministic branches, lines~\ref{gauss:set-Z}--\ref{gauss:introduce-slack} 
    build the formulae $F(\Phi)$, with $\Phi \in S_0$, 
    such that 
    \[ 
        \exists x_1,\dots,x_n \phi \iff \bigvee_{\Phi \in S_0}\,\exists x_1,\dots,x_n F(\Phi).
    \]
    By iterated application of~\Cref{lemma:corrgauss:key-second-step}, after the~\textbf{foreach} loop of line~\ref{gauss:mainloop} completes, (again across non-deterministic branches), the algorithm 
    considers the formulae $F(\Psi)$ with $\Psi \in S_n$, and we have
    \[ 
        \exists x_1,\dots,x_n \phi \iff \bigvee_{\Psi \in S_n} F(\Psi).
    \]
    Given $\Psi = (\phi_n,m_n,\ell_n,\pm_n,B_n,\chi_n) \in S_n$, let $F(\Psi) = \exists \vec y' \in \N : \exists \vec v \leq B_n (\phi_n \land \chi_n)$, 
    where $\vec y'$ and $\vec v$ are, respectively, the unassigned and assigned variables in $B_n$.
    By~\Cref{lemma:corrgauss:step-three}, lines~\ref{gauss:unslackloop}--\ref{gauss:return} 
    produce a formula $F'(\Psi) = \exists \vec v \leq B_n (\phi_n' \land \chi_n)$ equivalent to $F(\Psi)$.
    The output of the algorithm is $\bigvee_{\Psi \in S_n} F'(\Psi)$, 
    which is thus equivalent to $\exists x_1,\dots,x_n \phi$.
\end{proof}

\subsection{Complexity of~\BoundedQE} 
\label{subapp:gaussian-elimination:A3}

We now move to the complexity of the procedure. 
Recall the paremeters of \PPA formulae, which were introduced in \Cref{sec:complexity}.   

\begin{itemize}
    \item $\paratom(\phi) \coloneqq \big(\text{number of occurrences of atomic formulae in $\phi$}\big)$,
    \item $\parvars(\phi) \coloneqq \big(\text{number of variables in $\phi$}\big)$,
    \item $\parfunc(\phi) \coloneqq \big(\text{number of occurrences of $\floor{\frac{\cdot}{t^d}}$ and $(\cdot \bmod f(t))$  in an atomic formula of $\phi$}\big)$,
    \item $\parconst(\phi) \coloneqq \max\{ \bitlength{f} : \text{$f  \in \Z[t]$ occurs in $\phi$}\}$,
    \item $\parbound \coloneqq \max\{ 0, \bitlength{B(w)} : \text{$w$ is in the domain of $B$}\}$.
\end{itemize}

With this notation, let us prove the following lemma, 
which gives us essentially the proof that \Cref{algo:gaussianqe} runs in non-deterministic 
polynomial time (see \Cref{lemma:boundedqe-in-np} below).  

\begin{restatable}{lemma}{LemmaParamQE}
    \label{lemma:param-QE}
    Let $\exists \vec x : \phi(\vec x, \vec z)$ be a formula in input of~\BoundedQE, where $\phi$ is a positive 
    Boolean combination of linear \PPA constraints, $\vec x = (x_1,\dots,x_n)$ and $\vec z = (x_{n+1},\dots,x_{g})$, 
    where $n \geq 1$.
    Then, for every branch output $\exists \vec w \leq B : \psi(\vec w,\vec z)$ we have:
  \begin{equation*}
    \text{if }\; 
    \begin{cases}
        \paratom(\phi) &\!\leq a\\ 
        \parvars(\phi) &\!= g\\ 
        \parfunc(\phi) &\!= 0\\ 
        \parconst(\phi) &\!\leq c\\
    \end{cases}
    \quad\text{ then }\;
    \begin{cases}
        \paratom(\psi)  &\!\leq 2\cdot (n+a) + 1\\ 
        \parvars(\psi)  &\!= g\\ 
        \parfunc(\psi)  &\!= 0\\ 
        \parconst(\psi) &\!\leq (n+a)^3 \cdot c^3 
    \end{cases}
\end{equation*}
Moreover, $\vec w$ contains $n$ variables and $\parbound \leq (n+a)^{11} \cdot c^9$.
\end{restatable}

\begin{proof} 
    First observe that the same as the input formula $\phi$, the formula $\psi$ does not feature remainder 
    functions $(\cdot \bmod f(t))$ nor integer divisions $\sfloor{\frac{\cdot}{f(t)}}$, 
    and we have $\parfunc(\phi)=\parfunc(\psi) = 0$.
    Let us now consider other parameters.
    \begin{description}
        \item[Bound on {\rm$\paratom(\psi)$}:] 
        The relevant lines are: line~\ref{gauss:set-Z-end}, line~\ref{gauss:mt-is-positive}, 
        line~\ref{gauss:assert-sign-ft}, 
        line~\ref{gauss:restore}, and line~\ref{gauss:return}.
        Line~\ref{gauss:set-Z-end} adds one conjunct for each element in the set $Z$ guessed in line~\ref{gauss:set-Z}. 
        The cardinality of this set is bounded by $a$. 
        In parallel, the algorithm constructs a formula $\chi$, which has initially one atomic formula (line~\ref{gauss:mt-is-positive}).
        Lines~\ref{gauss:assert-sign-ft} and~\ref{gauss:restore} add, respectively, one conjunct to $\chi$ and 
        one conjunct to $\phi$ for each variable in $\vec x$ for which the non-deterministic 
        branching performed in~line~\ref{gauss:ast} does not lead to executing the~\textbf{continue} 
        instruction in line~\ref{gauss:guess-all-zeros}. Let us assume that line~\ref{gauss:restore} is executed $q \leq n$ times.
        Therefore, when the algorithm reaches line~\ref{gauss:unslackloop}, the formula $\phi$ has thus 
        at most $a + a + q$ atomic formulae and the formula $\chi$ has at most $q + 1$ atomic formulae. 
        The last line~\ref{gauss:return} returns the conjunction of $\phi$ and $\chi$, and we thus conclude that 
        $\paratom(\psi) \leq a + a + q + q + 1 \leq 2\cdot (n + a) + 1$.

        \item[Bound on {\rm$\parvars(\psi)$}:] The procedure simply ``replaces'' the $n$ variables from 
        $\vec x$ with variables $\vec w$ under the scope of bounded quantifiers. 
        Following lines~\ref{gauss:append-seq} and~\ref{gauss:subst-rem}, each variable in $\vec x$ corresponds 
        to a variable in $\vec w$ ($\vec w$ thus features $n$ variables). 
        Therefore, $\parvars(\psi) \leq v$.

        \item[Bound on {\rm$\parconst(\psi)$}:] 
        First observe that $\parconst(\psi)=max\{\parconst(\phi'),\parconst(\chi')\}$, where 
        $\phi'$ and $\chi'$ are the systems stored in the variables $\phi$ and $\chi$ in the last line~\ref{gauss:return}.
        From line~\ref{gauss:mt-is-positive} and line~\ref{gauss:assert-sign-ft} we see that  
        $\parconst(\chi') =\max\{\bitlength{\prod N}, \bitlength{f_0},\dots,\bitlength{f_{q-1}}\}$ for 
        the set $N$ such that $N \subseteq \{ f(t) : \text{the relation $(f(t) \divides \cdot)$ occurs in $\phi$} \}$ 
        and equations ${f_0(t) \cdot x_1 + \tau_0 = 0},$ $\dots,$ ${f_{q-1}(t) \cdot x_q + \tau_{q-1} = 0}$ are 
        guessed in line~\ref{gauss:guess-equation} within the main \textbf{foreach} loop 
        from line~\ref{gauss:mainloop}.
        Let us now upper-bound the bit size of $f_0,\dots,f_{q-1}$ together with the bit size of the 
        coefficients and constants of the formula $\phi'$.

        This bound uses in a crucial way~\Cref{prop:gauss-fundamental}. 
        From the definitions of the matrices $M_0,\dots,M_q$, given at the beginning of the appendix, 
        we see that all coefficients of the variables and 
        constants of the terms in $\phi'$ are entries of the matrix $M_q$ located in columns among 
        $1,\dots,g+r+1$, where $g$ is the overall number of free variables in the input formula $\phi$ 
        (including both $\vec x$ and~$\vec z$), and $r$ is the number 
        of slack variables introduced by the algorithm in line~\ref{gauss:introduce-slack}.

        Following~\Cref{prop:gauss-fundamental}, we see that each 
        $f_i$ is the determinant of $B[1,\dots,i+1; 1,\dots,i+1]$ and that the entries of $M_q$ 
        are determinants of $q \times q$ or $(q + 1) \times (q + 1)$ sub-matrices of $M_0$.
        All entries in the columns $1,\dots,g$ of $M_0$ are coefficients of the variables $\vec x$ and $\vec z$ 
        in the formula stored in $\phi$ at the beginning of the loop in line~\ref{gauss:mainloop}, 
        whereas the $(g+1)$-th column contains constants of terms in this formula, 
        and the bit size of all these entries is bounded by $\parconst(\phi)$. 
        (In particular, the equalities $f(t) = 0$ added in line~\ref{gauss:set-Z-end} do not 
        feature any variable, and $f(t)$ is an entry located in the $(g+1)$-th column 
        of $M_0$.) All entries in the columns $g+2,\dots,g+r+1$ of $M_0$ are either $0$ or $1$ 
        (slack variables have coefficient $1$ initially). 
        Therefore, to establish the bound on $\parconst(\phi')$ and 
        $\bitlength{f_0}$,\dots,$\bitlength{f_{q-1}}$, it suffices 
        to bound the bit size of these determinants. Below, we provide the computation for
        a $(q + 1) \times (q + 1)$ sub-matrix, which subsumes the one for $q \times q$ matrices.

        Let~$B \in \Z[t]^{(q+1) \times (q+1)}$ be a matrix with entries of bit size at most 
        $c \geq \max(1,\parconst(\phi))$:
         \begin{equation*}
            \label{eq:apx:B-matrix}    
            B \coloneqq \begin{pmatrix}
                b_{1,1} & \ldots  & b_{1,q+1}\\
                \vdots  & \ddots & \vdots\\
                b_{q+1,1} & \ldots  & b_{q+1,q+1}\\
            \end{pmatrix}.
        \end{equation*}
        Let us write $S_{q+1}$ for the set of all permutations of $(1,\dots,q+1)$, and 
        $\text{sgn}(\sigma)$ for the signature of a permutation $\sigma$, i.e., 
        $\text{sgn}(\sigma) = +1$ if $\sigma$ can be obtained from $(1,\dots,q+1)$ by exchanging entries an even number of times, and $\text{sgn}(\sigma) = -1$ otherwise.
        Let us also write $\sigma(i)$ for the~$i$th entry of $\sigma$.
        The determinant of $B$ is the element of~$\Z[t]$ defined as
        \[
            \det B = \sum_{\sigma \in S_{q+1}} \text{sgn}(\sigma) \cdot b_{1,\sigma(1)} \cdot \ldots \cdot b_{q+1,\sigma(q+1)}.
        \]
        Let $D \coloneqq  \max_{i,j} \,\deg(b_{i,j})$ and $H \coloneqq \max_{i,j} h(b_{i,j})$.
        We first bound the degree and height of $\det B$. For the degree, we have:
        \begin{align*}
            \deg(\det B) &\leq \max_{\sigma \in S_{q+1}}\big( \deg(b_{1,\sigma(1)}) + \ldots + \deg(b_{q+1,\sigma(q+1)}) \big) \leq (q+1) \cdot D.
        \end{align*}
        To bound the height, consider first a polynomial $b_{1,\sigma(1)} \cdot \ldots \cdot b_{q+1,\sigma(q+1)}$, 
        for some $\sigma \in S_{q+1}$. 
        Writing ${b_{j,\sigma(j)} = \sum_{i=0}^D g_j \cdot t^{j}}$,
        we have  
        \begin{equation*}
            b_{1,\sigma(1)} \cdot \ldots \cdot b_{q+1,\sigma(q+1)}
            = \sum_{j_1,\dots,j_{q+1} \in [0,D]} \nolimits\prod\nolimits_{i=1}^{(q+1)} g_{j_i} \cdot t^{j_{i}}.
        \end{equation*}
        When performing all multiplications and additions, the height of $b_{1,\sigma(1)} \cdot \ldots \cdot b_{q+1,\sigma(q+1)}$ 
        is thus bounded by $(D+1)^{(q+1)} \cdot H^{q+1}$.
        Therefore,
        \begin{align}
            \label{eq:parcoeff:height}
            h(\det B) &\leq (q+1)^{q+1} \cdot (D+1)^{(q+1)} \cdot H^{q+1}.
        \end{align}
        
        We can now bound the bit size of $\det B$: 
        \begin{align*}
            \bitlength{\det B} &= (\deg(\det B)+1) \cdot (\ceil{\log_2(h(\det B)+1)}+1)\\ 
            &\leq ((q+1) \cdot D + 1) \cdot (\sceil{\log_2(1+(q+1)^{q+1} \cdot (D+1)^{(q+1)} \cdot H^{q+1})}+1)\\
            &\leq (q+1) \cdot (D + 1) \cdot \big((q+1) \cdot \sceil{\log_2((q+1) \cdot (D+1) \cdot (H+1))}+1\big)\\
            &\leq (q+1)^2 \cdot (D + 1) \cdot \big(\sceil{\log_2((q+1) \cdot (D+1)) + \log_2(H+1)}+1\big)\\
            &\leq (q+1)^3 \cdot (D + 1)^2 \cdot \big(\sceil{\log_2(H+1)}+1\big)\\
            &\leq (n+1)^3 \cdot c^3.
        \end{align*}

        It remains to compare the bounds for $\parconst(\phi')$ and $\parconst(\chi')$. 
        Let us define the integers $D \coloneqq \max\{\deg(f) : f \in N\}$ and $H \coloneqq \max\{h(f) : f \in N\}$, 
        then we have: 
        \begin{equation*}
            \deg(\bitlength{\textstyle\prod N}) \leq a \cdot D 
            \qquad\qquad h(\bitlength{\textstyle\prod N}) \leq (D+1)^{a} \cdot H^{a}.
        \end{equation*}
        Therefore,  
        \begin{align*}
            \bitlength{\textstyle\prod N} &\leq (a \cdot D+1) \cdot (\sceil{\log_2((D+1)^{a} \cdot H^{a}+1)}+1)\\ 
            &\leq a \cdot (D + 1) \cdot \big(a \cdot \sceil{\log_2((D+1) \cdot (H+1))}+1\big)\\
            &\leq a^2 \cdot (D + 1)^2 \cdot \big(\sceil{\log_2(H+1)}+1\big)\\
            &\leq a^2 \cdot c^2,
        \end{align*}
        and we can conclude that 
        \begin{align*}
            \parconst(\psi) &\,=\, max\{\parconst(\phi'),\parconst(\chi')\}     \\
                            &\,\leq\, \max\{(n+1)^3 \cdot c^3, a^2 \cdot c^2\}  \\
                            &\,\leq\, (n + a)^3 \cdot c^3.
        \end{align*}

        \item[Bound on {\rm$\bitlength{B(w)}$}:] The relevant lines are line~\ref{gauss:guess-mt}, 
        line~\ref{gauss:subst-rem}, line~\ref{gauss:new-factor}, and line~\ref{gauss:append-seq}. 
        Line~\ref{gauss:guess-mt} defines a product of polynomials ${m(t) = \pm \prod N}$ for 
        $N \subseteq \{ f(t) : \text{the relation $(f(t) \divides \cdot)$ occurs in $\phi$} \}$.
        Let ${f_0(t) \cdot x_1 + \tau_0 = 0},$ $\dots,$ ${f_{q-1}(t) \cdot x_q + \tau_{q-1} = 0}$ be the equations guessed in line~\ref{gauss:guess-equation} during the run of the branch~$\beta$. 
        Because of the updates performed in line~\ref{gauss:new-factor},
        when line~\ref{gauss:append-seq} is executed the $(i+1)$-th time, 
        the bound given to the slack variable $y$ is $\pm f_0(t) \cdot \dots \cdot f_i(t) \cdot \prod N-1$. 
        Therefore, the bound given to variables in line~\ref{gauss:subst-rem} 
        is~$\pm f_0(t) \cdot \dots \cdot f_{q-1}(t) \cdot \prod N-1$.
        Let us now define 
        \begin{align*}
            D &\coloneqq {\max(\{\deg(f) : f \in N\} \cup \{\deg(f_i) : i \in [1,q]\})} \\
            H &\coloneqq \max(\{h(f) : f \in N\} \cup \{h(f_i) : i \in [1,q]\}).
        \end{align*} 
        If we consider for every $i \in [0,q-1]$ the polynomial 
        $f(t) \coloneqq f_0(t) \cdot \dots \cdot f_i(t) \cdot \prod N-1$, then 
        (since $\card{N} \leq a$ and $q \leq n$) we have: 
        \begin{equation*}
            \deg(f) \leq (n+a) \cdot D \qquad\qquad h(f) \leq (D+1)^{n+a} \cdot H^{n+a}+1,
        \end{equation*}
        where the bound on $h(f)$ is found using similar arguments as the ones used to derive the bound 
        in~\Cref{eq:parcoeff:height}. Therefore, for every variable $w$ in $\vec w$, we obtain 
        \begin{align*}
            \bitlength{B(w)} &\leq ((n+a) \cdot D+1) \cdot (\sceil{\log_2((D+1)^{n+a} \cdot H^{n+a}+2)}+1)\\ 
            &\leq (n+a) \cdot (D + 1) \cdot \big((n+a) \cdot \sceil{\log_2((D+1) \cdot (H+1))}+1\big)\\
            &\leq (n+a)^2 \cdot (D + 1)^2 \cdot \big(\sceil{\log_2(H+1)}+1\big)\\
            &\leq (n+a)^2 \cdot (n+1)^9 \cdot c^9\\
            &\leq (n+a)^{11} \cdot c^9.
        \end{align*}
        In the second-to-last line, we have used the fact that $\bitlength{f_i} \leq (n+1)^3 \cdot c^3$ 
        (recall that $f_i$ is the determinant of $B[1,\dots,i+1; 1,\dots,i+1]$).
        \qedhere
    \end{description}
\end{proof}

\begin{lemma}
    \label{lemma:boundedqe-in-np}
    The algorithm~\BoundedQE runs in non-deterministic polynomial time.
\end{lemma}

\begin{proof}
    Following the proof of~\Cref{lemma:param-QE}, we now know that all objects~\BoundedQE constructs during its execution are of polynomial bit size. 
    With this information, the proof that the algorithm 
    only requires non-deterministic polynomial time is simple: it suffices to show that each line of the algorithm can be implemented in non-deterministic polynomial time, and that all loops only iterate a polynomial number of times. 

    Regarding the operations performed, a quick glance at the pseudocode of~\BoundedQE reveals that, in fact, except for the lines that perform guesses, all other operations can be implemented in polynomial time.
    In particular, note that the polynomial divisions performed in line~\ref{gauss:divide}
    can be performed in polynomial time with the standard Euclidean algorithm.

    For the number of loops iterations, the \textbf{foreach} loop of 
    line~\ref{gauss:Z-for} iterates at most $a \coloneqq \paratom(\phi)$ times; 
    the \textbf{foreach} loop of line~\ref{gauss:mainloop} iterates $n$ times, where $n$ is the number of eliminated variables;
    and the \textbf{foreach} loop of line~\ref{gauss:unslackloop} 
    iterates at most $2(n+a)+1$ times (following the bound on $\paratom(\psi)$ from~\Cref{lemma:param-QE}).
\end{proof}

\section{Extended material for Section~\ref{sec:division}}\label{app:division}

The main result of \Cref{sec:division} is the correctness proof for \Cref{algo:qe} (\ElimBounded), which was 
only sketched in that section. In this appendix, we first formally prove \Cref{lemma:elimbounded:Dtwo}
(with the help of two auxiliary lemmas). Together with 
the correctness of~\BoundedQE 
(\Cref{lemma:boundedqe:correctness}, proven~in~\Cref{subapp:gaussian-elimination:A2}), 
this will give us the necessary equivalences to complete 
the proof of \Cref{lemma:elimbounded:correctness}. 
The complexity of \Cref{algo:qe} is analysed at the end of the appendix. 

In this appendix, we write $\card{\vec x}$ for the number 
of variables in the vector $\vec x$. 

\subsection{Proof of Lemma~\ref{lemma:elimbounded:preprocessing}}

\LemmaElimBoundedPreProcessing*

\begin{proof}
    It is simple to see that this lemma follows as soon 
    as we prove that, in any solution to the input formula, the base $t$ expansion of every bounded variable $w$ has at most $\bitlength{B(w)}$ \mbox{$t$-digits}. For this, it suffices to establish 
    $\abs{g(k)} \leq k^{\bitlength{g}}$, for every~$g \in \Z[t]$ and $k \geq 2$.
    Consider a non-zero polynomial $g(t) \coloneqq \sum_{j=0}^d a_j \cdot t^j$, where $a_0,\dots,a_i \in \Z$. (Trivially, $0 \leq k^{\bitlength{0}}$.)
    Let $a \coloneqq \max\{ \abs{a_j} : j \in [0,d]\}$.
    Then, 
    \begin{align*}
        \abs{g(k)} 
        \,=\, \abs{\sum\nolimits_{j=0}^d a_j \cdot k^j} 
        \,\leq\, (d+1) \cdot a \cdot k^d
        \,\leq\, k^{d + \log_k((d+1) \cdot a)}
        \,\leq\, k^{d + \log_2((d+1) \cdot a)},
    \end{align*}
    and analysing the exponent we see
    \begin{align*}
        d + \log_2((d+1) \cdot a) 
        \leq d + (d+1) \cdot \log_2(a+1) 
        \leq (d+1) \cdot (\log_2(a+1) +1) 
        \leq \bitlength{g}.
        &\qedhere
    \end{align*}
\end{proof}

\subsection{Proof of Lemma~\ref{lemma:elimbounded:Dtwo}}

Before delving into the proof of this lemma, 
we need two auxiliary results (\Cref{lemma:elimbounded:constraint} and~\Cref{lemma:elimbounded:divide}).
The first lemma describes the \PPA constraints produced during the 
execution of the body of the \textbf{while} loop.

\begin{lemma}\label{lemma:elimbounded:constraint}
    At the beginning of every iteration of the \textbf{while} loop in line \ref{line:elimbounded:while}, every 
    constraint which has a non-integer coefficient of a variable from $\vec y$ has the form
    \begin{equation}\label{eq:elimbounded:constraint}
        \sigma(\vec y)\cdot t + \rho(\vec y) + \floor{\frac{\tau(\vec z)}{t^k}} \sim 0 
        \qquad\qquad \text{with $\sim$ in $\{=,\leq\}$},        
    \end{equation}
    where $\sigma(\vec y)$ is a linear \PPA term, $\rho(\vec y)$ is a linear term with coefficients in $\Z$, 
    $\tau(\vec z)$ is either a linear non-shifted \PPA term or $(\tau'(\vec z) \bmod f(t))$, where $\tau'$ is 
    linear non-shifted, $k$ is a non-negative integer, and $\sfloor{\frac{\tau}{t^0}}=\tau$. 
\end{lemma}
\begin{proof}
    By induction on the number of iterations $k\geq 0$. 
    Initially, when the constraint $(\eta \sim 0)$ is taken from the formula $\phi_\varnothing$, the term $\eta$ has 
    one of the two forms: 
    \begin{itemize}
        \item  $\sigma(\vec y)\cdot t + \rho(\vec y) + \tau(\vec z)$, 
        where we have split the linear \PPA term over the variables $\vec y$ 
        into $\sigma(\vec y)\cdot t$ and 
        a linear term with integer coefficients $\rho(\vec y)$; or 
        \item $\sigma(\vec y)\cdot t + \rho(\vec y) + (\tau(\vec z) \bmod f(t))$, where, again, $\rho(\vec y)$ is 
        a linear term with integer coefficients. 
    \end{itemize}
    Here, in both cases, $\tau(\vec z)$ is a linear non-shifted term.
    It is clear that \Cref{eq:elimbounded:constraint} unifies both items.
    
    Let us assume that at the beginning of the $k$-th iteration every term with 
    a non-integer coefficient of a variable from $\vec y$ has the form \Cref{eq:elimbounded:constraint}.  
    Consider the updates in $\phi$ performed during one iteration of the loop. 
    This formula is only modified in line~\ref{line:elimbounded:replace-in-phi}, where 
    the constraint $(\eta \leq 0)$ is replaced with a conjunction 
    \begin{equation}\label{eq:elimbounded:constraint:2}
        \gamma \,\land\, \Big(\sigma(\vec y) + r + \Bigl\lfloor{\frac{\floor{\frac{\tau(\vec z)}{t^{k}}}}{t}}\Bigr\rfloor \sim 0\Big),        
    \end{equation}
    where $r$ is an integer (see lines~\ref{line:elimbounded:guess-r} and~\ref{line:elimbounded:increment-r}) 
    and $\gamma$ is either an equality $(t \cdot r=\rho(\vec y))$ (line~\ref{line:elimbounded:add-eq-to-psi})
    or a conjunction of two inequalities $(t\cdot r\leq \rho(\vec y)) \,\land\, (\rho(\vec y) \leq t\cdot(r+1)-1)$ 
    (line~\ref{line:elimbounded:add-neq-to-psi}).

    By induction hypothesis, $\sigma(\vec y)$ is a linear term with coefficients in $\Z[t]$. 
    Let us represent this term as $\sigma'(\vec y) \cdot t + \rho'(\vec y)$, 
    where $\rho'$ is a linear term with integer coefficients. 
    The term $\sigma'(\vec y)$ can be equal to zero; this happens when $\sigma(\vec y)$ is a linear term with integer 
    coefficients. We can now rewrite \Cref{eq:elimbounded:constraint:2} as 
    \begin{equation}\label{eq:elimbounded:constraint:3}
        \gamma \,\land\, \Big(\sigma'(\vec y) \cdot t + (\rho'(\vec y) + r) + \floor{\frac{\tau(\vec z)}{t^{k+1}}} \sim 0\Big).        
    \end{equation}
    It remains to notice that, by induction hypothesis, $\rho(\vec y)$ is a linear term over 
    $\Z$, and thus there are no occurrences of a variable from $\vec y$ with a non-integer coefficient in $\gamma$.
\end{proof}

The ``divisions by the parameter $t$'' 
are formalised in the second auxiliary lemma, which gives us
the two necessary equivalences already introduced in the sketch 
of the proof of \Cref{lemma:elimbounded:Dtwo} given in the body of the paper. 

\begin{lemma}\label{lemma:elimbounded:divide}
    Let $C,D \in \Z$, with $C \leq D$. Then, for the parameter $t \in \N$ such that $t\geq 2$, $z \in \Z$, 
    and $w \in [t \cdot C, t \cdot D]$, the following equivalences hold:
    \begin{alphaenumerate}
        \item $t \cdot z + w = 0 \iff \bigvee_{r\in[C, D]}(z + r = 0 \land t \cdot r = w)$,\label{item:div-eq}
        \item $t \cdot z + w < 0 \iff \bigvee_{r\in[C, D]}(z + r < 0 \land t \cdot r \leq w \land w < t \cdot (r+1))$.\label{item:div-ineq}
    \end{alphaenumerate}
\end{lemma}
\begin{proof}
    Since we know that $w \in [t \cdot C, t \cdot D]$, there is an integer $r^\ast \in [C, D]$ such that 
    $\floor{\frac{w}{t}} = r^\ast$. 
    \begin{description}
        \item[Item (\ref{item:div-eq}):] For the left-to-right direction, the equality 
        $t \cdot z + w = 0$ implies that $w$ is divisible by $t$. Hence, $w = t\cdot r^\ast$, and we have 
        $t \cdot z + t\cdot r^\ast = 0$ which can be rewritten as $z + r^\ast = 0$ (since $t\geq 2$). 
        The right-to-left direction is trivial.
        \item[Item (\ref{item:div-ineq}):] Observe that the fact that $t \cdot r \leq w \land w < t \cdot (r+1)$ 
        for some $r \in [C, D]$ means that $r=r^\ast$, and it is sufficient to prove the equivalence 
        $t \cdot z + w < 0 \iff z + r^\ast < 0$. 
        Below, we write $\frpart{x}$ for the fractional part of a real 
        number $x \in \R$, that is, 
        $\frpart{x}$ is the only real number such that 
        $0 \leq \frpart{x} < 1$ and $x = \floor{x} + \frpart{x}$.
        We have: 
        \begin{align*}
            t \cdot z + w < 0 \iff & t \cdot z + t \cdot \frac{w}{t} < 0 
            & \text{recall: $t \geq 2$}\\ 
            \iff & t \cdot z + t \cdot \left(\floor{\frac{w}{t}}+\frpart{\frac{w}{t}}\right) < 0
            \\
            \iff & z + \left(\floor{\frac{w}{t}}+\frpart{\frac{w}{t}}\right) < 0 &
            \\
            \iff & z + \floor{\frac{w}{t}} < 0 & \text{since $0 \leq \frpart{\frac{w}{t}} < 1$}
            \\
            \iff & z + r^\ast < 0.  &&\qedhere
        \end{align*}
    \end{description}
\end{proof}


Thanks to the auxiliary \Cref{lemma:elimbounded:constraint,lemma:elimbounded:divide}, we are now ready to prove~\Cref{lemma:elimbounded:Dtwo}:

\LemmaElimBoundedDtwo*

\begin{proof}
    The $(i+1)$-th iteration of the \textbf{while} loop in lines \ref{line:elimbounded:while}--\ref{line:elimbounded:replace-in-phi} 
    picks a constraint $(\eta \sim 0)$ from $\phi_{\vec s}$ and divides it using \Cref{lemma:elimbounded:divide}. 
    Since this lemma works with strict inequalities, if the symbol $\sim$ is the non-strict inequality sign $\leq$, 
    in line~\ref{line:elimbounded:sim} we transform $(\eta \leq 0)$ into an equivalent formula $(\eta - 1 < 0)$. 
    We further assume that, in this case, $\eta$ has been replaced with $\eta -1$, 
    and we work with the strict inequality $(\eta < 0)$. 
    
    By \Cref{lemma:elimbounded:constraint}, every constraint $(\eta \sim 0)$, which has an occurrence 
    of a variable from $\vec y$ with a non-integer coefficient, has the form 
    (\ref{eq:elimbounded:constraint}), where now the symbol $\sim$ is from $\{=,<\}$. 
    The term $\eta$ is represented as 
    \begin{equation*}
        \eta \,=\,\eta' \cdot t + \rho',
        \text{ where }
        \qquad
        \eta' \,=\, \sigma(\vec y) \,+\, \floor{\frac{\tau(\vec z)}{t^{k+1}}}
        \quad
        \text{ and }
        \quad
        \rho' \,=\, \rho(\vec y) \,+\, 
        \Big(\floor{\frac{\tau(\vec z)}{t^k}} \bmod t\Big)
    \end{equation*}
    Notice that the term $\rho'$ is exactly the term assigned to the variable $\rho$ in line~\ref{line:elimbounded:rho} 
    of \Cref{algo:qe}. This term is a sum of a linear polynomial with integer coefficients $\rho(\vec y)$ and 
    a linear occurrence of the integer remainder function with coefficient $1$. 
    Therefore, we can compute $\onenorm{\rho'}$ and obtain that 
    $\rho' \in \big[t\cdot (-\onenorm{\rho'}), \: t \cdot \onenorm{\rho}\big]$.
    Indeed, every variable from $\vec y$ and the summand with linear occurrence of $(\cdot\!\!\mod t)$ 
    belong to the segment $[0, t-1]$. 
    With these bounds on $\rho'$, we can apply \Cref{lemma:elimbounded:divide} to the constraints 
    $(\eta' \cdot t + \rho' = 0)$ and $(\eta' \cdot t + \rho' < 0)$ for the integers 
    $C = (-\onenorm{\rho'})$ and $D = \onenorm{\rho'}$. This gives us two equivalences 
    \begin{align}
        (\eta' \cdot t + \rho' = 0) \iff & \!\!\!\bigvee_{r\in [-\onenorm{\rho'}, \: \onenorm{\rho'}]}\!\!\!
        (\eta' + r = 0) \hphantom{+11} \, \land \, (t \cdot r = \rho')\label{eq:elimbounded:divide-eq}
        \\
        (\eta' \cdot t + \rho' < 0) \iff & \!\!\!\bigvee_{r\in [-\onenorm{\rho'}, \: \onenorm{\rho'}]}\!\!\!
        (\eta' + r + 1 \leq 0) \, \land \, \underbrace{(t \cdot r \leq \rho') \, \land \, (\rho' < t \cdot (r+1))}_{\gamma}.
        \label{eq:elimbounded:divide-neq}
    \end{align}
    In the second equivalence, the strict inequality $(\eta' + r < 0)$ originated from (\ref{item:div-ineq})
    is replaced with the equivalent non-strict one $(\eta' + r + 1 \leq 0)$. 
    This replacement is done in \Cref{algo:qe} by line~\ref{line:elimbounded:increment-r} and does not 
    affect the formula $\gamma$, which is defined above, in line~\ref{line:elimbounded:add-neq-to-psi}.  
    
    Since $\phi_{\vec s}$ is a positive Boolean combination of atomic formulae, the logical equivalences 
    $(\psi_1 \,\lor\, \psi_2)\,\land\,\psi_3 \iff (\psi_1 \,\land\, \psi_3)\,\lor\,(\psi_2 \,\land\, \psi_3)$ and 
    $(\psi_1 \,\lor\, \psi_2)\,\lor\,\psi_3 \iff (\psi_1 \,\lor\, \psi_3)\,\lor\,(\psi_2 \,\lor\, \psi_3)$ 
    together with \Cref{eq:elimbounded:divide-eq,eq:elimbounded:divide-neq} give us the equivalence 
    \begin{equation}\label{eq:elimbounded:disj}
        \phi_{\vec s} \iff \bigvee_{r\,\in\, [-\onenorm{\rho'}, \: \onenorm{\rho'}]} \phi_{\vec s r}.    
    \end{equation}
    The formula $\phi_{\vec s r}$ is the result of replacement of the constraint $(\eta \sim 0)$ with 
    the conjunction on the right-hand side of 
    \Cref{eq:elimbounded:divide-eq} if the symbol $\sim$ is $=$, and with the conjunction on the 
    right-hand side of \Cref{eq:elimbounded:divide-neq} if $\sim$ is the strict inequality sign $<$. 
    The formula $\phi_{\vec s r}$ is assigned to $\phi$ in line~\ref{line:elimbounded:replace-in-phi}, 
    and this finishes the $(i+1)$-th iteration of the \textbf{while} loop. 
    
    Line~\ref{line:elimbounded:guess-r} guesses an integer $r$ in $[-\onenorm{\rho'}, \, \onenorm{\rho'}]$, 
    and thus $G = \{\vec sr : r \in [-\onenorm{\rho'},-\onenorm{\rho'}]\}$, and \Cref{eq:elimbounded:disj} 
    imply Item~\ref{lemma:elimbounded:Dtwo:i1}. To prove Item~\ref{lemma:elimbounded:Dtwo:i2}, observe that
    \begin{align*}
        \deg(\vec y,\phi_{\vec s r}) &\,=\, \deg(\vec y,\phi_{\vec s}) - \deg(\vec y,\eta \sim 0) + \deg(\vec y,\eta' \sim 0) + \deg(\vec y, \gamma)\\
        &\,=\, \deg(\vec y,\phi_{\vec s}) - \deg(\vec y,\sigma \cdot t \sim 0) + \deg(\vec y,\sigma \sim 0) = \deg(\vec y, \phi_{\vec s}) - 1.   
        &&
    \end{align*}
    Indeed, the constraints from $\gamma$ have degree $0$ and $\deg(\theta + \tau \sim 0) \,=\, \deg(\theta \sim 0)$ for 
    every \PPA term $\tau$ not featuring variables from $\vec y$. Thus, \Cref{lemma:elimbounded:Dtwo} is proved.    
\end{proof}

Let us now inductively apply \Cref{lemma:elimbounded:Dtwo} to show that 
the \textbf{while} loop performs at most $\deg(\vec y, \phi_{\varnothing})$ iterations,
and that the disjunction, across non-deterministic branches, over all formulae~$\phi_{\vec r}$ 
obtained at the end of this loop is equivalent to $\phi_{\varnothing}$.

\begin{lemma}\label{lemma:elimbounded:Dthree}
    The \textbf{while} loop of line~\ref{line:elimbounded:while}  
    performs at most $\deg(\vec y, \phi_{\varnothing})$ iterations.
\end{lemma}

\begin{proof}
    Directly from Item~\ref{lemma:elimbounded:Dtwo:i2} of \Cref{lemma:elimbounded:Dtwo}.
\end{proof}

\begin{lemma}\label{lemma:elimbounded:Dfour}
    Let $R_{\ast} \coloneqq \{\vec s \in R_i :  i \in \N, \deg(\vec y, \phi_{\vec s}) = 0\}$.
    We have 
    $\phi_{\varnothing} \iff \bigvee_{\vec s \in R_{\ast}} \phi_{\vec s}$.
\end{lemma}

\begin{proof}
    Let us define a set $R_{\ast}^{(i)} \coloneqq \{\vec s \in R_j :  j \in \N, \text{ and } \deg(\vec y, \phi_{\vec s}) = 0 \text{ or } i = j\}$, 
    and an non-negative integer $m \coloneqq \deg(\vec y, \phi_{\varnothing})$.  
    By induction on $i$ from $m$ to $0$, we prove the following equivalence:  
    $\bigvee_{\vec s \in R_{\ast}^{(i)}} \phi_{\vec s} 
    \iff \bigvee_{\vec s \in R_{\ast}} \phi_{\vec s}$.
    The statement then follows from $R_{\ast}^{(0)} = \{\varnothing\}$. 
    
    For the base case $i = m$, directly from~\Cref{lemma:elimbounded:Dthree} we have $R_{\ast}^{(m)} = R_{\ast}$. 
    For the induction step, assume that $\bigvee_{\vec s \in R_{\ast}^{(i+1)}} \phi_{\vec s} 
    \iff \bigvee_{\vec s \in R_{\ast}} \phi_{\vec s}$ (induction hypothesis), and let $\vec r \in R_i$. 
    If we have $\deg(\vec y,\phi_{\vec r}) = 0$, then the sequence $\vec r$ is also in the set $R_{\ast}^{(i+1)}$. 
    Else $\deg(\vec y,\phi_{\vec r}) > 0$, 
    and by Item~\ref{lemma:elimbounded:Dtwo:i1} of \Cref{lemma:elimbounded:Dtwo} we have 
    $\phi_{\vec r} \iff \bigvee_{\vec g \in G}\phi_{\vec g}$, 
    where ${G \coloneqq \{ \vec r g \in R_{i+1} : g \in \Z\}}$.
    Since $G \subseteq R_{\ast}^{(i+1)}$, 
    we conclude that $\bigvee_{\vec s \in R_{\ast}^{(i)}} \phi_{\vec s} 
    \iff \bigvee_{\vec s \in R_{\ast}^{(i+1)}} \phi_{\vec s}$, 
    which together with the induction hypothesis concludes the proof.
\end{proof}

\subsection{Correctness of \ElimBounded (proof of Lemma~\ref{lemma:elimbounded:correctness})}

We can now complete the proof of \Cref{lemma:elimbounded:correctness}. 
By~\Cref{lemma:elimbounded:Dthree}, the \textbf{while} loop does at most 
$\deg(\vec y, \phi_\varnothing)$ iterations, and~\Cref{lemma:elimbounded:Dfour} implies that $\phi_\varnothing$ 
is equivalent to $\bigvee_{\vec r \,\in\, R_{\ast}}\phi_{\vec r}$, where $R_{\ast} \coloneqq \{\vec s \in R_i :  i \in \N, \deg(\vec y, \phi_{\vec s}) = 0\}$. Every constraint in this disjunction as~the~form 
\[
    f(t) + \tau(\vec z) + \sum_{i=1}^{N} a_{i} \cdot y_i \sim 0,    
\]
where $\vec y = (y_1,\dots,y_N)$, the coefficients $a_{1},\dots,a_N$ are in~$\Z$, $f$ is in $\Z[t]$, $\tau(\vec z)$ is 
a non-shifted \PPA term, and the symbol $\sim$ is, as usual, from the set $\{=,\leq\}$. 
Recall that $N \,=\, \#\vec w\cdot M$, because every variable $y_i$ is a $t$-digit of a variable 
$w\in\vec w$ from the input formula. We now see that the variables $\vec y$ can be eliminated from 
$\exists \vec y : \phi_{\vec r}$ via any quantifier elimination procedure for \PA. In order to perform this 
step efficiently, we call \BoundedQE in line~\ref{line:elimbounded:gauss}.

Formally, \BoundedQE works with positive Boolean combinations of linear constraints with coefficients in $\Z[t]$, and 
it is necessary to replace each term $\tau(\vec z)$ with a fresh variable. 
These replacements are performed by lines~\ref{line:elimbounded:intro-s}--\ref{line:elimbounded:replace-tau}; 
let us call $\phi_{\vec r}'(\vec y, \vec z')$ the resulting formula. 
The map $S$ collects the key-value pairs $(z', \tau(\vec z))$ in order to restore 
the terms with the variables $\vec z$ after elimination of $\vec y$ from 
$\phi_{\vec r}'$. This, in particular, means that 
$\phi_{\vec r} = \phi_{\vec r}'[\sub{S(z')}{z'} : z'\in\vec z']$. 

By \Cref{lemma:boundedqe:correctness,lemma:boundedqe:int}, the result of application of the non-deterministic procedure 
\BoundedQE to $\exists \vec y : \phi_{\vec r}'(\vec y, \vec z')$ is a bounded formula
$\exists\vec w'_\delta \leq B'_\delta : \psi_\delta(\vec w'_\delta,\vec z')$, 
where for every variable $w\in \vec w'_\delta$, the bound $B'_\delta(w)$ is a non-negative integer, and    
we have the equivalence
\begin{equation}\label{eq:elimbounded:aftergauss:app}
    \exists \vec y : \phi_{\vec r}'(\vec y, \vec z') \iff \bigvee_{\delta \,\in\, \Delta_{\vec r}} 
    \exists \vec w'_\delta \leq B'_\delta : \psi_\delta(\vec w'_\delta,\vec z'),
\end{equation}
where the set $\Delta_{\vec r}$ is the union over all non-deterministic branches of \BoundedQE.
Since this algorithm only produces a bounded existential formula, to eliminate 
every bounded variable $w\in \vec w'_\delta$ we guess 
an integer $g$ from the segment $[0,B'_\delta(w)]$ and replace $w$ with $g$ in $\psi_\delta$. 
This is done by lines~\ref{line:elimbounded:foreach-w-to-int}--\ref{line:elimbounded:replace-w}. 
Let us gather all the substitutions $\sub{g}{w}$ from line~\ref{line:elimbounded:replace-w} together 
and denote by $\vec g$ the corresponding sequence of substitutions. The set of all possible values 
of $\vec g$ for the formula $\psi_\delta$ will be denoted by $G_\delta$. 

\begin{proof}[Proof of Lemma~\ref{lemma:elimbounded:correctness}]      
Let us define the set $\mathcal{B}$ as the union of all sequences $(\vec r, \delta, \vec g)$, where
    the sequence $\vec r$ is from $R_{\ast}$, $\delta$ is in $\Delta_{\vec r}$, and $\vec g\in G_\delta$. 
    For every sequence $\beta=(\vec r, \delta, \vec g) \in \mathcal{B}$, denote by $\psi'_\beta(\vec z')$ 
    the result of application of the substitutions $\vec g$ to the formula~$\psi_\delta(\vec w'_\delta, \vec z')$. 
    We obtain the following chain of equivalences:  
    \begin{align*}
        & \exists\vec w \leq B :\phi(\vec w, \vec z) 
        \\
        \iff & \exists \vec y : \phi_\varnothing(\vec y, \vec z) 
        & \text{\Cref{lemma:elimbounded:preprocessing}}
        \\
        \iff & \exists\vec y : \bigvee_{\vec r \,\in\, R_\ast}\phi_{\vec r}(\vec y, \vec z) 
        & \text{\Cref{lemma:elimbounded:Dfour}}
        \\ 
        \iff & \bigvee_{\vec r \,\in\, R_\ast} \big(\exists\vec y : \phi'_{\vec r}(\vec y, \vec z')\big)[\vec S] 
        & \!\!\!\!\!\phi_{\vec r} \,=\, \phi_{\vec r}'[\vec S] \text{ for } \vec S = \{\sub{S(z')}{z'} : z'\in\vec z'\}
        \\
        \iff & \bigvee_{\vec r \,\in\, R_\ast} \, \bigvee_{\delta \,\in\, \Delta_{\vec r}} 
        \big(\exists \vec w'_\delta \leq B'_\delta : \psi_\delta(\vec w'_\delta,\vec z')\big)[\vec S] 
        & \text{\Cref{eq:elimbounded:aftergauss:app}}
        \\ 
        \iff & \bigvee_{\beta \,\in\, \mathcal{B}} \psi'_\beta(\vec z')[\vec S] 
        & \text{by definition of }\mathcal{B}
        \\ 
        \iff & \bigvee_{\beta \,\in\, \mathcal{B}} \psi_\beta(\vec z). 
        & \text{where }\psi_\beta(\vec z)\coloneqq \psi'_\beta(\vec z')[\vec S]
    \end{align*}
    Hence, we have obtained the desired equivalence from the specification of \Cref{algo:qe}. 
\end{proof}

Before moving to the complexity of~\ElimBounded, let us establish~\Cref{lemma:properties-elimbounded}. 

\LemmaPropertiesElimBounded*

\begin{proof}
    The first statement follows from the fact that, in line~\ref{line:elimbounded:eta}, 
    \ElimBounded splits the term $\eta$ as $(\sigma(\vec y)\cdot t + \rho(\vec y) + \tau(\vec z))$, 
    where $\rho$ does not contain $t$, and $\tau$ \emph{is non-shifted}. 
    All functions introduced by the algorithm are applied to these terms~$\tau(\vec z)$. Following the specification of~\ElimBounded, occurrences of functions that comes instead from the input formula are of the form $(\tau'(\vec z) \bmod f(t))$, where $\tau'$ is again linear and non-shifted.

    For the second statement, observe that the input of~\ElimBounded does not contain any divisibility constraint, 
    and only~\BoundedQE adds these constraints. From~\Cref{lemma:boundedqe:int}, the divisibility constraints added are of the form $(d \divides \cdot)$, where $d$ is an integer. 
\end{proof}

\subsection{Complexity of \ElimBounded} 
\label{subapp:complexity-elimbounded}

We now proceed to the complexity analysis, which uses the parameters 
of formulae defined in \Cref{subapp:gaussian-elimination:A3}.

\begin{lemma}\label{lemma:param-elimbounded}
    Let $\exists\vec w \leq B :\phi(\vec w, \vec z)$ be a formula in input of~\ElimBounded, 
    where $\phi$ is a positive Boolean combination of linear (in)equalities with coefficients in $\Z[t]$ 
    and constraints $\sigma(\vec w) + (\tau(\vec z) \bmod{f(t)}) = 0$, with $\sigma$ linear, 
    and $\tau$ linear and non-shifted.
    Let $\vec w = (x_1,\dots,x_n)$ and $\vec z = (x_{n+1},\dots,x_{g})$, where $n \geq 1$. 
    Then, for every branch output $\psi(\vec z)$ of the algorithm, the following holds:
    \begin{equation*}
        \text{if }\; 
        \begin{cases}
            \paratom(\phi) &\!\leq a\\ 
            \parfunc(\phi) &\!\leq 1\\ 
            \parconst(\phi) &\!\leq c
        \end{cases}
        \quad\text{ then }\;
        \begin{cases}
            \paratom(\psi)  &\!\leq 2^7 \cdot a \cdot n \cdot c \cdot \bitlength{B}\\ 
            \parfunc(\psi)  &\!\leq 2^7 \cdot a\cdot n \cdot c \cdot \bitlength{B}\\ 
            \parconst(\psi) &\!\leq (2^{11}\cdot a \cdot n^3 \cdot c^3 \cdot \bitlength{B}^4)^{16}
        \end{cases}
    \end{equation*}
\end{lemma}


\begin{proof} The same as in the proof of \Cref{lemma:param-QE}, we consider each parameter separately:
    \begin{description}
        \item[Bound on {\rm$\paratom(\psi)$}:] In the first line of \ElimBounded, the bounds on the variables $\vec w$, 
        i.e., $2$~inequalities for each of the $n$ variables, are explicitly added to the formula. Then, in line~\ref{line:elimbounded:div-x}, 
        we add $2\cdot(\bitlength{B}+1)$ inequalities for every $w\in \vec w$ specifying that the variables introduced 
        by line~\ref{line:elimbounded:append-y} are $t$-digits. The same as in \Cref{sec:division}, let us denote by $\phi_\varnothing$ 
        the formula obtainded from $\phi$ after the execution of the \textbf{foreach} loop in 
        line~\ref{line:elimbounded:foreach-div-x}. We see that 
        \[
            \paratom(\phi_\varnothing)\leq a + 2\cdot n + 2\cdot n \cdot (\bitlength{B}+1) = a + 2\cdot n \cdot (\bitlength{B}+2).
        \]  
        
        By \Cref{lemma:elimbounded:Dthree}, the \textbf{while} loop in line~\ref{line:elimbounded:while} performs 
        at most $\deg(\vec y, \phi_\varnothing)$ iterations. 
        Observe that for every variable from $\vec y$, its coefficent $f(t)$ is such that 
        $\deg(f) \leq \bitlength{B} + \parconst(\phi)$ (see line~\ref{line:elimbounded:div-x}, 
        where the coefficients of the newly introduced variables are equal to a coefficient of $x$ multiplied 
        by $t^k$ for $k\in[0,\bitlength{B}]$.) Since $\phi_\varnothing$ have at most $(a + 2\cdot n)$ constraints where the variables 
        $\vec y$ can have non-integer coefficients, i.e., the constraints of $\phi$ and those added by 
        line~\ref{line:elimbounded:add-bounds}, 
        we conclude that $\deg(\vec y, \phi_\varnothing)\leq (a+2\cdot n) \cdot (\bitlength{B} + c)$. 
        \smallbreak
        Every iteration of the \textbf{while} loop replaces one constraint with either two 
        (lines~\ref{line:elimbounded:add-eq-to-psi} and~\ref{line:elimbounded:replace-in-phi}) 
        or three constraints (lines~\ref{line:elimbounded:add-neq-to-psi} and~\ref{line:elimbounded:replace-in-phi}). 
        Here and in the following, $\phi'$ will be the formula obtained by the end of execution of this loop. 
        We see that
        \begin{align*}
            \paratom(\phi') &\leq a + 2\cdot n \cdot (\bitlength{B}+2) + 3\cdot (a+2\cdot n) \cdot (\bitlength{B} + c) \\
                            &\leq a + 6\cdot n \cdot \bitlength{B} + 18 \cdot a\cdot n \cdot \bitlength{B} \cdot c \\
                            &\leq 2^5 \cdot a \cdot n \cdot c \cdot \bitlength{B}
        \end{align*}
        Lines~\ref{line:elimbounded:foreach-constraint}--\ref{line:elimbounded:replace-tau} and 
        \ref{line:elimbounded:foreach-w-to-int}--\ref{line:elimbounded:replace-w} does not change the number of 
        atomic formulae, and in order to obtain the bounds on $\paratom(\psi)$ it remains to apply 
        \Cref{lemma:param-QE}:
        \begin{align*}
            \paratom(\psi) &\leq 2\cdot(\#\vec y+\paratom(\phi'))+1 \\
                           &\leq 2\cdot(n\cdot(\bitlength{B}+1) + 2^5 \cdot a \cdot n \cdot c \cdot \bitlength{B} + 1) \\
                           &\leq 2^7 \cdot a \cdot n \cdot c \cdot \bitlength{B}.
        \end{align*} 
        \item[Bound on {\rm$\parfunc(\psi)$}:] To estimate the number of occurrences of $\floor{\frac{\cdot}{t^d}}$ 
        and $(\cdot \bmod f(t))$ in each atomic formula of $\psi$, consider lines~\ref{line:elimbounded:rho} 
        and~\ref{line:elimbounded:replace-in-phi}. 
        By \Cref{lemma:elimbounded:constraint}, the term $\eta$ in line~\ref{line:elimbounded:eta} contains 
        at most one occurrence of $\floor{\frac{\cdot}{t^d}}$ and at most one occurrence of $(\cdot \bmod f(t))$.
        Hence, every constraint of the formula $\gamma$ (see lines~\ref{line:elimbounded:rho}, 
        \ref{line:elimbounded:add-eq-to-psi}, and~\ref{line:elimbounded:add-neq-to-psi}) has at most at most 
        $2$ occurrences of $(\cdot \bmod f(t))$ and at most $1$ occurrence of $\floor{\frac{\cdot}{t^d}}$. 
        These constraints have only integer coefficients of the variables in $\vec y$ and 
        will not be taken into account by the next iterations of the \textbf{while} loop. 
        Line~\ref{line:elimbounded:replace-in-phi} only changes the power of the parameter $t$ in the denominator (rewriting $\sfloor{\frac{\sfloor{\frac{\tau}{t}}}{t^d}}$ as $\floor{\frac{\tau}{t^{d+1}}}$)
        and does not change the number of occurrences of $(\cdot \bmod f(t))$ and $\floor{\frac{\cdot}{t^d}}$ 
        in the constraint $(\eta\sim 0)$, which was replaced. 
        For the formula $\phi'$ (obtained by the end of execution of the \textbf{while} loop), 
        we see that $\parfunc(\phi') \leq 3$. 
        \smallbreak
        Lines~\ref{line:elimbounded:foreach-constraint}--\ref{line:elimbounded:replace-tau} introduce at most 
        $\paratom(\phi')$ new variables $\vec z'$, and for every $z'\in\vec z'$, the term $S(z')$ has at most 
        $3$ occurrences of $(\cdot \bmod f(t))$ and $\floor{\frac{\cdot}{t^d}}$. 
        The variables $\vec z'$ will occur in linear constraints with integer 
        coefficients $\psi'(\vec w', \vec z')$ produced by 
        the non-deterministic procedure \BoundedQE in line~\ref{line:elimbounded:gauss}. 
        Therefore, after the replacement in $\psi'$ all the variables $z'\in\vec z'$ with $S(z')$ 
        in line~\ref{line:elimbounded:return}, we obtain the bounds 
        \[
            \parfunc(\psi) \leq \parfunc(\phi')\cdot \paratom(\phi') \leq 2^7 \cdot a \cdot n \cdot c \cdot \bitlength{B}. 
        \] 

        \item[Bound on {\rm$\parconst(\psi)$}:] The value of $\parconst$ for the formulae added to $\phi$ in 
        line~\ref{line:elimbounded:add-bounds} is clearly bounded by $\bitlength{B}$. After the replacement of $w$ in line~\ref{line:elimbounded:div-x}, 
        the coefficients of the variables $\vec y$ are at most 
        $\max\{\bitlength{f(t)\cdot t^{\bitlength{B}}} : f(t)\text{ is a coefficient of }w\in\vec w\}$. 
        Before we continue, let us observe that 
        \begin{remark}\label{remark:bitlength-mul}
            For any two polynomials $p_1,p_2 \in \Z[t]$ we have $\bitlength{p_1 \cdot p_2} \leq \bitlength{p_1} \cdot \bitlength{p_2}$.
        \end{remark}
        This follows from the following sequence of inequalities:
        \begin{align*}
            \bitlength{p_1 \cdot p_2} &\leq (\deg(p_1) + \deg(p_2) + 1) \cdot (\ceil{\log_2(h(p_1) \cdot h(p_2)+1)}+1)\\ 
            &\leq (\deg(p_1) + \deg(p_2) + 1) \cdot (\ceil{\log_2(h(p_1)+1)} + \ceil{\log_2(h(p_2)+1)}+1)\\
            &\leq (\deg(p_1)+1) \cdot (\deg (p_2) + 1) \cdot (\ceil{\log_2(h(p_1)+1)}+1) \cdot  (\ceil{\log_2(h(p_2)+1)}+1)\\
            &\leq \bitlength{p_1} \cdot \bitlength{p_2}.
        \end{align*}
        Therefore, we have 
        \[
            \parconst(\phi_\varnothing) 
            \leq \max\{c, \bitlength{B}, c\cdot \bitlength{t^{\bitlength{B}}}, \bitlength{t-1}\}
            \leq c\cdot(\bitlength{B}+1)\cdot 2 \leq 4 \cdot c\cdot\bitlength{B},
        \]
        where, as before, $\phi_\varnothing$ 
        is the formula obtainded from $\phi$ after the execution of the \textbf{foreach} loop in 
        line~\ref{line:elimbounded:foreach-div-x}. Let us now consider the constants of the 
        formula $\phi'$. 
        \smallbreak
        It is sufficient to estimate the size of the polynomials in the 
        constraints introduced for an inequality $(\eta\leq 0)$, which 
        is affected by line~\ref{line:elimbounded:replace-in-phi} at most $K\coloneqq\bitlength{B}+c$ times. 
        Since the term $\tau(\vec z)$ is non-shifted, at every $(k+1)$-th iteration of the \textbf{while} loop 
        the constant term of $\eta$ goes to the linear term with integer coefficients $\rho_{k+1}(\vec y)$, 
        which is the sum of the coefficient $\rho_{k}(\vec y)$ of $t^{k}$ in $\sigma(\vec y)$, 
        and of the integer $r_{k}$ guessed in line~\ref{line:elimbounded:guess-r} during the $k$-th iteration 
        (see line~\ref{line:elimbounded:replace-in-phi}). Here, we assume that $r_0=0$.  
        It is clear that $\onenorm{\rho_{k+1}}\leq\onenorm{\rho'_k} + \abs{r_k}$ and, moreover,
        that \emph{for every $k\in[0,K-1]$} we have 
        \[
            \bitlength{\onenorm{\rho'_k}} \,\leq\, \bitlength{(\#\vec y + 1)} \cdot \parconst(\phi_\varnothing) 
            \,\leq\, n\cdot(\bitlength{B}+1) \cdot 4 \cdot c\cdot\bitlength{B}
            \,\leq\, 2^3 \cdot n \cdot c \cdot \bitlength{B}^2.
        \]    
        The integer $r_k$ is guessed in line~\ref{line:elimbounded:guess-r} from the segment 
        $[-\onenorm{\rho_{k}}-1,\onenorm{\rho_{k}}+1]$ (see the previous 
        line~\ref{line:elimbounded:rho} for the '$+1$'). Therefore, by the end of the execution of 
        the loop, we obtain 
        \begin{align*}
            \bitlength{\onenorm{\rho_K}} &\,\leq\, \bitlength{\onenorm{\rho'_{K-1}} + \abs{r_{K-1}}} \\
                             &\,\leq\, \bitlength{\onenorm{\rho'_{K-1}} + \onenorm{\rho_{K-1}} + 1}\\
                             &\,\leq\, \bitlength{\onenorm{\rho'_{K-1}} + \onenorm{\rho'_{K-2}} + \abs{r_{K-2}} + 1} \leq \dots\\
                             &\,\leq\, \bitlength{K\cdot\max\{\onenorm{\rho'_{k}} : k\in[0,K-1]\} + \abs{r_0} + (K-1)}\\
                             &\,\leq\, (\bitlength{B}+c) \cdot (2^3 \cdot n \cdot c \cdot \bitlength{B}^2 + 1) 
                              \,\leq\, 2^5 \cdot n \cdot c^2 \cdot \bitlength{B}^3.
        \end{align*}
        We can now upper-bound $\parconst(\phi')$, where, again, $\phi'$ is the formula obtained by the 
        end of execution of the loop. The linear constraints introduced by lines~\ref{line:elimbounded:add-eq-to-psi} 
        and~\ref{line:elimbounded:add-neq-to-psi} have coefficients bounded by $\onenorm{\rho_K} + 2$; 
        the coefficients of the term $\sigma$ are now integers bounded the same as $\onenorm{\rho'_k}$ 
        for any $k\in[0,K]$, and thus are bounded by $\onenorm{\rho_K}$. 
        The coefficients of the variables $\vec z$ in $\tau(\vec z)$ does not change, i.e., their bit size 
        is at most $\parconst(\phi_\varnothing)$; 
        the denominator of the term $\sfloor{\frac{\tau(\vec z)}{t^K}}$ is the polynomial $t^K$ and, 
        by definition, $\bitlength{t^K} = (K + 1) \cdot (1 + 1) \leq 4 \cdot K < \bitlength{\onenorm{\rho_K}}$. 
        Summarising, 
        \[
            \parconst(\phi') \,\leq\,\bitlength{\onenorm{\rho_K} + 2} \leq 2^7 \cdot n \cdot c^2 \cdot \bitlength{B}^3. 
        \]    
        \smallbreak
        Lines~\ref{line:elimbounded:foreach-constraint}--\ref{line:elimbounded:replace-tau} 
        replace the \PPA terms with $\vec z$ with new variables $\vec z'$, which will be restored in the formula 
        in line~\ref{line:elimbounded:return}. For this reason, we first consider the bounds on the bit size of 
        integer coefficients of $\vec z'$ in the output $\exists\vec w' \leq B' : \psi'(\vec w', \vec z')$ of the non-deterministc procedure \BoundedQE. 
        Applying the bounds from \Cref{lemma:param-QE}, we obtain   
        \begin{align*}
            \parconst(\psi') &\,\leq\, (\#\vec y+\paratom(\phi'))^3 \cdot \parconst(\phi')^3     \\
                             &\,\leq\, (n\cdot(\bitlength{B}+1)+2^5 \cdot a \cdot n \cdot c \cdot \bitlength{B})^3 \cdot \parconst(\phi')^3\\
                             &\,\leq\, (2^6 \cdot a \cdot n \cdot c \cdot \bitlength{B} \cdot 2^7 \cdot n \cdot c^2 \cdot \bitlength{B}^3)^3\\
                             &\,\leq\, (2^{13} \cdot a \cdot n^2 \cdot c^3 \cdot \bitlength{B}^4)^3.
        \end{align*}
        The size of the vector of variables $\vec w'$ is equal to $\#\vec y$, and every constraint of $\psi'$ is a 
        linear (in)equality or divisibility with integer coefficients. 
        Denote by $\psi''(\vec z')$ the result of the replacements performed by 
        lines~\ref{line:elimbounded:foreach-w-to-int}--\ref{line:elimbounded:replace-w}.
        Using the bounds from \Cref{lemma:param-QE} on the bit size of $B'(w')$ for the variables 
        $w'\in\vec w'$, the constants of $\psi''$ are such that 
        \begin{align*}
            \parconst(\psi'') &\,\leq\, (\#\vec y+1) \cdot \parconst(\psi') \cdot \big((\#\vec y+\paratom(\phi'))^{11} 
                                        \cdot \parconst(\phi')^9\big) \\
                              &\,\leq\, 2\cdot\parconst(\psi')\cdot \#\vec y \cdot (\#\vec y+\paratom(\phi'))^{11} 
                                        \cdot \parconst(\phi')^9 \\
                              &\,\leq\, 2\cdot\parconst(\psi') \cdot (2\cdot n \cdot \bitlength{B})^{12} 
                                        \cdot (2^5 \cdot a\cdot n \cdot c \cdot \bitlength{B})^{11} 
                                        \cdot (2^7 \cdot n \cdot c^2 \cdot \bitlength{B}^3)^9 \\
                              &\,=\, \parconst(\psi') \cdot 2^{131} \cdot a^{11} \cdot n^{32} \cdot c^{29} \cdot \bitlength{B}^{50} \\
                              &\,\leq\,2^{39} \cdot(a \cdot n^2 \cdot c^3 \cdot \bitlength{B}^4)^3\cdot 2^{131}\cdot a^{11} \cdot n^{32} \cdot c^{29} \cdot \bitlength{B}^{50} \\
                              &\,=\,2^{170} \cdot a^{14} \cdot n^{38} \cdot c^{38} \cdot \bitlength{B}^{62}\\
                              &\,<\,(2^{11}\cdot a \cdot n^3 \cdot c^3\cdot \bitlength{B}^4)^{16}.
        \end{align*}
        It remains to observe that the replacements of line~\ref{line:elimbounded:return} transforms 
        $\psi''(\vec z')$ into $\psi(\vec z)$, where the constants have the desired bounds: 
        \begin{align*}
            \parconst(\psi) &\,\leq\, \max\{\parconst(\psi''), \parconst(\psi')\cdot\parconst(\phi')\} \\
                            &\,\leq\, \max\{\parconst(\psi''), (2^{13} \cdot a \cdot n^2 \cdot c^3 \cdot \bitlength{B}^4)^3\cdot 2^7 \cdot n \cdot c^2 \cdot \bitlength{B}^3\}\\
                            &\,=\,\parconst(\psi'')\,<\,(2^{11}\cdot a \cdot n^3 \cdot c^3\cdot \bitlength{B}^4)^{16}.\qedhere
        \end{align*}
    \end{description}   
\end{proof}

\section{Complexity of \BoundedElimDiv and \Master}
\label{app:complexity-master}

To prove \Cref{{lemma:small-formula}}, which describes the complexity of \Cref{algo:master} (\Master), 
we first formally prove correctness of \Cref{algo:elimdiv} (\BoundedElimDiv) together with an analog 
of \Cref{lemma:param-QE,lemma:param-elimbounded} for this algorithm 
(see \Cref{lemma:param-boundedelimdiv} below). Combining bounds from 
\Cref{lemma:param-QE,lemma:param-elimbounded,lemma:param-boundedelimdiv}, we obtain 
the desired result.   

\subsection{Correctness and complexity of \BoundedElimDiv}

An informal description of the procedure \BoundedElimDiv is given in \Cref{sec:FO-procedure}: 
every linear divisibility in the input formula is transformed into an equality by means of a bounded 
existential quantifier. Below we formalize the correctness of the transformation:

\LemmaBounedElimDivCorrect*

\begin{proof}
    It is sufficient to prove that the divisibility $f(t) \divides \sigma(\vec w) + \tau(\vec z)$, 
    which is picked from the input formula $\psi$, is equivalent to 
    either $\exists y : f(t) \cdot y + \sigma(\vec w) + (\tau\bmod{f(t)}) = 0$ or 
    $\exists y : (-f(t) \cdot y + \sigma(\vec w) + (\tau\bmod{f(t)}) = 0)$ for a variable $y\in[0, t^d]$, 
    where $d$ is a positive integer defined in line~\ref{algo:bounded-elimi-div:max-degree}.  

    The direction from right to left is obvious. 
    For the converse direction, notice that if there is no upper bound on the value of 
    the non-negative integer variable $y$, then the replacement directly follows from 
    the definition of the integer divisibility predicate. 
    The sign of the coefficent of $y$ corresponds to the sign of the integer 
    $\frac{\sigma + (\tau\bmod{f(t)})}{f(t)}$: when it is positive, 
    line~\ref{algo:bounded-elimi-div:guess-sign} should take '$-$', and vice versa.  
    Let us show that the bitlength of this integer 
    is less than $\bitlength{t^d}=2\cdot(d+1)$, and this will justify the bound $y\in[0, t^d]$ 
    introduced in line~\ref{algo:bounded-elimi-div:bounded-quantifier}. 
    Recall that we represent the linear \PPA term $\sigma$ as 
    $f_0(t) + \sum_{i=1}^n f_i(t) \cdot w_i$, where $\vec w = (w_1,\dots,w_n)$ is a vector of variables bounded 
    by the map $B$, and $d$ is defined as 
    \[
        d = (n + 3) \cdot \max\{\bitlength{f}, \bitlength{f_0}, \bitlength{f_i} \cdot \bitlength{B(w_i)}\, : i \in [1..n]\}.
    \] 
    Using \Cref{remark:bitlength-mul}, we obtain
    \begin{align*}
        \Big\langle\frac{\sigma + (\tau\bmod{f(t)})}{f(t)}\Big\rangle 
                &\,\leq\,\bitlength{\sigma + (\tau\bmod{f(t)})} 
                 \,\leq\,\bitlength{f_0(t) + \sum_{i=1}^n f_i(t) \cdot B(w_i) + f(t)} \\
                &\,\leq\,\bitlength{n+2}\cdot\max\{\bitlength{f}, \bitlength{f_0}, \bitlength{f_i\cdot B(w_i)}\, : i \in [1..n]\} \\
                &\,\leq\,(\ceil{\log_2(n+3)} + 1)\cdot\max\{\bitlength{f}, \bitlength{f_0}, \bitlength{f_i} \cdot \bitlength{B(w_i)}\, : i \in [1..n]\} \\
                &\,\leq\,d < \bitlength{t^d}. \qedhere
    \end{align*}
\end{proof}

We now move to the complexity of the procedure.

\begin{lemma}\label{lemma:param-boundedelimdiv}
    Let $\exists\vec w \leq B :\psi(\vec w, \vec z)$ be a formula in input of~\BoundedElimDiv, 
    where $\vec w = (x_1,\dots,x_n)$, $\vec z = (x_{n+1},\dots,x_{g})$ with $n \geq 1$ and 
    $\psi$ is a positive Boolean combination of linear constraints.
    Then, for every branch output $\exists \vec w' \leq B' : \psi'(\vec w',\vec z)$ 
    of the algorithm, the following holds:
    \begin{equation*}
        \text{if }\; 
        \begin{cases}
            \paratom(\psi) &\!\leq a\\ 
            \parvars(\psi) &\!= g\\ 
            \parfunc(\psi) &\!=0\\ 
            \parconst(\psi) &\!\leq c
        \end{cases}
        \quad\text{ then }\;
        \begin{cases}
            \paratom(\psi')  &\!\leq a\\ 
            \parvars(\psi')  &\!\leq g+a\\ 
            \parfunc(\psi')  &\!\leq 1\\ 
            \parconst(\psi') &\!\leq c
        \end{cases}
    \end{equation*}
    Moreover, the vector $\vec w'$ contains $\parvars(\psi')+n-g$ variables, and 
    $\bitlength{B'} \leq 2\cdot(n+4) \cdot c \cdot \bitlength{B}$.
\end{lemma}

\begin{proof}
    The bounds on the parameters $\paratom$, $\parvars$, $\parfunc$, and $\parconst$ are trivial. 
    They are obtained by inspecting the lines of the pseudocode of \Cref{algo:elimdiv}. 
    The number of constraints does not change: line~\ref{algo:bounded-elimi-div:update-psi} 
    replaces a linear divisibility with an equality with a single occurrence of the function $(\cdot \bmod f(t))$. 
    Also notice that the polynomial coefficients occurring in this new equality are already present 
    in the input formula $\psi$, and thus $\parconst(\psi)=\parconst(\psi')$.  
    Finally, since the \textbf{foreach} loop has at most $\paratom(\psi)$ iterations, 
    line~\ref{algo:bounded-elimi-div:bounded-quantifier} introduces at most $a$ new variables, each of which 
    is added to the list of bounded variables $\vec w'$. It is clear that 
    $\#\vec w'=(\parvars(\psi')-\#\vec z) \,=\, \parvars(\psi')+n-g$.
    
    To upper-bound the bitlength of $B'(w')$ for the newly introduced variables $w'\in\vec w'$, 
    consider lines~\ref{algo:bounded-elimi-div:max-degree} and~\ref{algo:bounded-elimi-div:bounded-quantifier}. 
    Let $\bitlength{B'}\coloneqq\max\{\bitlength{B'(w')} : w'\in\vec w'\}$, then 
    we see that 
    \[
        \bitlength{B'}\leq \max\{\bitlength{B}, \bitlength{t^{(n+3)\cdot c\cdot \bitlength{B}}}\} 
                      \leq \max\{\bitlength{B}, 2\cdot(n+4)\cdot c \cdot \bitlength{B}\} = 2\cdot(n+4)\cdot c \cdot \bitlength{B}.   \qedhere 
    \]    
\end{proof}

\subsection{Complexity of \Master}

The main \Cref{algo:master} (\Master) calls \Cref{algo:gaussianqe,algo:elimdiv,algo:qe}, and the bounds from 
\Cref{lemma:param-elimbounded,lemma:param-QE,lemma:param-boundedelimdiv} will give us the upper 
bounds on the values of the four parameters from these lemmas. 
The next lemma essentially proves \Cref{{lemma:small-formula}} from \Cref{sec:complexity}.


\begin{lemma}\label{lemma:param-master}
    Let $\exists \vec x :\phi(\vec x, \vec z)$ be a formula in input of~\Master, 
    where $\phi$ is a positive Boolean combination of $\PPA$ constraints. 
    Let $\vec x = (x_1,\dots,x_n)$ and $\vec z = (x_{n+1},\dots,x_{g})$, where $n \geq 1$,
    then for every branch output $\psi(\vec z)$ of the algorithm, the following holds:
    \begin{equation*}
        \text{if }\; 
        \begin{cases}
            \paratom(\psi) &\!\leq a\\ 
            \parfunc(\psi) &\!\leq d\\ 
            \parconst(\psi) &\!\leq c
        \end{cases}
        \quad\text{ then }\;
        \begin{cases}
            \paratom(\psi)  &\!\leq 2^5\cdot(n + d + 4) \cdot (a+n+3\cdot d)^{19} \cdot c^{15}\\ 
            \parfunc(\psi)  &\!\leq 2^5\cdot(n + d + 4) \cdot (a+n+3\cdot d)^{19} \cdot c^{15}\\ 
            \parconst(\psi) &\!\leq \big(2^{6}\cdot(n + d + 4)\cdot (a+n+3\cdot d)^{18} \cdot c^{15}\big)^{64}.
        \end{cases}
    \end{equation*}
\end{lemma}

\begin{proof}
    The two \textbf{while} loops in the pre-processing part remove all functions from the input formula. 
    Denote by $\exists \vec w_0 : \phi_0(\vec w_0, \vec z)$ the result of this step. 
    These loops together perform $d$ iterations, and add at most $2$ new constraints to the formula 
    (see line~\ref{line:master:replace-frac}, 
    line~\ref{line:master:replace-mod-zero} and line~\ref{line:master:replace-mod}). 
    Therefore, $\paratom(\phi_0)\leq a+2\cdot d$.
    Next, these loops also introduce at most one variable during each iteration: 
    $\parvars(\phi_0)\leq g + d$.
    It is easy to see that $\parfunc(\phi_0)=0$ and $\parconst(\phi_0)=\parconst(\phi)$. 
    Finally, the bounds introduced by line~\ref{line:master:bounded-quantifier} (denote the corresponding map by $B_0$)
    can be bounded as follows: $\bitlength{B_0}\leq2\cdot c$. 
    
    We are now going to apply the bounds constructed by \Cref{lemma:param-elimbounded,lemma:param-QE,lemma:param-boundedelimdiv} 
    sequentially to the quintuple 
    $(\paratom(\phi_0), \parvars(\phi_0), \parfunc(\phi_0), \parconst(\phi_0), \bitlength{B_0}) 
    = (a+2\cdot d,\, g + d,\, 0,\, c,\, 2\cdot c)$. Moreover, we know that $\#\vec w_0\leq n+d$.          
    \begin{description}
        \item[\BoundedQE:] 
        Applying \Cref{lemma:param-elimbounded}, we obtain that every branch output of 
        \BoundedQE on input $\phi_0$ will have at most $2\cdot (n+d+a+2\cdot d) + 1\leq 2\cdot (a+n+3\cdot d) +1$ 
        atomic formulae. It is now easy to obtain the bounds on parameters for every formula 
        $\exists\vec w_1 \leq B_1 : \phi_1(\vec w_1, \vec z)$ given on input of 
        algorithm \BoundedElimDiv: 
        \begin{equation*}
            \begin{cases}
                \paratom(\phi_1)          &\!\leq 2\cdot (a+n+3\cdot d)+1\\ 
                \parvars(\phi_1)          &\!\leq g + d = (n + d) + \#\vec z\\ 
                \parfunc(\phi_1)          &\!\leq 0\\ 
                \parconst(\phi_1)         &\!\leq (a+n+3\cdot d)^3 \cdot c^3\\
                \bitlength{B_1}           &\!\leq (a+n+3\cdot d)^{11} \cdot c^9,
            \end{cases}
        \end{equation*}
        because the bound on $\bitlength{B_1}$, which is the union of $B_0$ and the map $B$ from the 
        output of \BoundedQE is such that $\bitlength{B_1}\leq\max\{ 
        (a+n+3\cdot d)^{11} \cdot c^9, \bitlength{B_0}\} = (a+n+3\cdot d)^{11} \cdot c^9$.
        \Cref{lemma:param-elimbounded} also says that $\#\vec w_1 = \#\vec w_0 \leq n+d$. 
        \item[\BoundedElimDiv:] According to \Cref{lemma:param-boundedelimdiv}, 
        algorithm \ElimBounded will get a bounded formula $\exists\vec w_2 \leq B_2 : \phi_2(\vec w_2, \vec z)$ 
        with the same values of the parameters $\paratom$ and $\parconst$ as in formula $\phi_1$. 
        It is easy to derive the bounds 
        for $\parvars(\phi_2)\leq \parvars(\phi_1) + \paratom(\phi_1)
        \leq g + d + 2\cdot (a+n+3\cdot d)+1 \leq 3\cdot (a+n+3\cdot d)+\#\vec z$ 
        and $\parfunc(\phi_2) \leq 1$. The bit size of the bounds is such that 
        \begin{align*}
            \bitlength{B_2}  &\,\leq\, 2\cdot(n + d + 4) \cdot (a+n+3\cdot d)^3 \cdot c^3 \cdot (a+n+3\cdot d)^{11} \cdot c^9 \\
                             &\,\leq\, 2\cdot(n + d + 4) \cdot (a+n+3\cdot d)^{14} \cdot c^{12}.
        \end{align*}
        To sum up, $\#\vec w_2\leq 3\cdot (a+n+3\cdot d)$ and we have   
        \begin{equation*}
            \begin{cases}
                \paratom(\phi_2)          &\!\leq 2\cdot (a+n+3\cdot d)+1\\ 
                \parvars(\phi_2)          &\!\leq 3\cdot (a+n+3\cdot d)+\#\vec z\\ 
                \parfunc(\phi_2)          &\!\leq 1\\ 
                \parconst(\phi_2)         &\!\leq (a+n+3\cdot d)^3 \cdot c^3\\
                \bitlength{B_2}     &\!\leq 2\cdot(n + d + 4) \cdot (a+n+3\cdot d)^{14} \cdot c^{12}.
            \end{cases}
        \end{equation*}
        \item[\ElimBounded:] 
        We are now going to apply \Cref{lemma:param-elimbounded} to the formula 
        $\exists\vec w_2 \leq B_2 : \phi_2(\vec w_2, \vec z)$.
        \begin{description}
            \item[Bound on {\rm$\paratom(\psi)$}:]
            \begin{align*}
                \paratom(\psi) &\,\leq\, \paratom(\phi_2) \cdot 3\cdot (a+n+3\cdot d) 
                \cdot (a+n+3\cdot d)^3 \cdot c^3 \cdot \bitlength{B_2} \\
                &\,\leq\, (2\cdot (a+n+3\cdot d)+1) \cdot 3\cdot (a+n+3\cdot d)^4 \cdot c^3 \cdot \bitlength{B_2} \\
                &\,\leq\, 9\cdot (a+n+3\cdot d)^5 \cdot c^3 \cdot \bitlength{B_2} \\
                &\,\leq\, 9\cdot (a+n+3\cdot d)^5 \cdot c^3 \cdot 2\cdot(n + d + 4) \cdot (a+n+3\cdot d)^{14} \cdot c^{12} \\
                &\,\leq\, 2^5\cdot(n + d + 4) \cdot (a+n+3\cdot d)^{19} \cdot c^{15}.
            \end{align*} 
            \item[Bound on {\rm$\parfunc(\psi)$}:] the same as $\paratom(\psi)$ above.
            \item[Bound on {\rm$\parconst(\psi)$}:] 
            \begin{align*}
                \parfunc(\psi) &\,\leq\, \big(2^{11}\cdot \paratom(\phi_2) \cdot (3 \cdot (a + n + 3 \cdot d)^4 \cdot c^3)^3 \cdot \bitlength{B_2}^4\big)^{16} \\
                &\,\leq\, \big(2^{11}\cdot(2\cdot (a+n+3\cdot d)+1) \cdot (3 \cdot (a + n + 3 \cdot d)^4 \cdot c^3)^3 \cdot \bitlength{B_2}^4\big)^{16} \\
                &\,\leq\, \big(2^{18}\cdot (a+n+3\cdot d)^{13} \cdot c^9 \cdot \bitlength{B_2}^4\big)^{16} \\
                &\,\leq\, \big(2^{18}\cdot (a+n+3\cdot d)^{13} \cdot c^9 \cdot 2^4\cdot(n + d + 4)^4 \cdot (a+n+3\cdot d)^{56} \cdot c^{48}\big)^{16} \\
                &\,=\, \big(2^{22}\cdot(n + d + 4)^4\cdot (a+n+3\cdot d)^{69} \cdot c^{57}\big)^{16} \\
                &\,\leq\, \big(2^{6}\cdot(n + d + 4)\cdot (a+n+3\cdot d)^{18} \cdot c^{15}\big)^{64}. \qedhere
            \end{align*} 
        \end{description}
    \end{description}
\end{proof}

\begin{proof}[Proof of \Cref{{lemma:small-formula}}]
    As for~\Cref{lemma:boundedqe-in-np}, 
    following the proof of~\Cref{lemma:param-master}, 
    we now know that all object constructed by~\Master during its execution 
    are of polynomial bit size. It is then simple to deduce that 
    the algorithm runs in non-deterministic polynomial time; 
    it suffices to check the number of iterations of the various loops, and the sets used for the various guesses performed by the algorithm.

    In Lines~\ref{line:master:while:frac}--\ref{line:master:replace-mod}, 
    the number of iterations in both \textbf{while} loops is bounded 
    by the size of the input formula $\phi$, 
    and guesses (lines~\ref{line:master:first-nondet-branch} and~\ref{line:master:mod-sign}) 
    are all on a $2$-element domain (for the non-deterministic branches).

    We already know from~\Cref{lemma:boundedqe-in-np} 
    that the computation performed by~\BoundedQE takes non-deterministic polynomial time. Since its output has size polynomial in the input formula~$\phi$, 
    the \textbf{foreach} loop in algorithm~\BoundedElimDiv also 
    performs only polynomially many iterations; all guesses are again on a $2$-element domain (line~\ref{algo:bounded-elimi-div:guess-sign}). 

    Going into~\ElimBounded, the loop in line~\ref{line:elimbounded:foreach-div-x}
    only executes $\card{\vec w}$ times, adding at most $\card{\vec w} \cdot \bitlength{B}$ variables~$\vec y$ overall (a polynomial amount). 
    Let $\phi_{\varnothing}$ be the formula obtained after these lines.
    The $\vec y$-degree $\deg(\vec y, \phi_\varnothing)$ is polynomial in the input formula $\phi$. 
    By~\Cref{lemma:elimbounded:Dtwo}, 
    the \textbf{while} loop of line~\ref{line:elimbounded:while}
    only iterates a polynomial amount of times. 
    Following~\Cref{lemma:param-elimbounded}, the guesses 
    performed in line~\ref{line:elimbounded:guess-r} are over an interval of numbers having all polynomial bit size. 
    After the \textbf{foreach} loop of line~\ref{line:elimbounded:foreach-constraint} (whose number of iterations is, 
    again, polynomial), the algorithm calls again~\BoundedQE, 
    which returns an output in (non-deterministic) polynomial time. 
    In this output, all elements in the range of the map $B'$ are integers of polynomial bit size. 
    Then, in the \textbf{foreach} loop of line~\ref{line:elimbounded:foreach-w-to-int} (executing polynomially many times) 
    all guesses are over a range of polynomial-size integers.
 \end{proof}
 


\end{document}